\documentclass[11pt, a4paper, reqno]{amsart}
\usepackage{color}
\usepackage{amsmath}
\usepackage{amsthm}
\usepackage{amssymb}
\usepackage{amsmath}
\usepackage{mathcomp}
\usepackage{dsfont}
\usepackage[toc,page]{appendix}
\usepackage{graphicx} 
\usepackage{subfigure}
\usepackage{enumerate}
\usepackage{amsfonts}
\usepackage{amsthm}
\usepackage{amsmath}
\usepackage{times}
\usepackage{amscd}
\usepackage[latin2]{inputenc}
\usepackage{t1enc}
\usepackage{enumitem}
\usepackage[mathscr]{eucal}
\usepackage{indentfirst}
\usepackage{graphicx}
\usepackage{graphics}
\usepackage{pict2e}
\usepackage{epic}
\usepackage{esint}
\usepackage[margin=2.9cm]{geometry}
\usepackage{epstopdf} 
\usepackage{verbatim}
\usepackage{cancel}
\usepackage{amssymb}
\usepackage{xcolor}
\usepackage{dsfont}
\usepackage{hyperref}
\usepackage{physics}
\usepackage{amsmath}

\usepackage[margin=2.9cm]{geometry}

\definecolor{Greenish}{RGB}{34, 139 , 34}
\definecolor{Blueish}{RGB}{39,64, 139}

\usepackage{hyperref}
\hypersetup{
    colorlinks=true,
    linkcolor=Greenish,
    citecolor=Blueish,
    }

\setitemize{itemsep=5pt, topsep=7pt}
\numberwithin{equation}{section}

\newtheorem{definition}{Definition}[section]
\newtheorem{lemma}[definition]{Lemma}

\newtheorem{proposition}[definition]{Proposition}

\newtheorem{theorem}[definition]{Theorem}

\newtheorem{corollary}[definition]{Corollary}

\numberwithin{equation}{section}

\def\bR{\mathbb{R}}

\def\cN{\mathcal{N}}

\def\cF{\mathcal{F}}
\def\cM{\mathcal{M}}

\def\cG{\mathcal{G}}

\def\rd{\mathrm{d}}

\def\tr{\mathrm{tr}}

\def\ad{\mathrm{ad}}

\def\N{\mathcal{N}_+}

\def\b1{\mathds{1}}
\def\tr{\mathrm{tr}}
\def\Re{\mathrm{Re}}
\def\Im{\mathrm{Im}}

\def\wO{\widetilde{O}}

\def\ph{\varphi}

\newcommand{\vertiii}[1]{{\left\vert\kern-0.25ex\left\vert\kern-0.25ex\left\vert #1 
    \right\vert\kern-0.25ex\right\vert\kern-0.25ex\right\vert}}

\title{Large deviations for the ground state of weakly interacting Bose gases}

\begin{document}

\author{Simone Rademacher}
\address{Department of Mathematics, LMU Munich, Theresienstrasse 39, 80333 Munich}

\date{\today}

\begin{abstract}
We consider the ground state of a Bose gas of $N$ particles on the three-dimensional unit torus in the mean-field regime that is known to exhibit Bose-Einstein condensation. Bounded one-particle operators with law given through the interacting Bose gas' ground state correspond to dependent random variables due to the bosons' correlation. We prove that in the limit $N \rightarrow
 \infty$ bounded one-particle operators with law given by the ground state satisfy large deviation estimates. We derive a lower and an upper bound on the rate function that match up to second order and that are characterized by quantum fluctuations around the condensate. 


\end{abstract}
\maketitle

\section{Introduction}

We consider $N$ bosons on the three dimensional unit torus $\Lambda = [0,1]^3$ in the mean-field regime described by the Hamiltonian 
 \begin{align}
 \label{def:HN}
H_N = \sum_{j=1}^N (- \Delta_{x_j} ) + \frac{1}{N} \sum_{i<j}^N v( x_i - x_j) 
\end{align}
acting on $L_s^2 \left(\Lambda^N \right)$, the symmetric subspace of $L^2\left( \Lambda^N \right) $. We consider two-particles interaction potentials with Fourier transform $\widehat{v}$ given by 
\begin{align}
\label{eq:ass-v}
v(x) = \sum_{p \in \Lambda^*} \widehat{v} (p ) e^{i p  \cdot x} \quad \text{for} \quad \Lambda^* = 2 \pi \mathbb{Z}^3 \quad  \text{ with} \quad \widehat{v} \geq 0, \; \widehat{v} \in \ell^1 ( \Lambda^*) \; . 
\end{align}
At zero temperature the bosons relax to the unique ground state $\psi_N$ of $H_N$ realizing 
 \begin{align}
 \label{def:eN}
E_N = \inf_{ \| \psi \|_2=1} \langle \psi, \; H_N \psi \rangle = \langle \psi_N, H_N \psi_N \rangle \; . 
 \end{align}
The ground state $\psi_N$ exhibit Bose-Einstein condensation, i.e. a macroscopic fraction of the $N$ particles occupies the same quantum state, called the condensate. Mathematically, $\psi_N$ is said to satisfy the property of Bose-Einstein condensation if its corresponding one-particle reduced density given by 
\begin{align}
\gamma_{\psi_N}^{(k)} := \tr_{k+1 , \dots, N} \vert \psi_N \rangle \langle \psi_N \vert 
\end{align}
for $k=1$ converges in trace norm to 
\begin{align}
\gamma_{\psi_N}^{(1)} \rightarrow \vert \varphi_0 \rangle \langle \varphi_0 \vert \quad \text{as} \quad N \rightarrow \infty \; \label{eq:BEC}
\end{align}
where $\varphi_0 \in L^2( \Lambda)$ denotes the condensate wave function. In fact the convergence \eqref{eq:BEC} holds true not only for the one- but also in for general $k$-particle reduced density densities. However due to particle's correlation, the ground state $\psi_N$ is not a purely factorized state of the condensate's wave function.  

Both, the computation of the ground state energy \eqref{def:eN} and the ground state's property of BEC \eqref{eq:BEC} and beyond are widely studied in the literature (see for example \cite{BPS,DN,GS,LNSS,N,NN,NS,R,YY}). In fact, \cite{R} proves besides BEC that the Bose gas' excitation spectrum is well described by Bogoliubov theory. Consequently, the fluctuations around the condensate, namely the particles orthogonal to the condensate can be effectively described as a quasi-free state (namely a Gaussian quantum state) on an appropriate Fock space. This characterization of the condensates' excitations will be important for our analysis. 

\subsection{Probabilistic approach} Recently the characterization of Bose-Einstein condensation through probabilistic concepts became of interest. In fact, the property of Bose-Einstein condensation \eqref{eq:BEC} implies a law of large numbers for bounded one particle operators \cite{BKS}. To be more precise, let $O$ denote a bounded one particle operator on $L^2( \mathbb{R}^3)$ for which we define the $N$-particle operator $O^{(i)}$ by 
\begin{align}
\label{def:Oj}
O^{(i)} = \mathds{1} \otimes \mathds{1} \otimes \cdots \otimes \mathds{1} \otimes O \otimes \mathds{1} \otimes \cdots \otimes \mathds{1} 
\end{align}
i.e. as operator acting as $O$ on the $i$-th particle and as identity elsewhere. We consider $O^{(i)}$ as a random variable with law given by 
\begin{align}
\mathbb{P}_{\psi} \left[ O^{(i)} \in A \right] = \langle \psi, \chi_A ( O^{(i)} ) \psi \rangle \quad \text{with} \quad \psi \in L^2 ( \Lambda^N )  
\end{align}
where $\chi_A$ denotes the characteristic function of $A \subset R$. We remark that factorized states lead to i.i.d. random variables in this picture \cite{R_dp} and thus a law of large numbers and a large deviation principle hold true from basic theorems of probability theory. 

Random variables with law given by the ground state $\psi_N$ (known to not be a factorized state) of $H_N$, satisfy a law of large numbers, too (though they are not independent random variables).  To be more precise, the averaged centered (w.r.t. to the condensate's expectation value $\langle \varphi_0, O \varphi_0 \rangle$) sum 
\begin{align}
O_N := \frac{1}{N} \sum_{i=1}^N  \left( O^{(i)} - \langle \varphi_0, O \varphi_0 \rangle \right) 
\end{align}
with $O^{(i)}$ given by \eqref{def:Oj} satisfies for any $\delta >0$
\begin{align}
\label{eq:lln}
\lim_{N \rightarrow \infty} \mathbb{P}_{\psi_N} \left[\vert  O_N \vert > \delta \right] = 0   \; . 
\end{align}
The law of large numbers, in fact, is a consequence of the property of Bose-Einstein condensation \cite{BKS,R_dp} namely of the trace norm convergence of the one- and two-particle reduced density matrix. 

Here, we are interested in the precise decay of the probability distribution \eqref{eq:lln} in probability theory described through the rate function 
\begin{align}
\Lambda_{\psi_N}^* (x) = \lim_{N \rightarrow \infty} N^{-1} \log \mathbb{P}_{\psi_N} \left[ O_N > x \right] 
\end{align}
if the limit exists. In case of i.i.d. random variables (i.e. factorized states $\varphi^{\otimes N}$) Cramer's theorem shows that the rate functions exists and is given in terms of the Legendre-Fenchel transform through 
\begin{align}
\Lambda^*_{\varphi^{\otimes N}} (x) = \inf_{\lambda \in \mathbb{R}} \left[ - \lambda x + \Lambda_{\varphi^{\otimes N}} (\lambda) \right] 
\end{align}
where $\Lambda_{\varphi^{\otimes N}}( \lambda)$ equals for i.i.d. random variables the logarithmic moment generating function 
\begin{align}
\Lambda_{\varphi^{\otimes N}}( \lambda) = \log \langle \varphi, \; e^{\lambda ( O^{(1)} - \langle \varphi, O \varphi \rangle )} \varphi \rangle \; . 
\end{align}
In our main theorem we show that for the ground state $\psi_N$ of $H_N$, known to be not factorized due to particles' correlation, still large deviation estimates hold true. 

\subsection{Results} Before stating our main theorem, we introduce some more notation. In our result we consider operators $O$ such that the norm 
\begin{align}
\vertiii{O} : = \| (1 - \Delta ) O (1- \Delta)^{-1} \| 
\end{align}
is bounded. Furthermore, we define 
\begin{align}
\Lambda_+^* = 2 \pi \mathbb{Z}^3 \setminus \lbrace 0 \rbrace 
\end{align}
and the function $f \in \ell^2 ( \Lambda_+^*)$ by 
\begin{align}
\label{def:f}
f (p) =  \cosh ( \mu_p )  \widehat{qO \varphi_0} (p) + \sinh ( \mu_{p} )  \widehat{qO \varphi_0} (-p)  
\end{align}
where $q$ denotes the projection onto the orthogonal complement of the span of the condensate wave function (i.e. $q= 1- \vert \varphi_0 \rangle \langle \varphi_0 \vert$) and $\mu_p$ is given by the identity 
\begin{align}
\label{def:mu}
\coth ( 2 \mu_p ) = - \frac{p^2 + \widehat{v} (p) }{\widehat{v} (p)} \; . 
\end{align}

\begin{theorem}
\label{thm:ldp}
Let $v$ be a real-valued, even function with $0 \leq \widehat{v} \in \ell^1( \Lambda^*)$, such that $\| \mu \|_{\ell^2( \Lambda_+^*)}$ is sufficiently small. Let $\psi_N$ denote the ground state of the Hamiltonian $H_N$ defined in \eqref{def:HN}. 

Let $O$ denote a self-adjoint operator on $L^2 (  \Lambda)$ such that $\vertiii{O}< \infty$ and let $f$ be defined by \eqref{def:f}. For $O^{(j)}$ given by \eqref{def:Oj}, we define $O_{N} = N^{-1} \sum_{j=1}^N \left( O^{(j)} - \langle \varphi_0, O \varphi_0 \rangle \right)$.  \\

Then, there exist $C_1,C_2 > 0$ (independent of $O$) such that 

\begin{enumerate}
\item[(i)] for all $0 \leq x \leq  1/( C_1 \vertiii{O})$ 
\begin{align}
\limsup_{N \rightarrow \infty} N^{-1} \log \mathbb{P}_{\psi_{N}}  \left[ O_{N} > x\right] 
&\leq -  \frac{x^2} {2 \| f \|_{\ell^2( \Lambda_+^*)}^2} +  x^3 \frac{C_1 \vertiii{O}^3 }{\| f \|_{\ell^2( \Lambda_+^*)}^3}
\end{align}

\item[(ii)] for all $0 \leq x \leq  \| f \|_{\ell^2( \Lambda_+^*)}^4 /(C_2  \vertiii{O}^3)$ 
\end{enumerate}
 \begin{align}
\limsup_{N \rightarrow \infty} N^{-1} \log \mathbb{P}_{\psi_{N}}  \left[ O_{N} > x\right] 
&\geq  - \frac{x^2 }{ 2 \| f \|_{\ell^2( \Lambda_+^*)}^2 }-   x^{5/2} \frac{ C_2  \vertiii{O}^{3/2}  }{ \| f \|_{\ell^2( \Lambda_+^*)}^4} \;  . 
\end{align}
\end{theorem}

We remark that for sufficiently small $x \leq \min \lbrace \| f \|_{\ell^2( \Lambda_+^*)}^4 /(C_2  \vertiii{O}^3), \; 1/ C_1 \vertiii{O} \rbrace $, Theorem \ref{thm:ldp} characterizes the rate function up to second order. Namely Theorem \ref{thm:ldp} shows that in the regime of large deviations, i.e. $x=O(1)$, we have 
\begin{align}
\Lambda_{\psi_N}^* (x) = -  \frac{x^2 }{ 2 \| f \|_{\ell^2( \Lambda_+^*)}^2 } + O( x^{5/2} ) \; . 
\end{align}

\subsubsection*{Regime of large deviations.} The present result in Theorem \ref{thm:ldp} provides a first characterization of the regime of large deviations (i.e. $x= O(1)$) for fluctuations around the condensate of bounded one-particle operators in the ground state. We remark that the variance $\| f \|_{\ell^2( \Lambda_+^*)}$ differs from the variance of factorized state $\varphi_0^{\otimes_s N}$ and is, in particular, fully characterized by the ground state's Bogoliubov approximation (for more details see \eqref{def:Q} and subsequent discussions resp. Lemma \ref{lemma:step3} in Section \ref{sec:proof}) representing the particles' correlation. 

Up to now, results in the regime of large deviations are available for the dynamics in the mean-field regime only. For factorized initial data, the rate function characterizing the fluctuations of bounded one-particles operators around the condensate's Hartree dynamics, were proven to satisfy a upper bound of the form of Theorem \ref{thm:ldp} (i) first  \cite{KRS}, and a lower bound of the form of Theorem \ref{thm:ldp} (ii) later \cite{RSe}.

\subsubsection*{Regime of standard deviations.} In the regime of standard deviations, i.e. $ x = O ( N^{-1/2})$, Theorem \ref{thm:ldp} furthermore implies 
\begin{align}
\label{eq:CLT}
\lim_{N \rightarrow \infty} \mathbb{P}_{\psi_N} \left[ \sqrt{N} O_N < x\right] = \int_{- \infty}^x e^{-x^2/ (2\| f \|_{\ell^2( \Lambda_+^*)}^2)}, 
\end{align}
thus a central limit theorem where the limiting Gaussian random variable's variance is given by $\| f \|_{\ell^2( \Lambda_+^*)}$ agreeing with earlier results \cite{RS}. In fact \cite{RS} proves a central limit theorem for fluctuations around the condensate for the ground state in the Gross-Pitaevski regime. The Gross-Pitaevski scaling regime considers instead of $v$, the $N$-dependent two-body interaction potential $v_N^\beta = N^{3 \beta} v_N (N^\beta \cdot )$ with $\beta =1$ (for more details and recent progress on results in the Gross-Pitaevski regime see \cite{BBCS_optimal,BBCS_Acta,BSchSch,HST,NT}). However, \eqref{eq:CLT} follows from adapting the analysis in \cite{RS} to the mathematically easier accessible mean-field scaling regime (corresponding to $\beta =0$).

Recently, \cite{BP} refined the characterization of the regime of standard deviations and derived an edge-worth expansion. 

Central limit theorems were proven first for the mean-field dynamics of Bose gases. Fluctuations of bounded one-particle operators around the Hartree equations were proven to have Gaussian behavior \cite{BKS}, though they do not correspond to i.i.d. random variables. These results were later generalized to multivariate central limit theorem \cite{BSS}, dependent random variables (i.e. $k$-particle operators) \cite{R_dp} and singular particles interaction in the intermediate scaling regime (for $v_N^\beta$ with $\beta \in (0,1))$) \cite{R}.  \\

Theorem \ref{thm:ldp} follows (similarly to \cite{KRS,RSe}) from estimates on the logarithmic moment generating function given in the following theorem. 

 \begin{theorem}
 Under the same assumptions as in Theorem \ref{thm:ldp}, 
 \label{thm:main}
 \begin{enumerate}
 \item[(i)] there exists a constant $C_1 >0$ such that for all $0 \leq \lambda \leq 1 / \vertiii{O}$ we have 
 \begin{align}
\liminf_{N \rightarrow \infty} N^{-1} \ln \mathbb{E}_{\psi_{N}} \left[ e^{\lambda O_{N}} \right]  \leq  \frac{\lambda^2 } 2 \| f \|^2_{\ell^2( \Lambda_+^*)} + C_1 \lambda^3 \vertiii{O}^3 
\end{align}
\item[(ii)] there exists a constant $C_2>0$ such that for all $0 \leq \lambda \leq 1 / \| \vertiii{O}$ we have 
 \end{enumerate}
 \begin{align}
\limsup_{N \rightarrow \infty} N^{-1}  \ln \mathbb{E}_{\psi_{N}} \left[ e^{\lambda O_{N}} \right]   \geq  \frac{\lambda^2 }{2} \| f \|^2_{\ell^2( \Lambda_+^*)} - C_2 \lambda^3 \vertiii{O}^3  
\end{align}
\end{theorem}

Theorem \ref{thm:ldp} follows from Theorem \ref{thm:main} by a generalization of Cramer's theorem(see \cite[Section 2]{RSe}).    \\

\subsubsection*{Idea of the proof.} The rest of this paper is dedicated to the proof of Theorem \ref{thm:main}, thus on estimates on the moment generating function. We recall that for the result of Theorem \ref{thm:main} we are interested in the leading order of the exponential of the moment generating function that is $o ( N\lambda^2 )$ in the limit of small $\lambda$ and large $N$. We will show that for the leading order fluctuations around the condensate are crucial that we describe by the excitation vector $\mathcal{U}_N \psi_N$ (for a precise definition of the unitary map $\mathcal{U}_N$ to the Fock space of excitations see \eqref{def:UN}). As a first step we prove that we can replace the moment generating function  $\mathbb{E}_{\psi_N} \left[ e^{\lambda O_N} \right]$ with the expectation value
\begin{align}
\langle \mathcal{U}_N  \psi_N, \;e^{\lambda \phi_+ (q O \varphi_0) /2} e^{\lambda \kappa \mathcal{N}_+} e^{\lambda \phi_+ (q O \varphi_0) /2} \mathcal{U}_N \psi_N \rangle \; 
\end{align} 
paying a price exponentially  $O ( N \lambda^3)$ and thus subleading (see Lemma \ref{lemma:step1}). Here we introduced the notation
\begin{align}
\label{eq:def-phi+-0}
\phi_+ (q_0 O \varphi_0) = \sqrt{N - \mathcal{N}_+} a(q_0 O \varphi_0) +  a^*(q_0 O \varphi_0)\sqrt{N - \mathcal{N}_+} \;  
\end{align}
where $a,a^*$ denote the creation and annihilation operators on the bosonic Fock space and $\cN_+$ the number of excitations (for a precise definition see Section \ref{sec:fluc}). Note that the operator $\phi_+$, in contrast to its asymptotic limit
\begin{align}
\widetilde{\phi}_+ (h) = \sqrt{ N} a(q_0 O \varphi_0) + \sqrt{ N} a^* (q_0 O \varphi_0)
\label{eq:widetilde-phi}
\end{align}
for $N \rightarrow \infty$, does not increase the number of excitations which will be crucial for our analysis. 

We remark that the excitation vector $ \mathcal{U}_N \psi_N$ is the ground state of the excitation Hamiltonian 
\begin{align}
\mathcal{U}_N H_N \mathcal{U}_N^* = \mathcal{Q} + \mathcal{R}_N  \; . 
\end{align}
Its quadratic (in modified creation and annihilation operators) part $\mathcal{Q}$ and the remainder term $\mathcal{R}_N$ are given in \eqref{def:Q} resp. \eqref{def:R}. In the second step we show that replacing $\mathcal{U}_N \psi_N$ with the ground state $\psi_{\mathcal{Q}}$ of the quadratic operator $\mathcal{Q}$ leads to an error exponentially $O( N \lambda^3)$ ( and thus subleading). While the first step follows strategies presented in \cite{KRS} on the dynamical problem, the second step uses novel techniques. The proof is based on the Heymann-Feynmann theorem and Gronwall's inequality applied for $s \in [0,1]$ to the family of ground states $\psi_{\mathcal{G}_N (s)}$ that corresponds to the Hamiltonians $\mathcal{G}_N (s) = \mathcal{Q} + s \mathcal{R}_N$ and thus interpolates between the excitation vector $\mathcal{U}_N \psi_N$ and $\psi_{\mathcal{Q}}$ (for more details see Proposition \ref{claim:gs} and Lemma \ref{lemma:step2}). 

We remark that the ground state of operators quadratic in standard creation and annihilation operators is well known and given by a quasi-free state, i.e. by 
\begin{align}
\label{def:Bogo-a}
e^{\widetilde{B}( \mu )}  \Omega \quad \text{with} \quad \widetilde{B}(\mu) = \frac{1}{2} \sum_{p \in \Lambda_+^*}\mu_p  \left(a_p^*a_{-p}^* - a_pa_{-p} \right) \; 
\end{align}
where $\mu$ is given by \eqref{def:mu} and vacuum vector $\Omega$. Note that the operator $\mathcal{Q}$ is quadratic in modified creation and annihilation operators. However, we will prove that its ground state $\psi_{\mathcal{Q}}$ is approximately given by a generalized quasi-free state, i.e. by 
\begin{align}
\label{def:Bogo-b}
e^{B ( \mu)} \Omega \quad \text{with} \quad B(\mu) = \frac{1}{2} \sum_{p \in \Lambda_+^*}\mu_p  \left(b_p^*b_{-p}^* - b_pb_{-p} \right) \; . 
\end{align}
A crucial property of a Bogoliubov transform \eqref{def:Bogo-a} is that its action on creation and annihilation operators is explicitly known. In particular, we have for the asymptotic limit of $\phi_+$ that 
\begin{align}
e^{\widetilde{B} ( \mu)}  \widetilde{\phi}_+ (q_0 O \varphi_0) e^{-\widetilde{B} ( \mu )} 
= \widetilde{\phi}_+ (f) \; . 
\end{align}
Though the explicit action of the generalized Bogoliubov transform \eqref{def:Bogo-b} on the operator $\phi_+$  is not known we show in the third step that we still have 
\begin{align}
e^{B ( \mu )} \phi_+ ( q_0 O \varphi_0 ) e^{- B ( \mu )} \approx \phi_+ ( f) 
\end{align}
with an error exponentially $O ( N\lambda^3)$. This argument will be based again on the Hellmann-Feynmann theorem together with Gronwall's inequality applied to the family of ground states $\psi_{\mathcal{Q} (s)}$ of $\mathcal{Q}(s) = \mathcal{D} + s \mathcal{R}_{\mathcal{Q}}$ for $s \in [0,1]$ where $\mathcal{D}$ is a quadratic, diagonal operator. Thus $\mathcal{Q}(s)$ interpolates between the ground state $e^{B ( \mu)}\psi_{\mathcal{Q}}$ and the vacuum vector (see Lemma \ref{lemma:step3}). 

In the last step we then compute the remaining expectation value 
\begin{align}
\langle \Omega, e^{\lambda \phi_+ (f) /2} e^{\lambda \kappa \mathcal{N}_+} e^{\lambda \phi_+ (f )/2} \Omega \rangle 
\end{align}
with $f$ given by \eqref{def:f}. A comparison with the asymptotic limit $\widetilde{\phi}_+$ shows that the exponential of $\lambda \cN_+$ contributes exponentially $O( N\lambda^3)$, and thus subleading, leading to Theorem \ref{thm:main}. For the true operator $\phi_+$ this holds still true in the limit $N \rightarrow \infty$ and follows from arguments given in \cite{KRS} (see Lemma \ref{lemma:step4}).  

\subsubsection*{Structure of the paper} The paper is structured as follows: In Section \ref{sec:fluc} we introduce the description of the fluctuations (called excitations) around the condensate in the Fock space of excitations. In particular, we prove properties of the excitations' Hamiltonian $\mathcal{G}_N $ and the quadratic approximation $\mathcal{Q}$ and their corresponding ground states (see Propositions \ref{claim:gs}, \ref{claim:Qs}). In Section \ref{sec:prel} we recall preliminary results from \cite{KRS,RSe} and prove further auxiliary Lemmas (in particular for generalized Bogoliubov transforms \eqref{def:Bogo-b}) that we will use later for the proof of Theorem \ref{thm:main} in Section \ref{sec:proof}.

\section{Fluctuations around the condensate}
\label{sec:fluc}

\subsection{Fock space of excitations} On the unit torus the condensate wave function $\varphi_0$ is given by the constant function. To study the fluctuations around the condensate, we need to factor out the condensates contributions. For this we use an observation from \cite{LNSS} that any $N$-particle wave function $\psi_N \in L^2( \Lambda^N)$ can be decomposed as 
\begin{align}
\psi_N = \eta_0 \otimes_s \varphi_0^{\otimes N} + \mu_1 \otimes_s \varphi_0^{\otimes (N-1)} + \dots + \eta_N 
\end{align} 
where the excitation vectors $\eta_j$ are elements of $L^2_{\perp \varphi_0} (\Lambda^j)$, the orthogonal complement in $L^2 (\Lambda^j)$ of the condensate wave function $\varphi_0$ and $\otimes_s$ denotes the symmetric tensor product. Furthermore, we define the unitary 
\begin{align}
\label{def:UN}
\mathcal{U}_N: L_s^2 \left( \Lambda^{N}\right) \rightarrow \mathcal{F}_{\perp \varphi_0}^{\leq N}, \quad  \psi_N \mapsto \lbrace \eta_1, \dots, \eta_N \rbrace 
\end{align}
mapping any $N$-particle wave function $\psi_N$ onto its excitation vector $\lbrace \eta_1, \dots, \eta_N \rbrace$ that is an element of the Fock space of excitations 
\begin{align}
\mathcal{F}_{\perp \varphi_0}^{\leq N} = \bigoplus_{j=0}^N L^2_{\perp \varphi_0} (\Lambda)^{\otimes_s j}  \; . 
\end{align}
A crucial property of elements of the Fock space of excitations $\xi_N \in \mathcal{F}_{\perp \varphi_0}^{\leq N}$ is that the number of particles operator  $\mathcal{N} = \sum_{p \in \Lambda^*} a_p^* a_p$ is bounded, i.e. $ \langle \xi_N, \mathcal{N} \xi_N \rangle \leq N \| \xi_N \|^2$. Here we introduced the standard creation and annihilation operators $a_p^*,a_p$ in momentum space defined through the following relation by the well known creation and annihilation operators in position space $\check{a} (f), \check{a}^* (f)$
\begin{align}
\label{def:ap}
a_p^* = \check{a}( \varphi_p) , \quad \text{resp.} \quad a_p = \check{a} (\varphi_p) \quad \text{with} \quad \varphi_p = e^{i p \cdot x} \quad \text{for} \quad p \in \Lambda_+^* = 2 \pi \mathbb{Z}^3
\end{align}
that satisfy the canonical commutation relations 
\begin{align}
\label{eq:CCR}
\left[ a_p^*, a_q \right] = \delta_{p,q}, \quad \text{and} \quad \left[ a_p, a_q \right] = \left[ a_p^*, a_q^* \right] =0 \; . 
\end{align}
Contrarily, on the full bosonic Fock space built over $L^2 (\Lambda^j)$ (instead of $L^2_{\perp \varphi_0} (\Lambda^j)$) and given by 
\begin{align}
\mathcal{F} = \bigoplus_{j=0}^\infty L^2 (\Lambda)^{\otimes_s j}  \; . 
\end{align}
the number of particles $\mathcal{N} = \sum_{p \in \Lambda} a_p^* a_p$ is an unbounded operator. 

For our analysis it will be useful to work on the Fock space of excitations that is equipped with modified creation and annihilation operators $b^*_p, b_q$ that leave (in contrast to the standard ones $a^*_p, a_q$)  $\mathcal{F}_{\perp \varphi_0}^{\leq N}$ invariant and were first introduced in \cite{BS}. They are given by 
\begin{align}
\label{def:b}
b_p = \frac{\sqrt{N - \mathcal{N}_+}}{\sqrt{N}} a_p , \quad b^*_p = a_p^* \frac{\sqrt{N -\mathcal{N}_+}}{\sqrt{N}}  \; .
\end{align}
with the number of excitations 
\begin{align}
\label{def:N+}
\mathcal{N}_+ = \sum_{p \in \Lambda^*_+} a_p^* a_p \quad \text{and } \quad \Lambda^*_+ = \Lambda^* \setminus \lbrace 0 \rbrace \; . 
\end{align}
It follows from \eqref{eq:CCR} that $b^*_p, b_q$ satisfy the modified commutation relations 
 \begin{align}
 \label{eq:comm}
\left[ b_p , b^*_q \right] = \delta_{pq}\left( 1 - \frac{\mathcal{N}_+}{N} \right) - \frac{1}{N} a^*_q a_p  \; .  \left[ b^*_p , b^*_q \right] = \left[ b_p, b_q \right] = 0  \; . 
 \end{align}
We remark that in the limit of $N \rightarrow \infty$, the commutation relations of $b_p^*,b_q$ agree with the canonical commutation relations \eqref{eq:CCR}. However the corrections that are $O(N^{-1})$ lead to difficulties in the analysis later. 

The operator $b_p^*,b_q$ arise from the unitary $\mathcal{U}_N$ applied for $p, q \not=0 $ to products of creation and annihilation operators, namely 
\begin{align}
\label{eq:prop-U1}
\mathcal{U}_N a_0^* a_q \mathcal{U}_N = \sqrt{N} b_p, \quad  \text{and} \quad \mathcal{U}_N a_q^* a_0 \mathcal{U}_N = \sqrt{N} b_p^* \; . 
\end{align}
Moreover, $\mathcal{U}_N$ satisfies the property 
\begin{align}
\label{eq:prop-U2}
\mathcal{U}_N a_p^* a_q \mathcal{U}_N =  a_p^* a_q, \quad \mathcal{U}_N a_0^* a_0 \mathcal{U}_N = N - \mathcal{N}_+ \; .  
\end{align}

\subsection{Excitation Hamiltonian} We can embed $L^2_s ( \Lambda^N)$ into the full bosonic Fock space $\mathcal{F}$ where the Hamiltonian $H_N$ defined in \eqref{def:HN} then reads in momentum space
\begin{align}
H_N := \sum_{p \in \Lambda^*} p^2 a_p^* a_p + \frac{1}{2N} \sum_{p,q,k \in \Lambda^*} \widehat{v} (k) a_{p-k}^* a_{q+k}^* a_p a_q \; . 
\end{align}
If $\psi_N$ denotes the ground state of $H_N$, then the excitation vector $\mathcal{U}_N \psi_N =: \psi_{\mathcal{G}_N} $ denotes the ground state of the excitation Hamiltonian 
\begin{align}
\mathcal{G}_N := \mathcal{U}_N H_N \mathcal{U}_N^*  - \frac{N}{2} \widehat{v} (0)
\end{align}
that can be explicitly computed using the properties of the unitary \eqref{eq:prop-U1}, \eqref{eq:prop-U2} and is of the form 
\begin{align}
\label{def:G}
\mathcal{G}_N = \mathcal{Q} + \mathcal{R}_N 
\end{align}
where $\mathcal{Q}$ denotes an operator quadratic in (standard) creation and annihilation operators and is given with the notation $\Lambda_+^* = \Lambda_+ \setminus \lbrace 0 \rbrace $  by
\begin{align}
\label{def:Q}
\mathcal{Q} := \sum_{p \in \Lambda^*_+} \left[   p^2 a_p^*a_p+ \widehat{v} (p) b_p^* b_p  + \frac{1}{2} \widehat{v} (p) \left( b_p^* b_{-p}^* + b_p b_{-p} \right) \right]  \;  
\end{align}
whereas the remainder terms collected in $\mathcal{R}_N$ and given by 
\begin{align}
\label{def:R}
\mathcal{R}_N :=& \frac{1}{\sqrt{N }} \sum_{p,q \in \Lambda_+^*, p \not= -q} \widehat{v} (q)  \left( b_{p+ q}^* a_{-q}^* a_p + {\rm h.c.} \right) \notag \\
& + \frac{1}{2N} \sum_{p,q \in \Lambda_+^*, q \not= -p,k} \widehat{v} (k) a_{p+k}^* a_{k-q}^* a_p a_k 
\end{align}
will be shown to contribute to our analysis sub-leading only. In fact, in the proof of Theorem \ref{thm:main} in Section \ref{sec:proof} it turns out that $\mathcal{Q}$ resp. its corresponding ground state $\psi_\mathcal{Q}$ is approximately given by 
\begin{align}
\label{def:psiQ}
\psi_{\mathcal{Q}} = e^{B ( \mu)} \Omega \; \quad \text{with} \quad e^{B ( \mu )} = \exp \Big( \sum_{p \in \Lambda_+^*}  \left[\mu_p  b_p^*b_{-p}^* -  \overline{\mu_p}b_pb_{-p} \right] \Big)
\end{align}
with $\mu_p$ given by \eqref{def:mu},i.e. $\psi_{\mathcal{Q}}$ is a generalized Bogoliubov transform $e^{B( \mu)}$ applied to the vacuum $\Omega$, fully determines the variance (i.e. $\| f \|_{\ell^2 ( \Lambda_+^*)}^2$ in Theorem \ref{thm:ldp}). We remark that the approximation of $\mathcal{G}_N(s)$ by $\mathcal{Q}$ is often referred to as Bogoliubov approximation.

Furthermore, we introduce the family of Hamiltonians $\lbrace \mathcal{G}_N (s) \rbrace_{s \in [0,1]}$ given by
\begin{align}
\label{def:Gs}
\mathcal{G}_N (s) = \mathcal{Q} + s \mathcal{R}_N 
\end{align}
 interpolating between the excitation Hamiltonian $\mathcal{G}_N$ and its quadratic approximation $\mathcal{Q}$. In the following Proposition we collect useful properties of  $\lbrace \mathcal{G}_N (s) \rbrace_{s \in [0,1]}$. For this we introduce the following notation for the particles' kinetic energy
 \begin{align}
 \mathcal{K} = \sum_{p \in \Lambda_+^*} p^2 \; . 
 \end{align}
  
\begin{proposition}
\label{claim:gs}
Let $s \in [0,1]$, then there exists a ground state $\psi_N (s)$ of the Hamiltonian $\mathcal{G}_N (s)$ defined in \eqref{def:Gs}. Furthermore, there exists a constant $C >0$ such that 
\begin{align}
\langle \psi_{\mathcal{G}_N (s)}, \;  (\mathcal{N}_+ +1)^k \psi_{\mathcal{G}_N (s)} \rangle \leq C \label{eq:N-gs}
\end{align}
for $k=1,2$ and the spectrum of the Hamiltonian $\mathcal{G}_N(s)$ has a spectral gap above the ground state $E_N (s)$  independent of $s,N$.

 Moreover, there exists $C>0$ such that for any Fock space vector $\xi \in \mathcal{F}_{\perp \varphi_0}^{\leq N}$ we have 
\begin{align}
\label{eq:resolvent-1}
\Big\|(\mathcal{N}_+ +1) \frac{q_{\psi_{\mathcal{G}_N (s)}}}{ \mathcal{G}_N (s) - E_N (s)} \;  \;  \xi \Big\| &\leq  C \| \xi \| , \notag \\
\Big\|(\mathcal{N}_+ +1)^{-1/2} \frac{q_{\psi_{\mathcal{G}_N (s)}}}{ \mathcal{G}_N (s) - E_N (s)} \;  \;  \xi \Big\| &\leq  C \| (\mathcal{N}_+ +1)^{-3/2} \xi \|
\end{align} 
\end{proposition}

\begin{proof} The proof uses well known ideas and techniques introduced to prove results on the properties of $\mathcal{G}_N$ and its corresponding ground state $\psi_{\mathcal{G}_N}$ showing that the remainder $\mathcal{R}_N$ contributes sub-leading only (see for example \cite{LNSS,NN,R}). Since $\mathcal{G}_N (s)$ differs from $\mathcal{G}_N $ by a multiple of the remainder only, these techniques apply for $\mathcal{G}_N (s)$ as we shall show in the following.  

The strategy is as follows: First we show that $\mathcal{G}_N (s)$ is bounded from below by a multiple of $ \mathcal{N}_+   - C $ yielding the estimate \eqref{eq:N-gs} for $k=1$ and with further arguments for $k=2$, too. Then the remainder $\mathcal{R}_N$ can be proven to be sub-leading, and the existence of a spectral gap of the spectrum of $\mathcal{G}_N(s)$ independent of $N,s$ follows from the spectral properties of $\mathcal{Q}$. Finally we prove \eqref{eq:resolvent-1} from the previously proven properties. 

\subsubsection*{Proof of lower bound for $\mathcal{G}_N (s)$} First we shall prove that there exists constants $C_1, C_2 >0$ such that 
\begin{align}
\label{eq:estimate-N-G}
\mathcal{G}_N (s) \geq C_1 \mathcal{N}_+  - C_2 \; . 
\end{align}
To this end, we recall that by definition \eqref{def:Gs} we have $\mathcal{G}_N (s) =\mathcal{Q} + s \mathcal{R}_N$ for $s \in [0,1]$. For the quadratic operator $\mathcal{Q}$ we find since $\widehat{v}(p) = \widehat{v} (-p)$ and $\widehat{v}(p) \geq 0$
\begin{align}
\mathcal{Q} =& \sum_{p \in \Lambda_+^*} p^2 a_p^*a_p + \frac{1}{2} \sum_{p \in \Lambda_+^*} \widehat{v}(p) \left[b_p^* + b_{-p} \right] \left[b_{-p}^* + b_{p} \right]  -  \frac{1}{2 }\| \widehat{v} \|_{\ell^1} \notag \\
\geq& \sum_{p \in \Lambda_+^*} p^2 a_p^*a_p  - \| \widehat{v} \|_{\ell^1}\geq (2 \pi)^2 \mathcal{N}_+ - \frac{1}{2} \| \widehat{v} \|_{\ell^1}\; . \label{eq:estimate-N-Q}
\end{align}
Thus assuming that there exists sufficiently small $\varepsilon_1 >0$ with 
\begin{align}
\label{eq:estimate-R-N}
\mathcal{R}_N \geq - \varepsilon_1 \mathcal{N}_+ - C
\end{align}
then \eqref{eq:estimate-N-G} follows from \eqref{def:Gs} and \eqref{eq:estimate-N-Q}. We are left with proving \eqref{eq:estimate-R-N}. For this we use that the contribution of $\mathcal{R}_N$ quartic in creation and annihilation operator is non-negative, i.e. that we can write 
\begin{align}
\label{def:VN}
\mathcal{R}_N = \widetilde{\mathcal{R}}_N + \mathcal{V}_N, \quad \text{with} \quad \mathcal{V}_N = \frac{1}{2N} \sum_{p,q \in \Lambda_+^*, q \not= -p,k} \widehat{v} (k) a_{p+k}^* a_{k-q}^* a_p a_k 
\end{align}
and $\mathcal{V}_N \geq 0$ following from $\widehat{v} \geq 0$. Therefore to prove \eqref{eq:estimate-R-N} it suffices to show that 
\begin{align}
\widetilde{\mathcal{R}}_N \geq  - \varepsilon_1 \mathcal{N}_+  - \varepsilon_2 \mathcal{V}_N - C \label{eq:estimate-R-V}
\end{align}
for sufficiently small $\varepsilon_1, \varepsilon_2 >0$. We estimate the single contributions of $\mathcal{R}_N$ given in \eqref{def:R} separately using the bounds  
\begin{align}
\label{eq:bounds-a} 
\| a(h) \xi \| \leq \| h \|_{\ell^2( \Lambda^*)} \| \mathcal{N}^{1/2} \xi \|, \quad \| a^* (h) \xi \| \leq \| h \|_{\ell^2( \Lambda^*)} \| ( \mathcal{N}+ 1)^{1/2} \xi \|
\end{align}
where $a^*(h)= \sum_{p \in \Lambda^*} h_p a_p^* $ for any $h \in \ell^2 ( \Lambda_+^*)$, resp. for the modified creation and annihilation operators 
\begin{align}
\label{eq:bounds-b} 
\| b(h) \xi \| \leq \| h \|_{\ell^2( \Lambda_+^*)} \| \mathcal{N}^{1/2} \xi \|, \quad \| b^* (h) \xi \| \leq \| h \|_{\ell^2( \Lambda_+^*)} \| ( \mathcal{N}+ 1)^{1/2} \xi \| \; . 
\end{align}
By definition \eqref{def:R}, the operator $\widetilde{\mathcal{R}}_N$ is cubic in creation and annihilation operators and can be bounded with \eqref{eq:bounds-b}  by
\begin{align}
N^{-1/2} & \vert  \langle \psi , \; \sum_{p,q \in \Lambda_+^*, p \not= -q} \widehat{v} (q)   b_{p+ q}^* a_{-q}^* a_p \psi \rangle \vert \notag \\
= & N^{-1/2} \vert \langle \psi, \; \sum_{p,q \in \Lambda_+^*, p \not= -q} \widehat{v} (q)   a_{p+ q}^* a_{-q}^* a_p \sqrt{1 - (\mathcal{N}_+ +1 ) /N} \psi \rangle \vert \notag \\
\leq& C \Big((2N^{-1} \sum_{p,q \in \Lambda_+^*, p \not= -q} \widehat{v} (q)  \| a_{p+ q}a_{-q}  \psi \|^2 \Big)^{1/2} \notag \\
& \vspace{1cm} \times \Big(   \sum_{p,q \in \Lambda_+^*, p \not= -q} \widehat{v} (q) \| a_p \sqrt{1 - (\mathcal{N}_+ +1 ) /N} \psi \|^2 \Big)^{1/2}  \; . 
\end{align}
We switch to position space for the first factor and find 
\begin{align}
(2N)^{-1} \sum_{p,q \in \Lambda_*^*, p \not= q}  & \widehat{v} (q) \| a_{p+q} a_{-q} \psi \|^2 \notag \\
&= (2N)^{-1} \int dxdy \; v (x-y) \; \langle \psi, a_x^*a_y^*a_ya_x \; \psi \rangle = \langle \psi, \; \mathcal{V}_N \psi \rangle \; 
\end{align}
and therefore 
\begin{align}
N^{-1/2}  &  \vert  \langle \psi , \; \sum_{p,q \in \Lambda_+^*, p \not= -q} \widehat{v} (q)   b_{p+ q}^* a_{-q}^* a_p \psi \rangle \vert 
\leq& C \| \widehat{v} \|_{\ell^1} \| \mathcal{V}_N^{1/2} \psi \| \; \| ( \mathcal{N} + 1 )^{1/2} \psi \| \; . \label{eq:estimate-R-3} 
\end{align}
The hermitian conjugate can be estimated similarly. Thus with \eqref{eq:estimate-R-3} we arrive at \eqref{eq:estimate-R-V} and thus at \eqref{eq:estimate-N-G} using that $s \in [0,1]$. 

\subsubsection*{Proof of \eqref{eq:N-gs}} The lower bound \eqref{eq:estimate-N-G} has several consequences: On the one hand, the lower bound \eqref{eq:estimate-N-G} shows that $\mathcal{G}_N (s)$ is bounded from below by a constant $- C (N+1)$. On the other hand \eqref{eq:estimate-N-G} shows that for any normalized $\xi \in \mathcal{F}$ with $\mathds{1}_{\mathcal{G}_N (s) \leq \zeta } \xi = \xi$ (i.e. in particular for the ground state) that 
\begin{align}
\langle \xi, \; \mathcal{N}_+ \xi \rangle \leq C_1^{-1} \langle \xi, \;  \mathcal{G}_N (s) \xi \rangle + C_2 \leq C_1^{-1} \zeta + C_2 \; \label{eq:bound-N-k1}
\end{align}
which proves \eqref{eq:N-gs} for $k=1$. To prove \eqref{eq:N-gs} for $k=2$ we remark that \eqref{eq:estimate-N-G} furthermore implies 
\begin{align}
\langle \xi, \; (\mathcal{N}_+ +1)^2 \xi \rangle \leq& C \langle \xi, (\mathcal{N}_+ + 1 )^{1/2} \mathcal{G}_N (s) (\mathcal{N}_++1)^{1/2} \xi \rangle \; \notag \\
 =& C \zeta \langle \xi, \; \mathcal{N}_+ \xi \rangle + C \langle \xi, \mathcal{N}_+^{1/2} \left[ \mathcal{G}_N (s),  \mathcal{N}_+^{1/2} \right] \xi \rangle . \label{eq:bound-N-k2-1} 
\end{align}
With spectral calculus we find that 
\begin{align}
\label{eq:spec-calc}
\left[ \mathcal{G}_N (s),  \mathcal{N}_+^{1/2} \right] = \frac{1}{\pi}\int_0^\infty  \frac{\sqrt{t}}{ \mathcal{N}_+ + 1 + t} \left[ \mathcal{G}_N (s), \mathcal{N}_+ \right] \frac{1}{ \mathcal{N}_+ + 1 + t} dt 
\end{align}
We recall the definition of $\mathcal{G}_N (s)$ and compute the commutators for every term separately. We have 
\begin{align}
\left[ \mathcal{Q}, \mathcal{N}_+ \right] = \sum_{p \in\Lambda_+^*} \widehat{v} (p) \left[ b_p^*b_{-p}^* -  b_pb_{-p} \right] 
\end{align}
and thus 
\begin{align}
\vert \langle  & \xi, \mathcal{N}_+^{1/2} \left[ \mathcal{Q},  \mathcal{N}_+^{1/2} \right] \xi \rangle \vert \notag \\
 \leq&  \frac{1}{2} \int_0^\infty \sqrt{t} \; \Big( \sum_{p \in \Lambda_+^*} \widehat{v} (p) \Big\| b_p \frac{ (\mathcal{N}_+ + 1)^{1/2}}{\mathcal{N}_+ +1 + t}\xi \Big\| \Big)^{1/2} \Big( \sum_{p \in \Lambda_+^*} \widehat{v} (p) \Big\| b_{-p}^* \frac{ 1}{\mathcal{N}_+ +1 + t}\xi \Big\| \Big)^{1/2} dt \notag \\
&\quad +  \frac{1}{2} \int_0^\infty \sqrt{t} \; \Big( \sum_{p \in \Lambda_+^*} \widehat{v} (p) \Big\| b_p^* \frac{ (\mathcal{N}_+ + 1)^{1/2}}{\mathcal{N}_+ +1 + t}\xi \Big\| \Big)^{1/2} \Big( \sum_{p \in \Lambda_+^*} \widehat{v} (p) \Big\| b_{-p} \frac{ 1}{\mathcal{N}_+ +1 + t}\xi \Big\| \Big)^{1/2} dt \notag \\
\leq&  \int_0^\infty \frac{\sqrt{t}}{( 1 + t)^2}  dt \; \| ( \mathcal{N}_+ +  1 ) \xi \| \; \| ( \mathcal{N}_+ + 1)^{1/2} \xi \|  \leq C \| ( \mathcal{N}_+  + 1 ) \xi \| \; \| ( \mathcal{N}_+ + 1)^{1/2} \xi \| \; . 
\end{align}
The commutator with $\mathcal{R}_N$ follows similarly using that $[\mathcal{N}_+ , a_p] = - a_p$ resp.  $[\mathcal{N}_+ , a_p^*] = a_p^*$ and analogous estimates as in \eqref{eq:estimate-R-3}) (with $\mathcal{V}_N \leq CN^{-1} \mathcal{N}_+^2 \leq C \mathcal{N}_+$ on $\mathcal{F}_{\perp \varphi_0}^{\leq N}$). Thus we arrive at
\begin{align}
\label{eq:estimate-comm}
\vert\langle \xi, \mathcal{N}_+^{1/2} \left[ \mathcal{G}_N (s),  \mathcal{N}_+^{1/2} \right] \xi \rangle \vert \leq  C \| ( \mathcal{N}_+  + 1 ) \xi \| \; \| ( \mathcal{N}_+ + 1)^{1/2} \xi \| 
\end{align}
and with \eqref{eq:bound-N-k2-1} furthermore at 
\begin{align}
( 1- \frac{1}{2}) \langle \xi, \; (\mathcal{N}_+ +1)^2 \xi \rangle \leq C \langle \xi, \; (\mathcal{N}_+ +1) \xi \rangle \leq C
\end{align}
from \eqref{eq:bound-N-k1} yielding \eqref{eq:N-gs} for $k=2$. 

With \eqref{eq:bound-N-k1} we can now refine the estimates on the remainder. We shall prove that 
\begin{align}
\label{eq:R-norm}
\| \mathcal{R}_N \psi \|^2 \leq C N^{-1} \| (\mathcal{N} + 1 )^{3/2} \psi \|^2  \; . 
\end{align}
For the contribution of $\mathcal{R}_N$ in \eqref{def:R} cubic in (modified) creation and annihilation operators, we switch to position space and compute with the commutation relations for any vector $\psi \in \mathcal{F}_{\perp \varphi_0}^{\leq N}$ 
\begin{align}
\|  &\sum_{p,q \in \Lambda_+^*, p \not= -q}  \widehat{v} (q)   b_{p+ q}^* a_{-q}^* a_p \psi \|^2 \notag \\
=& \int_{\Lambda^{4}} dxdydzdw \; v(x-y) v(z- w ) \langle \psi, \; a_z^*a_x^*a_y^* a_w a_z a_x  ( N - \cN_+ - 1)/N  \; \psi \rangle  \notag \\
&+ \int_{\Lambda^{3}} dxdydz \; v (x-y) v(x-z) \langle \psi, \; a_z^* a_y^*a_z a_x  ( N - \cN_+ - 1)/N  \; \psi \rangle \notag \\
&+  \int_{\Lambda^{3}} dxdydw \;  \left( v (x-y) v(y-w) + v (x-y) v(w-x) \right)  \langle \psi, \; a_x^* a_y^*a_w a_x  ( N - \cN_+ - 1)/N  \; \psi \rangle \notag \\
&+  \int_{\Lambda^{3}} dxdydz \; v (x-y) v(z-y) \langle \psi, \; a_z^* a_x^*a_z a_x  ( N - \cN_+ - 1)/N  \; \psi \rangle  \notag \\
&+  \int_{\Lambda^{2}} dxdy \;  v (x-y)^2  \langle \psi, \; \left( a_y^*a_x  + a_x^*a_x \right)    ( N - \cN_+ - 1)/N  \;  \psi \rangle \; . 
\end{align}
Since $\| v \|_{L^\infty ( \Lambda)} \leq \| \widehat{v} \|_{\ell^1} ( \Lambda^*) \leq C$ we thus conclude by \eqref{eq:bounds-a} that 
\begin{align}
N^{-1}\|  \sum_{p,q \in \Lambda_+^*, p \not= -q}  \widehat{v} (q)   b_{p+ q}^* a_{-q}^* a_p \psi \|^2  \leq C N^{-1} \| ( \mathcal{N}_+ + 1)^{3/2} \psi \|^2 \; . \label{eq:R-norm2}
\end{align}
We proceed similarly for the contribution of $\mathcal{R}_N$ cubic in creation and annihilation operators $\mathcal{V}_N$ (see \eqref{def:R} resp.  \eqref{def:VN}). With the commutation relation we find for any $\psi \in \mathcal{F}_{\perp \varphi_0}^{\leq N}$
\begin{align}
(2N)^2  &\| \mathcal{V}_N   \psi \|^2 \notag\\
=& \int_{\Lambda^4} dxdydzdw \;   v(x-y) v(z-w) \; \langle \psi, \; a_z^*a_w^*a_x^*a_y^* a_za_wa_x a_y \psi \rangle \; \notag \\ 
&+ 2 \int_{\Lambda^3} dxdydz  \left( v(x-y) v(z-x) +  v(x-y) v(z-y) \right) \; \langle \psi, \; a_z^*a_x^*a_y^* a_za_x a_y \psi \rangle  \notag \\
&+ 2 \int_{\Lambda^2} dxdydz  \; v(x-y)^2  \; \langle \psi, \; a_x^*a_y^* a_x a_y \psi \rangle  
\end{align}
and thus we arrive with $\| v \|_{L^\infty ( \Lambda)} \leq \| \widehat{v} \|_{\ell^1 ( \Lambda^*)} \leq C$ at 
\begin{align}
\label{eq:R-norm3}
\| \mathcal{V}_N \psi \|^2 \leq C N^{-1} \| ( \mathcal{N}_+  + 1 )^2  \psi \|^2 \,. 
\end{align}
Summarizing \eqref{eq:R-norm1}, \eqref{eq:R-norm2} and \eqref{eq:R-norm3} we thus arrive at the desired estimate \eqref{eq:R-norm}. 

Therefore with \eqref{eq:R-norm} we find that for any $\psi \in \mathcal{F}_{\perp \varphi_0}^{\leq N}$ in the limit $N \rightarrow \infty$
\begin{align}
\langle \psi, \mathcal{G}_N (s) \psi \rangle =  \langle \psi, \mathcal{G}_N (s) \psi \rangle  + O (N^{-1/2} ) \label{eq:spectrum}
\end{align}
and with the min max principle it follows that the low energy states are determined through the quadratic Hamiltonian $\mathcal{Q}$. In particular the spectrum of $\mathcal{G}_N (s)$ has a spectral gap independent of $N,s$ (given in leading order by the spectral gap of $\mathcal{Q}$ (for more details see for example \cite{LNSS}). Furthermore with similar arguments as in \cite{GS} it follows that for every $s \in [0,1]$ there exists a ground state $\psi_N(s)$ approximated by the ground state of $\mathcal{Q}$. 

\subsubsection*{Proof of \eqref{eq:resolvent-1}}  

With \eqref{eq:estimate-N-G}  we find 
\begin{align}
&\frac{q_{\psi_{\mathcal{G}_N (s)}}}{ \mathcal{G}_N (s)  - E_N (s) } ( \mathcal{N}_+ + 1 )^2 \frac{q_{\psi_{\mathcal{G}_N (s)}}}{ \mathcal{G}_N (s)  - E_N (s) }  \notag \\
&\quad  \leq C \frac{q_{\psi_{\mathcal{G}_N (s)}}}{ \mathcal{G}_N (s)  - E_N (s) } ( \mathcal{N}_+ + 1 )^{1/2} \left( \mathcal{G}_N (s) - E_N (s)  + 1 \right)  ( \mathcal{N}_+ + 1 )^{1/2} \frac{q_{\psi_{\mathcal{G}_N (s)}}}{ \mathcal{G}_N (s)  - E_N (s) }   \label{eq:estimate-resolvent4}
\end{align}
In order to use \eqref{eq:estimate-N-G} once more, we write the r.h.s. as 
\begin{align}
&\frac{q_{\psi_{\mathcal{G}_N (s)}}}{ \mathcal{G}_N (s)  - E_N (s) } ( \mathcal{N}_+ + 1 )^{1/2} \left( \mathcal{G}_N (s) - E_N (s) \right)  ( \mathcal{N}_+ + 1 )^{1/2} \frac{q_{\psi_{\mathcal{G}_N (s)}}}{ \mathcal{G}_N (s)  - E_N (s) }    \notag \\
&\quad   = \frac{1}{2} \Big[  (\mathcal{N}_+ + 1 )^{1/2} , \; \Big[ (\mathcal{N}_+ + 1)^{1/2}, \;  \frac{q_{\psi_{\mathcal{G}_N (s)}}}{ \mathcal{G}_N (s)  - E_N (s) } \Big] \Big] \notag \\
&\quad \quad + \frac{1}{2} \frac{q_{\psi_{\mathcal{G}_N (s)}}}{ \mathcal{G}_N (s)  - E_N (s) }\left[ \left[  ( \mathcal{N}_+ + 1 )^{1/2}, \; \mathcal{G}_N (s) \right] \;  , ( \mathcal{N}_+ + 1 )^{1/2} \right] \frac{q_{\psi_{\mathcal{G}_N (s)}}}{ \mathcal{G}_N (s)  - E_N (s) }    \; . 
\end{align}
For the first term we find similarly to \eqref{eq:spec-calc} with spectral calculus
\begin{align}
& \Big[ ( \mathcal{N}_+ + 1)^{1/2},  \frac{q_{\psi_{\mathcal{G}_N (s)}}}{ \mathcal{G}_N (s)  - E_N (s) } \Big] \notag \\
=& \frac{q_{\psi_{\mathcal{G}_N (s)}}}{ \mathcal{G}_N (s)  - E_N (s) }  \left[ ( \mathcal{N}_+ + 1)^{1/2},  \mathcal{G}_N (s)  \right] \frac{q_{\psi_{\mathcal{G}_N (s)}}}{ \mathcal{G}_N (s)  - E_N (s) }  \notag \\
=& \frac{1}{\pi} \frac{q_{\psi_{\mathcal{G}_N (s)}}}{ \mathcal{G}_N (s)  - E_N (s) }  \int_0^\infty   \frac{\sqrt{t}}{ \mathcal{N}_+ + 1 + t} \left[ \mathcal{G}_N (s), \mathcal{N}_+ \right] \frac{dt}{ \mathcal{N}_+ + 1 + t} \; \;  \frac{q_{\psi_{\mathcal{G}_N (s)}}}{ \mathcal{G}_N (s)  - E_N (s) }  \label{eq:comm-G-N12}
\end{align}
and thus with similar estimates as in \eqref{eq:spec-calc}-\eqref{eq:estimate-comm} and \eqref{eq:estimate-N-G} 
\begin{align}
\vert& \langle \xi, \Big[ (\mathcal{N}_+ + 1 )^{1/2} , \; \Big[ \mathcal{N}_+^{1/2}, \;  \frac{q_{\psi_{\mathcal{G}_N (s)}}}{ \mathcal{G}_N (s)  - E_N (s) } \Big] \Big]   \xi \rangle \vert \notag \\
&\leq C \Big\| \frac{q_{\psi_{\mathcal{G}_N (s)}}}{ \mathcal{G}_N (s)  - E_N (s)} ( \mathcal{N}_+ +1 )^{1/2} \xi \| \; \|  ( \mathcal{N}_+ +1 ) \frac{q_{\psi_{\mathcal{G}_N (s)}}}{ \mathcal{G}_N (s)  - E_N (s)} \xi \|  \notag \\
&\leq C \| \xi \| \; \|  ( \mathcal{N}_+ +1 ) \frac{q_{\psi_{\mathcal{G}_N (s)}}}{ \mathcal{G}_N (s)  - E_N (s)} \xi \| \; .\label{eq:estimate-resolvent5}  
\end{align}
The second term of the r.h.s. of \eqref{eq:estimate-resolvent4} can be estimated similarly and we find  with \eqref{eq:spec-calc}-\eqref{eq:estimate-comm},  \eqref{eq:estimate-N-G} that 
\begin{align}
\vert & \langle \xi, \;  \frac{q_{\psi_{\mathcal{G}_N (s)}}}{ \mathcal{G}_N (s)  - E_N (s) }\left[ \left[  ( \mathcal{N}_+ + 1 )^{1/2}, \; \mathcal{G}_N (s) \right] \;  , ( \mathcal{N}_+ + 1 )^{1/2} \right] \frac{q_{\psi_{\mathcal{G}_N (s)}}}{ \mathcal{G}_N (s)  - E_N (s) }   \xi \rangle \vert \notag \\
&\leq C \Big\| \frac{q_{\psi_{\mathcal{G}_N (s)}}}{ \mathcal{G}_N (s)  - E_N (s)} ( \mathcal{N}_+ +1 )^{1/2} \xi \| \; \|  ( \mathcal{N}_+ +1 ) \frac{q_{\psi_{\mathcal{G}_N (s)}}}{ \mathcal{G}_N (s)  - E_N (s)} \xi \|  \notag \\
&\leq C \| \xi \| \; \|  ( \mathcal{N}_+ +1 ) \frac{q_{\psi_{\mathcal{G}_N (s)}}}{ \mathcal{G}_N (s)  - E_N (s)} \xi \| \; . \label{eq:estimate-resolvent6}
\end{align}
Thus we conclude with \eqref{eq:estimate-resolvent5}, \eqref{eq:estimate-resolvent6} from \eqref{eq:estimate-resolvent4} with the operator inequality 
\begin{align}
( 1- \frac{1}{2}) \frac{q_{\psi_{\mathcal{G}_N (s)}}}{ \mathcal{G}_N (s)  - E_N (s) } ( \mathcal{N}_+ + 1 )^2 \frac{q_{\psi_{\mathcal{G}_N (s)}}}{ \mathcal{G}_N (s)  - E_N (s) }  \leq C 
\end{align}
that finally leads to the desired first bound of \eqref{eq:resolvent-1}. 

The second bound follows with similar arguments from \eqref{eq:comm-G-N12}.

\end{proof}

\subsection{Generalized Bogoliubov transform} We note that the quadratic Hamiltonian $\mathcal{Q}$ is formulated w.r.t. to modified creation and annihilation operators. For operators quadratic in standard creation and annihilation operators, the corresponding unique ground state is explicitly known and given by a quasi-free state. However, here we do not have an explicit expression for the ground state $\psi_{\mathcal{Q}}$, but we will use that it is approximately given by the generalized quasi-free state $e^{B( \mu)} \Omega $ as defined in \eqref{def:Bogo-b}. In contrast to the standard Bogoliubov transform \eqref{def:Bogo-a} formulated w.r.t. to standard creation and annihilation operators, there is no exact formula for the action of $e^{B( \mu)}$ on modified creation and annihilation operators. However, we have 
\begin{align}
\label{def:d}
e^{B ( \mu)} b_p^* e^{-B ( \mu)} =& \gamma_p b_p^* + \sigma_p b_{-p} + d_p, \notag \\
 e^{B ( \mu)} b_p e^{-B ( \mu)} =& \gamma_p b_p +\sigma_p b_{-p}^* + d^*_p
\end{align}
where we write $\sigma_p = \sinh ( \mu_p)$ and $\gamma_p = \cosh ( \mu_p)$. The remainders $d_p, d_p^*$ are small on states with a small number of excitations. More precisely, \cite[Lemma 2.3]{BBCS_Acta} shows (since $\mu \in \ell^2( \Lambda_+^*)$)  that for any $k \in \mathbb{Z}$ there exists $C_k >0$ such that 
\begin{align}
\label{eq:estimates-dp}
\| ( \mathcal{N}_+  + 1)^{k/2} d_p \psi \| \leq C_k N^{-1} \left( \| b_p ( \mathcal{N}_+ 1)^{(k+2)/2} \psi \| + \vert \mu_p \vert \; \| ( \mathcal{N}_+  + 1)^{3/2} \psi \| \right) 
\end{align}
for all $p \in \Lambda_+^*$ and 
\begin{align}
\label{eq:estimates-dp*}
\| ( \mathcal{N}_+  + 1)^{k/2} d_p^* \psi \| \leq C_k N^{-1} \| ( \mathcal{N}_+  + 1)^{3/2} \psi \|
\end{align}
In particular this leads to 
\begin{align}
\label{eq:estimates-db}
\| ( \mathcal{N}_+ + 1)^{-1/2}  d_p^{\sharp_1} b_{\alpha p}^{\sharp_2} \psi \| , \; \| ( \mathcal{N}_+ + 1)^{-1/2}  b_p^{\sharp_1} d_{\alpha p}^{\sharp_2} \psi \|  \leq  C N^{-1/2} \| ( \mathcal{N}_+ + 1) \psi \|    
\end{align}
and 
\begin{align}
\label{eq:estimates-dd}
\| ( \mathcal{N}_+ + 1)^{-1/2}  d_p^{\sharp_1} d_{\alpha p}^{\sharp_2} \psi \| \leq  C N^{-1/2} \| ( \mathcal{N}_+ + 1) \psi \|
\end{align}
with $\sharp_i \in \lbrace \cdot, * \rbrace$ for $i=1,2$, either $\sharp_1 = \sharp_2$ or $\sharp_1 = * $ and $\sharp_2 = \cdot$ and $\alpha = -1$ if $\sharp_1 = \sharp_2$ and $\alpha =1$ otherwise. These estimates \eqref{eq:estimates-dp}, \eqref{eq:estimates-dp*}, \eqref{eq:estimates-dd} remain true when replacing $d_p, d_p^*$ resp. $d_p^{\sharp_1}d_{\alpha p}^{\sharp_2}$ with their (double commutator) with $\cN_+$; we have 
\begin{align}
\label{eq:bound-comm}
\| ( \mathcal{N}_+  + 1)^{k/2} \left[ \mathcal{N}_+,  d_p \right]  \psi \| \leq C_k  N^{-1} \left( \| b_p ( \mathcal{N}_+ 1)^{(k+2)/2} \psi \| + \vert \mu_p \vert \; \| ( \mathcal{N}_+  + 1)^{3/2} \psi \| \right) 
\end{align} 
resp. 
\begin{align}
\label{eq:bound-double-comm}
\| ( \mathcal{N}_+  + 1)^{k/2} \left[ \mathcal{N}_+, \left[ \mathcal{N}_+,  d_p \right] \right] \psi \| \leq C_k  N^{-1} \left( \| b_p ( \mathcal{N}_+ 1)^{(k+2)/2} \psi \| + \vert \mu_p \vert \; \| ( \mathcal{N}_+  + 1)^{3/2} \psi \| \right) 
\end{align}
and similarly for the other operators (note that \eqref{eq:bound-comm}, \eqref{eq:bound-double-comm} follow from the proof of \cite[Corollary 3.5]{BBCS_optimal}). Also, we know the generalized Bogoliubov transform approximate action on the kinetic term that is given by 
\begin{align}
e^{B( \mu)} \sum_{p \in \Lambda_+^*}&  p^2 a_p^*a_p e^{-B( \mu)} \notag \\
&= \sum_{p \in \Lambda_+^*} p^2 a_p^*a_p + \sum_{p \in \Lambda_+^*} p^2 \left( \sigma_p^2 + \sigma_p \gamma_p ( b_p^*b_{-p}^* + b_pb_{-p} ) + 2 \sigma_p^2 b_p^*b_p \right)  + \mathcal{R}_{\mathcal{K}}
\end{align}
where the remainder $\mathcal{R}_{\mathcal{K}}$ satisfies 
\begin{align}
\|( \mathcal{N}_+ + 1)^{-1/2}   \mathcal{R}_{\mathcal{K}} \psi \| \leq C  N^{-1/2} \| ( \mathcal{N}_+ + 1) \psi \| \; . \label{eq:bound-RK1}
\end{align}
and also similar bounds for its (double) commutator as formulated before \eqref{eq:bound-comm}, \eqref{eq:bound-double-comm}. Note that since $p^2 \mu_p \in \ell^1( \Lambda)$ and $\sigma_p, \gamma_p \in \ell^\infty( \Lambda_+^*)$, this is a consequence of \eqref{eq:estimates-db}, \eqref{eq:estimates-dd} (for more details see Lemma \ref{lemma:RK} below). Consequently, conjugating the quadratic Hamiltonian $\mathcal{Q}$ with the generalized Bogoliubov transform $e^{B( \mu)}$ almost diagonalizes $\mathcal{Q}$. More precisely, we have 
\begin{align}
e^{-B(\mu)} \mathcal{Q} e^{B(\mu)} = \mathcal{D} + \mathcal{R}_{\mathcal{Q}} 
\end{align}
where the diagonal operator $\mathcal{D}$ is given by 
\begin{align}
\mathcal{D} := \sum_{p \in \Lambda_+^*} \left( \left( p^2 + \widehat{v} (p) \right) ( \sigma_p^2 + \gamma_p^2) + 2 \sigma_p \gamma_p \widehat{v} (p) \right) a_p^*a_p 
\end{align}
and the remainder is 
\begin{align}
\mathcal{R}_{\mathcal{Q}}  =& \sum_{p \in \Lambda_+^*} (2 p^2  \sigma_p^2 + \widehat{v}(p) ) (b_p^*b_p -a_p^*a_p)  \notag \\
&+ \sum_{p \in \Lambda_+^*} \widehat{v}(p) \left( ( \gamma_p b_p^* + \sigma_p b_{-p} ) d_p + d_p^* ( \gamma_p b_p + \sigma_p b_{-p}^* ) + d_p^*d_p \right) \notag \\
&+ \sum_{p \in \Lambda_+^*} \widehat{v}(p)  \left( ( \gamma_p b_p^* + \sigma_p b_{-p} ) d_{-p}^* + d_p^* ( \gamma_p b_{-p}^* + \sigma_p b_{p} ) + d_p^*d_{-p}^* ) \right) +{\rm h.c.} \notag \\
&+ \mathcal{R}_{\mathcal{K}} \; . 
\end{align}
Though we do not have an explicit form of $e^{B(\mu)} \psi_{\mathcal{Q}}$, the ground state of the diagonal operator $\mathcal{D}$ is explicitly known and given by the vacuum $\Omega$. For this reason we will study for $s \in [0,1]$ the family of Hamiltonians
\begin{align}
\mathcal{Q} (s) := \mathcal{D} + s \mathcal{R}_{\mathcal{K}}
\end{align}
interpolating between $\mathcal{Q}(1) = e^{-B( \mu)}\mathcal{Q} e^{B(\mu)}$ and $\mathcal{Q}(0) = \mathcal{D}$. 

Similarly to Proposition \ref{claim:gs} we have the following properties. 

\begin{proposition}
\label{claim:Qs}
Let $s \in [0,1]$, then there exists a ground state $\psi_{\mathcal{Q}(s)}$ of the Hamiltonian $\mathcal{Q} (s)$ defined in \eqref{def:Gs}. Furthermore, there exists a constant $C_k >0$ (independent of $s,N$) such that 
\begin{align}
\langle \psi_{\mathcal{Q} (s)}, \;  (\mathcal{N}_+ +1)^k \psi_{\mathcal{Q} (s)} \rangle \leq C_k \label{eq:N-gs}
\end{align}
for $k=1,2$ and the spectrum of the Hamiltonian $\mathcal{Q}(s)$ has a spectral gap above the ground state $E (s)$  independent of $s,N$. Moreover, for $k=1,2$ there exists $C_k>0$ (independent of $N,s$) such that for any Fock space vector $\xi \in \mathcal{F}_{\perp \varphi_0}^{\leq N}$ we have 
\begin{align}
\label{eq:resolvent-1}
\Big\|(\mathcal{N}_+ +1)^{k/2} \frac{q_{\psi_{\mathcal{Q} (s)}}}{ \mathcal{Q} (s) - E (s)} \;  \;  \xi \Big\| &\leq  C_k \| \xi \| \; . 
\end{align} 
\end{proposition}

\begin{proof}
We proceed similarly as in the proof of Proposition \ref{claim:gs}. First note that from Proposition \ref{claim:gs} we have (since $\mathcal{Q} = \mathcal{G}_N(0)$) that 
\begin{align}
\mathcal{Q}(1) =  e^{-B( \mu)} (\mathcal{Q} - E(1)) e^{B( \mu)} \geq C_1 e^{-B( \mu)}\mathcal{N}_+e^{B( \mu)} - C_2
\end{align}
for some $C_1,C_2 >0$. The generalized Bogoliubov transform approximately preserves the number of particles. More precisely it follows from \cite[Lemma 2.4]{BBCS_Acta} that 
\begin{align}
e^{-B( \mu)} \mathcal{Q} e^{B( \mu)} \geq C_3 \mathcal{N}_+  - C_4
\end{align}
for some positive constants $C_3,C_4>0$. Since $\mathcal{D} \geq C \mathcal{N}_+ $ for some positive $C>0$ and $\mathcal{Q}(s) = s e^{-B( \mu)} \mathcal{Q} e^{B( \mu)} + (1-s) \mathcal{D}$ is a convex combination of both, we find 
\begin{align}
(\mathcal{Q}(s) - E(s)) \geq c_0 \mathcal{N}_+ - C_0 
\end{align}
for some positive constants $c_0,C>0$. This implies that for any normalized $\xi \in \mathcal{F}$ with $\xi = \mathds{1}_{\mathcal{Q}(s) \leq \zeta} \xi$ that 
\begin{align}
\langle \xi,  (\mathcal{N} + 1)^2 \xi \rangle \leq& C \langle \xi, \; (\mathcal{N}_+ +1)^{1/2} \mathcal{Q} (s) (\mathcal{N}_+ +1)^{1/2} \xi \rangle \notag \\
=& C \zeta \langle \xi, \; (\mathcal{N}_+ +1) \xi \rangle + C \langle \xi, \; \left[ (\mathcal{N}_+ +1)^{1/2}, \left[ \mathcal{Q} (s) , (\mathcal{N}_+ +1)^{1/2} \right] \right] \xi \rangle \;.
\end{align}
With formula \ref{eq:spec-calc} we write the commutator as 
\begin{align}
& \left[ (\mathcal{N}_+ +1)^{1/2}, \left[ \mathcal{Q} (s) , (\mathcal{N}_+ +1)^{1/2} \right] \right] \notag \\
 &=\frac{1}{\pi^2}\int_0^\infty dt_1 dt_2 \frac{1}{\mathcal{N}_+ + 1+ t_1}\frac{1}{\mathcal{N}_+ + 1+ t_2} \left[ \cN_+ , \; \left[ \cN_+, \mathcal{Q} (s) \right] \right] \frac{1}{\mathcal{N}_+ + 1+ t_2}\frac{1}{\mathcal{N}_+ + 1+ t_1}
\end{align}
Since $\widehat{v} \in \ell^2( \Lambda_+^*)$ and $\left[ \mathcal{N}_+, b_p^{\sharp}\right] = \alpha b_p^\sharp$ with $\sharp \in \lbrace \cdot, * \rbrace$ and $\alpha = -1$ if $\sharp  = \cdot $ and $\alpha =1$ if $\sharp = *$, it follows from \eqref{eq:bounds-b} and \eqref{eq:bound-double-comm}  that we can bound the double commutator in form by the number of particles. Thus we arrive for any $\xi \in \mathcal{F}_{\perp \varphi_0}^{\leq N}$ at 
\begin{align}
\langle \xi,  (\mathcal{N} + 1)^2 \xi \rangle \leq& C \zeta \langle \xi, \; (\mathcal{N}_+ +1) \xi \rangle + C \|( \mathcal{N} + 1) \xi \| \| \xi \| 
\end{align}
and we conclude by  $\langle \xi,  (\mathcal{N} + 1)^2 \xi \rangle \leq C$. The spectral gap and the bound on the resolvent follow with similar arguments as in the proof of Proposition \ref{claim:gs} using again the estimates on the double commutator. 

\end{proof}

\section{Preliminaries}
\label{sec:prel}

The proof of Theorem \ref{thm:main} is based on closed formulas derived in \cite{KRS} for the conjugation of operators of the form $b_p, b_p^*$ and 
\begin{align}
d \Gamma (H) = \sum_{p \in \Lambda_+^*} H_{p,q} \; a_p^*a_{-q} 
\end{align}
for any bounded operator $H$ on $\ell^2( \Lambda_+^*)$ with the exponential of $\mathcal{N}_+$ (given by \eqref{def:N+}) and the symmetric operator 
\begin{align}
\label{def:phi}
\phi_+ (h) = b^* (h) + b(h) = \sum_{p \in \Lambda_+^*} h_p \left[ b_p^* + b_{-p} \right] 
\end{align}
with $h \in \ell^2 ( \Lambda_+^*)$.  For this we furthermore define for any $h \in \ell^2 ( \Lambda_+^*)$ the anti-symmetric operator 
\begin{align}
i \phi_- (h)  = b(h) - b^*(h) = - \sum_{p \in \Lambda_+^*} h_p \left[  b_p^*- b_{-p} \right] \; . 
\end{align}
and (in abuse of notation) the shorthand notation
\begin{align}
\label{eq:gamma}
 \gamma_s = \cosh (s)  \quad \text{and} \quad \sigma_s = \sinh (s)
\end{align}

We recall the closed formulas from \cite{KRS} that are formulated in position space and easily translate with \eqref{def:ap} to momentum space relevant for the present analysis. 

\begin{lemma}[Proposition 2.2,2.4 in \cite{KRS}]\label{lemma:b-con}
With the shorthand notation $\| \cdot \|_{\ell^2( \Lambda_+^*)} = \| \cdot \|$ we have for $h \in \ell^2 (\Lambda_+^*)$ and  $p \in \Lambda_+^*$
\begin{equation}
\label{eq:prop1-1} 
\begin{split} 
e^{\sqrt{N} \phi_+ (h)} b_p e^{-\sqrt{N} \phi_+(h)} =&\; \gamma_{\| h \|} b_p + \gamma_{\| h \|} \frac{\gamma_{\| h \|} - 1}{\| h \|^2} h_{-p} i \phi_- (h) - \frac{\gamma_{\| h \|} - 1}{\| h \|^2} h_{-p}  b^* (h) \\ &- \sqrt{N} \, \gamma_{\|h \|} \frac{\sigma_{\| h \|}}{\| h \|}  h_{-p} \left( 1 - \frac{\N}{N} \right) + \frac{1}{\sqrt{N}} \, \frac{\sigma_{\| h \|}}{\| h \|} \frac{\gamma_{\| h \|} - 1}{ \| h \|^2}  h_{-p} a^* (h) a (h) \\ &+ \frac{1}{\sqrt{N}} \, \frac{\sigma_{\| h \|}}{\| h \|} a^* (h) a_{p}  \; .
\end{split} 
\end{equation}
Furthermore for any self-adjoint $H : \mathcal{D} ( H) \rightarrow \ell^2( \Lambda_+^*) $ with domain $\mathcal{D} (H) \subset \ell^2( \Lambda_+^*)$ and $h \in \mathcal{D} (H)$ we have 
\begin{equation}\label{eq:prop1-2}
\begin{split} 
& e^{\sqrt{N} \phi_+(h)} \rd \Gamma (H) e^{-\sqrt{N} \phi_+(h)} \notag \\
& \quad= \; \rd \Gamma (H) + \sqrt{N} \, \frac{\sigma_{\| h \|}}{\| h \|} i \phi_- (Hh) \\
 &\quad \quad - N \frac{\sigma_{\| h \|}^2}{\|h \|^2} \langle h, H h \rangle \left(1 - \frac{\N}{N} \right) + \frac{(\gamma_{\| h \|} -1)}{\| h \|^2} (a^* (h) a (Hh) + a^* (Hh) a (h) )
\\ &\quad \quad + \sqrt{N} \,  \frac{\sigma_{\| h \|}}{\| h \|} \frac{\gamma_{\| h \|} - 1}{\| h \|^2} \langle h, H h \rangle i \phi_- (h)  + \left( \frac{\gamma_{\| h \|} - 1}{\| h \|^2} \right)^2  \langle h, H h \rangle a^* (h) a (h)  \; . \end{split}\end{equation}
\end{lemma}

A similar formula as \eqref{eq:prop1-1} for  $e^{\sqrt{N} \phi_+ (h)} b^*_p e^{-\sqrt{N} \phi_+(h)}$ follows when replacing $h$ with its negative $-h$ and taking the hermitian conjugate of \eqref{eq:prop1-1}. 

Furthermore the following closed formulas hold for the conjugation with respect to the exponential of the number of particles operator on the excitation Fock space $\mathcal{N}_+$.

\begin{lemma}[Proposition 2.5 \cite{KRS}]
\label{prop:eN}
Let $\cN_+$ be given by \eqref{def:N+} and $h \in \ell^2 ( \Lambda_+^+)$. Then for every  $s \in \bR$ we have with the short hand notation \eqref{eq:gamma} 
\begin{align}\label{eq:prop2} 
e^{-s \cN_+} b (h) e^{s \cN_+} &= e^s  b(h),  \notag\\
e^{-s \cN_+} b^* (h) e^{s \cN_+} &= e^{-s} b^* (h) , \notag \\
e^{-s \cN_+} \phi_+ (h) e^{s \cN_+} &= \gamma_s \phi_+ (h) + \sigma_s i \phi_- (h), \notag\\
e^{-s \cN_+} i \phi_- (h) e^{s \cN_+} &= \gamma_s i \phi_- (h) + \sigma_s \phi_+ (h) \; . 
\end{align}
\end{lemma} 
 
Moreover we shall use the following Lemma proven in \cite{KRS}. 

\begin{lemma}[Proposition 2.6 \cite{KRS}] 
\label{prop:partialt} 
Let $h_\cdot : \mathbb{R} \rightarrow \ell^2( \Lambda_+^*), \;  t \mapsto h_t$ be a differentiable. For $\xi_1, \xi_2 \in \cF_{\perp \ph_0}^{\leq N}$ we find with the short hand notation \eqref{eq:gamma} 
\begin{align}\label{eq:partialt} 
&\Big\langle \xi_1 , \Big[ \partial_t e^{\sqrt{N} \phi_+ (h_t)} \Big] e^{-\sqrt{N} \phi_+ (h_t)} \xi_2 \Big\rangle \notag \\ 
&\quad = \; \sqrt{N} \, \frac{\sigma_{\| h_t \|}}{\| h_t \|}  \langle \xi_1 , \phi_+ ( \partial_t h_t) \xi_2 \rangle - 
\sqrt{N} \frac{\sigma_{\| h_t \|}}{\| h_t \|} \frac{\gamma_{\| h_t \|} -1}{\| h_t \|^2} \Im \langle \partial_t h_t , h_t \rangle  \langle \xi_1, \phi_- (h_t) \xi_2 \rangle 
\notag \\ 
&\quad \quad- \sqrt{N} \, \frac{\sigma_{\| h_t \|} - \| h_t \|}{\| h_t \|^3} \Re \langle \partial_t h_t , h_t \rangle \langle \xi_1, \phi_+ (h_t) \xi_2 \rangle - i N \frac{\sigma_{\| h_t \|}^2}{\| h_t \|^2} \Im \langle \partial_t h_t , h_t \rangle  \langle \xi_1 , (1 - \cN_+ / N) \xi_2 \rangle \notag \\ 
&\quad \quad +i  \left(\frac{\gamma_{\| h_t \|} - 1}{\| h_t \|^2} \right)^2 \Im \langle \partial_t h_t , h_t \rangle \langle \xi_1 , a^* (h_t ) a(h_t) \xi_2 \rangle \notag \\ 
&\quad\quad + 
 \frac{\gamma_{\| h_t \|} - 1}{\| h_t \|^2} \Big\langle \xi_1, \left[ a^* (h_t) a (\partial_t h_t) - a^* (\partial_t h_t) a (h_t) \right] \xi_2 \Big\rangle \; .
 \end{align}
\end{lemma} 

In the proof of the main theorem we consider operators conjugated w.r.t. to both exponentials $e^{\lambda \sqrt{N} \phi_+ (h)} e^{\lambda \cN_+}$ where the parameter $\lambda \in [0,1]$ is considered to be small. The previous Lemma yield in the following Corollary for the first order contributions. 

\begin{corollary}
\label{lemma:bb-conj}
Let $H \in \ell^2 ( \Lambda_+^*)$, $g \in \ell^2( \Lambda_+^*)$ and $\vert\lambda\vert, \vert \lambda \kappa \vert \leq 1$. Then, there exists $C>0$ (independent of $\kappa, \lambda$) such that 
\begin{align}
\|  \; \sum_{p \Lambda_+^*} H_{p}   & \left( e^{\kappa \lambda \cN_+} e^{\lambda \sqrt{N} \phi_+ (g)}  b_p^{\sharp_1} b_{\alpha p}^{\sharp_2}  e^{-\lambda \sqrt{N} \phi_+ (g)} e^{-\kappa \lambda \cN_+}  - b_p^{\sharp_1} b_{\alpha p}^{\sharp_2} \right)  \psi \|  \notag \\
&\leq C (\| g \|+ \vert \kappa\vert )  \vert \lambda \vert \sqrt{N} \| ( \mathcal{N}_+ +1)^{1/2} \psi \| + C \| g \|^2 \lambda^2 N \| \psi \| \label{eq:formula1}   
\end{align}
with $\sharp_i \in \lbrace  \cdot, * \rbrace $ for $i = 1,2$, either $\sharp_1 = \sharp_2$ and $\alpha = -1$ or $\sharp_1 = *, \sharp_2 = \cdot$ and  $\alpha =  1 $ otherwise. 
Furthermore, for $p^2 g \in \ell^2( \Lambda_+^*)$ there exists $C>0$ (independent of $\kappa,\lambda$) such that 
\begin{align}
\|  \sum_{p \Lambda_+^*} p^2 &   \left( e^{\kappa \lambda \cN_+} e^{\lambda \sqrt{N} \phi_+ (g)}  a_p^{*} a_{p}  e^{-\lambda \sqrt{N} \phi_+ (g)} e^{-\kappa \lambda \cN_+}  - a_p^* a_p \right) \psi \|  \notag \\
& \leq C \| p^2g \| \vert\lambda\vert \sqrt{N} \| ( \mathcal{N}_+ +1)^{1/2} \psi \| + C \| p^2g \|^2 \lambda^2 N \| \psi \|  \; . \label{eq:formula2}
\end{align}

\end{corollary}

\begin{proof}
We consider the case $\sharp_1 = *, \sharp_2 = \cdot$ and $\alpha =1$ first. 
We recall that by definition of the modified creation and annihilation operators in \eqref{def:b} we have 
\begin{align}
\sum_{p \in \Lambda_+}H_p \;  b_p^*b_{p} = \frac{N - \mathcal{N}_+ +1}{N} \sum_{p \in \Lambda_+} H_{p} \;  a_p^*a_p = \frac{N - \mathcal{N}_+ + 1}{N} \;  d \Gamma (H) 
\end{align}
and thus 
\begin{align}
\sum_{p \in \Lambda_+} H_{p}  & \; \left( e^{\kappa \lambda \cN_+} e^{\lambda \sqrt{N}\phi_+ (g)}  b_p^*b_{p} e^{ - \lambda \sqrt{N} \phi_+ (g)} e^{- \lambda \kappa \cN_+} - b_p^* b_{p}  \right) \notag \notag \\
=& \frac{N- \mathcal{N}_+ + 1}{N} \; \left(  e^{\kappa \lambda \cN_+} e^{\lambda \sqrt{N}\phi_+ (g)}    d \Gamma (H) e^{ - \lambda \sqrt{N} \phi_+ (g)} e^{- \lambda \kappa \cN_+} - d \Gamma (H) \right) \notag \\
&-  \left( e^{\kappa \lambda \cN_+} e^{\lambda \sqrt{N}\phi_+ (g)}  \frac{\mathcal{N}_+ }{N} e^{ - \lambda \sqrt{N} \phi_+ (g)} e^{- \lambda \kappa \cN_+} - \frac{\cN_+}{N} \right) d \Gamma (H) \; . 
\end{align}
From Lemma \ref{lemma:b-con} we find that 
\begin{align}
\sum_{p \in \Lambda_+} &  H_{p}   \; \left( e^{\kappa \lambda \cN_+} e^{\lambda \sqrt{N}\phi_+ (g)}  b_p^*b_{p} e^{ - \lambda \sqrt{N} \phi_+ (g)} e^{- \lambda \kappa \cN_+} - b_p^* b_{-q}\right) \notag \\
=&\frac{N- \mathcal{N}_+ + 1}{N} \left[ \lambda \gamma_{\kappa \lambda} \sqrt{N} \frac{\sigma_{\lambda\| g \|}}{\lambda\| g \|} i \phi_- (Hg)  + \lambda \sigma_{\kappa \lambda} \sqrt{N} \frac{\sigma_{\lambda \| g \|}}{\lambda \| g \|}  \phi_+ (Hg) \right.  \notag \\
 &\hspace{0.8cm}  - N  \lambda^2 \frac{\sigma_{\lambda\| g \|}^2}{(\lambda\|g \|)^2} \langle g, H g \rangle \left(1 - \frac{\N}{N} \right) + \lambda^2 \frac{(\gamma_{\lambda\| g \|} -1)}{(\lambda\| g \|)^2} (a^* (g) a (Hg) + a^* (Hg) a (g) ) 
\notag \\ 
&\hspace{0.8cm} + \sqrt{N} \lambda^3 \gamma_{\kappa \lambda} \frac{\sigma_{\lambda\| g \|}}{\lambda\| g \|} \frac{\gamma_{\lambda\| g \|} - 1}{(\lambda\| g \|)^2} \langle g, H g \rangle i \phi_- (g)  + \lambda^3 \sqrt{N} \, \sigma_{\kappa \lambda} \frac{\sigma_{\lambda\| g \|}}{\lambda\| g \|} \frac{\gamma_{\lambda\| g \|} - 1}{(\lambda\| g \|)^2} \langle g, H g \rangle \phi_+ (g)  \notag \\
& \hspace{0.8cm} \left.  +\lambda^4 \left( \frac{\gamma_{\lambda\| g \|} - 1}{(\lambda\| g\|)^2} \right)^2  \langle g, H g \rangle a^* (g) a (g)  \right] \; \notag \\
&+ \frac{1}{N}\left[  \lambda\gamma_{\kappa \lambda} \frac{\sigma_{\lambda\| g \|}}{\lambda\| g \|} i \phi_- (g)  + \lambda\sigma_{\kappa \lambda} \frac{\sigma_{\lambda\| g \|}}{\lambda\| g \|}  \phi_+ (g) \right.  \notag \\
 &\hspace{0.8cm} - N \lambda^2 \frac{\sigma_{\lambda\| g \|}^2}{(\lambda\|g \|)^2} \| g \|^2 \left(1 - \frac{\N}{N} \right) +  2 \lambda^2 \frac{(\gamma_{\lambda\| g \|} -1)}{(\lambda\| g \|)^2} a^* (g) a (g) 
\notag \\ 
&\hspace{0.8cm} + \sqrt{N} \lambda^3 \, \gamma_{\kappa \lambda} \frac{\sigma_{\lambda\| g \|}}{\lambda\| g \|} \frac{\gamma_{\lambda\| g \|} - 1}{(\lambda\| g \|)^2} \| g \|^2 i \phi_- (g)  +  \sqrt{N}\lambda^3 \, \sigma_{\kappa \lambda} \frac{\sigma_{\lambda\| g \|}}{\lambda\| g \|} \frac{\gamma_{\lambda\| g \|} - 1}{(\lambda\| g \|)^2} \|g \|^2 \phi_+ (g)  \notag \\
& \hspace{0.8cm} \left. + \lambda^4 \left( \frac{\gamma_{\lambda\| g \|} - 1}{(\lambda\| g\|)^2} \right)^2  \| g \|^2 a^* (g) a (g)  \right] d \Gamma (H) \; . \label{eq:dGamma-comm}
\end{align}
Since $g, Hg \in \ell^2( \Lambda_+^*)$ we find with \eqref{eq:bounds-a},\eqref{eq:bounds-b} for any $\psi \in \mathcal{F}_{\perp \varphi}^{\leq N}$ and $\vert\lambda \vert < 1$ that 
\begin{align}
\| \sum_{p \in \Lambda_+}    H_{p} &    \; \left( e^{\kappa \lambda \cN_+} e^{\lambda \sqrt{N}\phi_+ (g)}  b_p^*b_{p} e^{ - \lambda \sqrt{N} \phi_+ (g)} e^{- \lambda \kappa \cN_+} - b_p^* b_{-q}\right) \psi \| \notag \\
&\leq C \lambda \| g \| \sqrt{N} \| ( \mathcal{N}_+ +1)^{1/2} \psi \| + C \lambda^2 \|g \|^2 N \| \psi \|   \; . \label{eq:bound-1-bb*}
\end{align}
The remaining cases for \eqref{eq:formula1} (i.e. $\sharp_1 = \sharp_2$) follow similarly from Lemma \ref{lemma:b-con}. We have 
\begin{align}
\| \sum_{p \in \Lambda_+}    H_{p} &    \; \left( e^{\kappa \lambda \cN_+} e^{\lambda \sqrt{N}\phi_+ (g)}  b_p^*b^*_{-p} e^{ - \lambda \sqrt{N} \phi_+ (g)} e^{- \lambda \kappa \cN_+} - b_p^* b_{-q}\right) \psi \| \notag \\
&\leq C \lambda  (\| g \|  +  \kappa) \sqrt{N} \| ( \mathcal{N}_+ +1)^{1/2} \psi \| + C \lambda^2 \|g \|^2 N \| \psi \|   \; .
\end{align}
Note that the bound linear in $\lambda$ depends on $\kappa$ (in contrast to \eqref{eq:bound-1-bb*}) as $b_p^*b_{-p}^*$ does not commute with $e^{\lambda \kappa \cN_+}$. 

The second estimate \eqref{eq:formula2} follows with similar arguments from Lemma \ref{lemma:b-con} since 
\begin{align}
\sum_{p \in \Lambda_+} &  p^2   \; \left( e^{\kappa \lambda \cN_+} e^{\lambda \sqrt{N}\phi_+ (g)}  a_p^*a_{p} e^{ - \lambda \sqrt{N} \phi_+ (g)} e^{- \lambda \kappa \cN_+} - a_p^* a_{p}\right) \notag \\
=&\frac{N- \mathcal{N}_+ + 1}{N} \left[ \lambda \gamma_{\kappa \lambda} \sqrt{N} \frac{\sigma_{\lambda\| g \|}}{\lambda\| g \|} i \phi_- (\widetilde{g})  + \lambda \sigma_{\kappa \lambda} \sqrt{N} \frac{\sigma_{\lambda\| g \|}}{\lambda\| g \|}  \phi_+ (\tilde{g}) \right.  \notag \\
 &\hspace{2.2cm}  - N  \lambda^2 \frac{\sigma_{\lambda\| g \|}^2}{(\lambda\|g \|)^2} \langle g, H g \rangle \left(1 - \frac{\N}{N} \right) + \lambda^2 \frac{(\gamma_{\lambda\| g \|} -1)}{(\lambda\| g \|)^2} (a^* (g) a (\tilde{g}) + a^* (\tilde{g}) a (g) ) 
\notag \\ 
&\hspace{2.2cm} + \sqrt{N} \lambda^3 \gamma_{\kappa \lambda} \frac{\sigma_{\| g \|}}{\| g \|} \frac{\gamma_{\lambda\| g \|} - 1}{(\lambda\| g \|)^2} \langle g,  \tilde{g} \rangle i \phi_- (g)  + \lambda^3 \sqrt{N} \, \sigma_{\kappa \lambda} \frac{\sigma_{\lambda\| g \|}}{\lambda\| g \|} \frac{\gamma_{\lambda\| g \|} - 1}{(\lambda\| g \|)^2} \langle g,  \tilde{g} \rangle \phi_+ (g)  \notag \\
& \hspace{2.2cm} \left.  +\lambda^4 \left( \frac{\gamma_{\lambda\| g \|} - 1}{(\lambda\| g\|)^2} \right)^2  \langle g, H g \rangle a^* (g) a (g)  \right] \; \notag \\
\end{align}
where we denoted $\tilde{g}(p) = p^2 g (p)$. Since $\tilde{g} \in \ell^2( \Lambda_+^*)$ by assumption, \eqref{eq:formula2} follows with similar arguments. 
\end{proof}

For our analysis we need to improve those bounds and prove similar bounds for the conjugated operators
 \begin{align}
e^{\lambda \kappa \mathcal{N}_+} e^{\lambda \sqrt{N} \phi_+ (h)}\; d_p^{\sharp} \; e^{-\lambda \sqrt{N} \phi_+ (h)}e^{-\lambda \kappa \mathcal{N}_+}  \label{eq:estimates-dp-start}
 \end{align}
with $\sharp \in \lbrace \cdot, * \rbrace$ and $\kappa, \lambda \in \mathbb{R}$. For this, we are using closed formulas derived in \cite[Proposition 2.3-2.6]{KRS} and properties of $d_p,d_p^*$ from \cite{BS,BBCS_Acta} based on the expansion for any $p \in \Lambda_+^*$
\begin{align}
e^{-B(\mu)} b_p e^{B( \mu)} =& \sum_{n=1}^{m-1} (-1)^n \frac{{\rm ad}_{B(\mu)}^{(n)} (b_p)}{n!} \notag \\
 &+ \int_0^1 ds_1 \int_0^{s_1}  ds_2\dots \int_0^{s_{m-1}} ds_m e^{-s_m B( \mu)}{\rm ad}_{B( \mu)}^{(m)} (b_p) e^{s_m B( \mu)}
\end{align}
 with the recursive definition for the nested commutators 
\begin{align}
{\rm ad}_{B ( \mu)}^{(0)} (A) = A \quad \text{and} \quad {\rm ad}_{B( \mu)}^{(n)} = \left[ B( \mu), {\rm ad}_{B (\mu)}^{(n-1)} (A) \right] \; . 
\end{align}
In \cite{BS} it is shown that the nested commutators of $b_p,b_p^*$ are given in terms of the following operators: For $f_1, \dots , f_n \in \ell^2( \Lambda_+^*)$, $\sharp = ( \sharp_1, \dots, \sharp_n), \flat = (\flat_0, \dots, \flat_{n-1}) \in \lbrace \cdot, * \rbrace^n $ we define the $\Pi^{(2)}$-operator of order $n $ by 
\begin{align}
\label{def:Pi2}
\Pi_{\sharp, \flat}^{(2)} (f_1, \dots, f_n ) = \sum_{p_1, \dots, p_n \in \Lambda_+^*} b_{\alpha_0 p_1}^{\flat_0} a_{\beta_1 p_1}^{\sharp_1} a_{\alpha_1 p_2}^{\flat_1} a_{\beta_2 p_2}^{\sharp_2} a_{\alpha_2 p_3}^{\flat_2} \dots a_{\beta_{n-1} p_{n-1}}^{\sharp_{n-1}} a_{\alpha_{n-1} p_n}^{\flat_{n-1}} b_{\beta_n p_n}^{\sharp_n} \prod_{\ell=1}^n f_\ell (p_\ell ) 
\end{align}
were for $\ell = 0,1, \dots, n $ we define $\alpha_\ell =1$ if $\flat_\ell =*$., $\alpha_\ell = -1$ if $\flat_\ell = \cdot$, $\beta_\ell =1$ if $\sharp_\ell =\cdot$ and $\beta_\ell = -1$ of $\sharp_\ell = *$. Moreover, we require that for every $j=1, \dots, n-1$ we have either $\sharp_j = \cdot$ and $\flat_j = *$ or $\sharp_j = *$ and $\flat_j = \cdot$ (so that the product $a_{\beta_\ell p_\ell}^{\sharp_\ell} a_{\alpha_\ell p_{\ell + 1}}^{\flat_\ell}$ preserves the number of particles for all $\ell =1, \dots, n-1$). Then, the operator $\Pi_{\sharp, \flat}^{(2)} (f_1, \dots, f_n ) $ leaves the truncated Fock space invariant. Moreover if for some $\ell =1, \dots, n$, $\flat_{\ell -1} = \cdot$ and $\sharp_\ell =*$, we furthermore require that $f_\ell \in \ell^1 ( \Lambda_+^*)$ (so that we can normal order the operators).  
For $g, f_1, \dots, f_n \in \ell^2( \Lambda_+^*)$, $\sharp = ( \sharp_1, \dots, \sharp_n)\in \lbrace \cdot, *\rbrace^n, \flat = ( \flat_0, \dots, \flat_n) \in \lbrace \cdot, * \rbrace^{n+1}$ we define a $\Pi^{(1)}$-operator of order $n$ by 
\begin{align}
&\Pi^{(1)}_{\sharp,\flat} (f_1, \dots, f_n; g) \notag \\
&\quad = \sum_{p_1, \dots, p_n \in \Lambda_+^*} b_{\alpha_0, p_1}^{\flat_0} a_{\beta_1 p_1}^{\sharp_1} a_{\alpha_1 p_2}^{\flat_1} a_{\beta_2p_2}^{\sharp_2}a_{\alpha_2 p_3}^{\flat_2} \dots a_{\beta_{n-1}p_{n-1}}^{\sharp_{n-1}} a_{\alpha_{n-1} p_n}^{\flat_{n-1}} a_{\beta_n p_n}^{\sharp_n} a^{\flat_n} (g) \prod_{\ell = 1}^n f_\ell (p_\ell)
\end{align}
where $\alpha_\ell$ and $\beta_\ell$ are defined as before. Also here, we require that for all $\ell= 1, \dots, n$ either $\sharp_\ell = \cdot$ and $\flat_\ell =*$ or $\sharp = *$ and $\flat_\ell = \cdot$. Then the operators $\Pi^{(1)}$ leave the truncated Fock space invariant, too. Furthermore we require that $f_\ell \in \ell^1 ( \Lambda_+^*)$ if $\flat_{\ell -1} = \cdot$ and $\sharp_\ell = *$ for some $\ell=1, \dots, n$. The following Lemma proven in \cite{BS} shows that nested commutators ${\rm ad}_{B(\mu)} (b_p)$ can be expressed in terms of $(N-\mathcal{N}_+)/N$, $( N - ( \mathcal{N}_+ - 1))/N$ and $\Pi^{(1)}$ resp. $\Pi^{(2)}$-operators. 

\begin{lemma}
\label{lemma:dp}
Let $\mu \in \ell^2( \Lambda_+^*)$ be such that $\mu_p = \mu_{-p}$ for all $p \in \Lambda_+^*$. To simplify the notation, assume also $\mu$ to be real valued. Let $B( \mu)$ be defined as in \eqref{def:Bogo-b}, $n \in \mathbb{N}$ and $p \in \Lambda_+^*$. Then the nested commutator ${\rm ad}_{B(\mu)}^{(n)} (b_p)$ can be written as the sum of exactly $2^n n!$ terms wit the following properties. 

\begin{enumerate}
\item[(i)] Possibly up to a sign, each term has the form 
\begin{align}
\label{eq:lemma(i)}
\Lambda_1 \Lambda_2 \dots \Lambda_i N^{-k} \Pi_{\sharp, \flat}^{(1)} ( \mu^{j_1}. \dots, \mu^{j_k}; \mu_p^s \varphi_{\alpha p} )
\end{align}
for some $i,k,s \in \mathbb{N}$, $j_1, \dots, j_k \in \mathbb{N} \setminus \lbrace 0 \rbrace$, $\sharp \in \lbrace \cdot, * \rbrace^k, \flat \in \lbrace \cdot, * \rbrace^{k+1}$ and $\alpha \in \lbrace \pm \rbrace$ chosen so that $\alpha =1$ if $\flat_k = \cdot$ and $\alpha =-1$ of $\flat_k = *$ (recall that $\varphi_p (x) = e^{-ip \cdot x}$). In \ref{eq:lemma(i)} each operator $\Lambda_w : \mathcal{F}^{\leq N} \rightarrow \mathcal{F}^{\leq N}$, $w=1, \dots, i$ is either a factor of $(N-\mathcal{N}_+)/N$, a factor $( N - ( \mathcal{N}_+ - 1))/N$ or an operator of the form 
\begin{align}
\label{eq:lemma(ii)}
N^{-h}\Pi_{\sharp',\flat'}^{(2)} ( \mu^{z_1}, \mu^{z_2}, \dots, \mu^{z_h} )
\end{align}
for some $h,z_1, \dots, z_h \in \mathbb{N} \setminus \lbrace 0 \rbrace, \sharp, \beta \in \lbrace \cdot, * \rbrace^h$. 
\item[(ii)]  If a term of the form \eqref{eq:lemma(i)} contains $m \in \mathbb{N}$ factors $(N-\mathcal{N}_+)/N$ or $( N- ( \mathcal{N}_+ +1))/N$ and $j \in \mathbb{N}$ factors of te form \eqref{eq:lemma(i)} with $\Pi^{(2)}$ operators pf order $h_1, \dots, h_j \in \mathbb{N} \setminus \lbrace 0 \rbrace$, then we have 
\begin{align}
m+ (h_1 + 1) + \dots + (h_j +1) + (k+1) = n+1 
\end{align}
\item[(iii)]  If a term of the form \eqref{eq:lemma(i)} contains (considering all $\Lambda$-operators and the $\Pi^{(1)}$-operator) the arguments $\mu^{i_1}, \dots, \mu^{i_m}$ and the factor $\mu_p^s$ for some $m,s \in \mathbb{N}$ and $i_1, \dots, i_m \in \mathbb{N} \setminus \lbrace 0 \rbrace$, then 
\begin{align}
i_1 + \dots + i_m + s =n \; . 
\end{align}
\item[(iv)]  There is exactly one term having the form \eqref{eq:lemma(i)} with $k=0$ and such that all $\Lambda$-operators are factors of $(N-\mathcal{N}_+)/N$ or of $( N+1-\mathcal{N})/N$. It is given by 
\begin{align}
\left( \frac{N-\mathcal{N}_+}{N}\right)^{n/2} \left( \frac{N+1-\mathcal{N}_+}{N} \right)^{n/2} \mu_p^n b_p 
\end{align}
if $n$ is even, and by 
\begin{align}
- \left( \frac{N-\mathcal{N}_+}{N}\right)^{(n+1)/2} \left( \frac{N+1-\mathcal{N}_+}{N} \right)^{(n-1)/2} \mu_p^n b_{-p}^* 
\end{align}
if $n$ is odd. 
\item[(v)]  If the $\Pi^{(1)}$-operator in \eqref{eq:lemma(i)} is of order $k \in \mathbb{N}\setminus \lbrace 0 \rbrace$, it has either the form 
\begin{align}
\sum_{p_1, \dots, p_k} b_{\alpha_0 p_1}^{\flat_0} \prod_{i=1}^{k-1} a_{\beta_ip_i}^{\sharp_i} a_{\alpha_i p_{i+1}}^{\flat_i} a_{-p_k}^* \mu_{p}^{2r} a_p \prod_{i=1}^k \mu_{p_i}^{j_i} 
\end{align}
or the form 
\begin{align}
\sum_{p_1, \dots, p_k} b_{\alpha_0p_1}^{\flat_0} \prod_{i=1}^{k-1} a^{\sharp_i}_{\beta_i p_i} a_{\alpha_i p_{i+1}}^{\flat_i} a_{p_k}\mu_p^{2r+1} a_p^* \prod_{i=1}^k
 \mu_{p_i}^{j_i}
 \end{align}
for some $r \in \mathbb{N}$, $j_1, \dots, j_k \in \mathbb{N} \setminus \lbrace 0 \rbrace$. If it is of order $k=0$, then it is either given by $\mu_p^{2r}b_p$ or by $\mu_p^{2r+1} b_{-p}^*$ for some $r\in \mathbb{N}$. 

\item[(vi)] For every non-normally ordered term of the form 
\begin{align}
\sum_{q \in \Lambda^*} \mu_q^i a_q a_q^* , \quad \sum_{q \in \Lambda^*} \mu_q^i b_q a_q^*, \quad \sum_{q \in \Lambda^*} \mu_q^i a_q b_q^* \quad \text{or} \quad \sum_{q \in \Lambda^*} \mu_q^i b_q b_q^* 
\end{align}
appearing either in the $\Lambda$-operators or in the $\Pi^{(1)}$-operator in \eqref{eq:lemma(i)}, we have $i\geq 2$.
\end{enumerate}
\end{lemma}
As a consequence of Lemma \ref{lemma:dp}, for $\| \mu \| $ small enough we have 
\begin{align}
e^{-B( \mu)} b_p e^{B (\mu)} = \sum_{n=0}^\infty \frac{(-1)^n}{n!} {\rm ad}_{B ( \mu)}^{(n)} ( b_p ), \quad e^{-B( \mu)} b_p^* e^{B (\mu)} = \sum_{n=0}^\infty \frac{(-1)^n}{n!} {\rm ad}_{B ( \mu)}^{(n)} ( b_p^* ) 
\end{align}
and the series converge absolutely (see \cite[Lemma 3.3]{BBCS_optimal}). From this, we also get an explicitly define the remainder operators \eqref{def:d} by 
\begin{align}
\label{def:d-2}
d_p = \sum_{m \geq 0} \frac{1}{m!} \left[ {\rm ad}_{-B(\mu)}^{(m)} (b_p) - \mu_p^m b_{\alpha_m p}^{\sharp_m} \right], \quad d_p^*= \sum_{m \geq 0} \frac{1}{m!} \left[ {\rm ad}_{-B(\mu)}^{(m)} (b_p^*) - \mu_p^m b_{\alpha_m p}^{\sharp_{m+1}} \right]
\end{align}
where $p \in \Lambda_+^*$, $( \sharp_m, \alpha_m) = (\cdot, +1)$ if $m$ is even and $(\sharp_m, \alpha_m)= (*,-1)$ if $m$ is odd. This representation allows to prove the following improved error estimates on the remainder terms $d_p$ using Lemma \ref{lemma:dp} and Lemma \ref{lemma:b-con}. We start with the conjugation w.r.t. to $e^{\lambda \cN_+}$ first. 

\begin{lemma}
\label{lemma:eN-dp}
Under the same assumptions and notations of Lemma \ref{lemma:dp}, for $\vert \lambda \vert \leq 1$ and sufficiently small $\| \mu \| $, there exists $C>0$ (independent of $\lambda, \kappa$) such that 
\begin{align}
\label{eq:eN-dp}
\|\left( e^{\lambda \cN_+}d_p e^{-\lambda \cN_+} - d_p \right)\psi \| \leq C \vert  \lambda\vert    N^{-1} \left(  \| b_p ( \mathcal{N}_+ +1) \psi \| + \vert\mu_p \vert \| ( \mathcal{N}_+ +1)^{3/2} ) \psi \| \right) 
\end{align}
and 
\begin{align}
\label{eq:eN-dpstar}
\|\left( e^{\lambda \cN_+}d_p^* e^{-\lambda \cN_+} - d_p^* \right)\psi \| \leq C \vert \lambda\vert  N^{-1} \|  ( \mathcal{N}_+ +1)^{3/2} \psi \| \; . 
\end{align}
\end{lemma}

\begin{proof}
From \eqref{def:d-2} we find that 
\begin{align}
\| &\left( e^{\lambda \cN_+}d_p e^{-\lambda \cN_+} - d_p \right)\psi \| \notag \\
&\leq  \sum_{m \geq 0} \frac{1}{m!}  \Big\|\left( e^{\lambda \cN_+}\left[ {\rm ad}_{-B(\mu)}^{(m)} (b_p) - \mu_p^m b_{\alpha_m p}^{\sharp_m} \right] e^{-\lambda \cN_+}- \left[ {\rm ad}_{-B(\mu)}^{(m)} (b_p) - \mu_p^m b_{\alpha_m p}^{\sharp_m} \right] \right) \psi \Big\|  \label{eq:term10}
\end{align}
and by Lemma \ref{lemma:dp} the difference 
\begin{align}
e^{\lambda \cN_+}\left[ {\rm ad}_{-B(\mu)}^{(m)} (b_p) - \mu_p^m b_{\alpha_m p}^{\sharp_m} \right] e^{-\lambda \cN_+}- \left[ {\rm ad}_{-B(\mu)}^{(m)} (b_p) - \mu_p^m b_{\alpha_m p}^{\sharp_m} \right]
\end{align}
is the sum of one term of the form 
\begin{align}
A_p =  e^{\lambda \cN_+} & \left( \frac{N-\cN_+}{N}\right)^{\frac{m+(1-\alpha_m)/2}{2}} \left( \frac{N+1-\cN_+}{N}\right)^{\frac{m+(1+\alpha_m)/2}{2}} \mu_p b_{\alpha_m p}^{\sharp_m}  e^{-\lambda \cN_+} \notag \\
 &- \left( \frac{N-\cN_+}{N}\right)^{\frac{m+(1-\alpha_m)/2}{2}} \left( \frac{N+1-\cN_+}{N}\right)^{\frac{m+(1+\alpha_m)/2}{2}} \mu_p b_{\alpha_m p}^{\sharp_m} \label{eq:term11}
\end{align}
and $2^m m! - 1$ terms are of the form 
 \begin{align}
B_p =  e^{\kappa \lambda \cN_+}  & \Lambda_1 \dots \Lambda_{i_1} N^{-k} \Pi_{\sharp, \flat}^{(1)} ( \mu^{j_1}, \dots, \mu^{j_{k_1}}; \mu_p^{\ell_1} \varphi_{\alpha_{\ell_1} p} g_p ) e^{- \lambda \kappa\cN_+ } \notag \\
 &- \Lambda_1 \dots \Lambda_{i_1} N^{-k} \Pi_{\sharp, \flat}^{(1)} ( \mu^{j_1}, \dots, \mu^{j_{k_1}}; \mu_p^{\ell_1} \varphi_{\alpha_{\ell_1} p} ) \label{eq:term12}
\end{align} 
where $i_1, k_1, \ell_1 \in \mathbb{N}, j_1, \dots, j_k \in \mathbb{N} \setminus \lbrace 0 \rbrace$ and where each operator $\Lambda_r$ is either a factor $( N- \cN_+)/N$, a factor $(N+1-\cN_+)/N$  or a $\Pi^{(2)}$ operator of the form 
\begin{align}
\label{eq:term13}
N^{-h}\Pi^{(2)}_{\sharp,\flat} ( \mu^{z_1}, \dots, \mu^{z_h} ) 
\end{align}
with $h, z_1, \dots, z_h \in \mathbb{N} \setminus \lbrace 0 \rbrace$. We consider \eqref{eq:term11} and \eqref{eq:term12} separately, thus each term that is of the form \eqref{eq:term11} either has $k_1 >0$ or contains at least one operator of the form \eqref{eq:term13}. We start with estimating \eqref{eq:term11} first that vanishes for $m=0$. Furthermore we have from Lemma \ref{lemma:b-con} 
\begin{align}
 & \|  A_p \psi \| \notag \\
 &=  \Big\| \left( \frac{N-\cN_+}{N}\right)^{\frac{m+(1-\alpha_m)/2}{2}} \left( \frac{N+1-\cN_+}{N}\right)^{\frac{m+(1+\alpha_m)/2}{2}} \mu_p^m \left(   e^{\lambda \cN_+}  b_{\alpha_m p}^{\sharp_m} e^{-\lambda \cN_+} -  b_{\alpha_m p}^{\sharp_m}\right) \psi \| \notag \\ 
&\leq  \kappa \lambda C^m \vert \mu_p \vert^m N^{-1}\|( \mathcal{N}_+ + 1)^{3/2} \psi \| \; . \label{eq:term14}
\end{align}
For \eqref{eq:term12} we find 
\begin{align}
B_p =& \sum_{u=1}^i \left(  \prod_{t=1}^{u-1} e^{\lambda \cN_+}  \Lambda_t e^{-\lambda \cN_+} \right) \left( e^{\lambda \cN_+}  \Lambda_u e^{-\lambda \cN_+} - \Lambda_u \right) \prod_{t=u+1}^{i}   \Lambda_t  \notag \\
& \quad \times N^{-k} \Pi_{\sharp, \flat}^{(1)} ( \mu^{j_1}, \dots, \mu^{j_{k_1}}; \mu_p^{\ell_1} \varphi_{\alpha_{\ell_1} p} ) \notag \\
+& \left( \prod_{t=1}^{i}   \Lambda_t \right)   N^{-k} \notag \\
& \quad \times \left(  e^{\lambda \cN_+}\Pi_{\sharp, \flat}^{(1)} ( \mu^{j_1}, \dots, \mu^{j_{k_1}}; \mu_p^{\ell_1} \varphi_{\alpha_{\ell_1} p} ) e^{-\lambda \cN_+} - \Pi_{\sharp, \flat}^{(1)} ( \mu^{j_1}, \dots, \mu^{j_{k_1}}; \mu_p^{\ell_1} \varphi_{\alpha_{\ell_1} p} )\right) \;.
\end{align}
In case $\Lambda_u$ is of the form $( N- \mathcal{N}_+)/N $ or $( N +1- \mathcal{N}_+)/N$ then $e^{\lambda \cN_+}  \Lambda_u e^{-\lambda \cN_+} - \Lambda_u $ vanishes. Otherwise, if $\Lambda_u$ is an operator of the form $\Pi^{(2)}$ it creates resp. annihilates two particles, thus, in this case it follows from Proposition \ref{prop:eN} that $e^{\lambda \cN_+}  \Lambda_u e^{-\lambda \cN_+} - \Lambda_u =  ( e^{\lambda \kappa_u} -1)\Lambda_u$ with $\kappa_u =2$ or $\kappa_u = -2$. Similarly, as the operator $\Pi^{(1)}$ creates or annihilates one particle, we have 
\begin{align}
\Pi_{\sharp, \flat}^{(1)}  &( \mu^{j_1}, \dots, \mu^{j_1}, \dots, \mu^{j_{k_1}}; \mu_p^{\ell_1} \varphi_{\alpha_{\ell_1} p} ) e^{-\lambda \cN_+} - \Pi_{\sharp, \flat}^{(1)} ( \mu^{j_1}, \dots, \mu^{j_{k_1}}; \mu_p^{\ell_1} \varphi_{\alpha_{\ell_1} p} ) \notag \\
=&(e^{\lambda \kappa} - 1) \Pi_{\sharp, \flat}^{(1)} ( \mu^{j_1}, \dots,   \mu^{j_1}, \dots, \mu^{j_{k_1}}; \mu_p^{\ell_1} \varphi_{\alpha_{\ell_1} p} ) 
\end{align}
with $\kappa = 1$ or $\kappa=-1$. Therefore we find 
\begin{align}
\Big\| B_p \psi \Big\| \leq& \left(\sum_{u=1}^i (e^{\lambda\kappa_u} -1)+ (e^{\lambda\kappa}-1) \right) \|  \prod_{t=1}^{i}   \Lambda_t N^{-k} \Pi_{\sharp, \flat}^{(1)} ( \mu^{j_1}, \dots, \mu^{j_{k_1}}; \mu_p^{\ell_1} \varphi_{\alpha_{\ell_1} p} ) \psi \| \; .
\end{align}
We consider the case $\ell_1 =0$ and $\ell_1 >0$ separately (see for example \cite[Lemma 3.4]{BBCS_optimal} resp. \cite[Section 5]{BS}) and arrive with $\vert \mu_p \vert \leq \| \mu \|$ at 
\begin{align}
\Big\| B_p \psi \Big\| \leq & \vert \lambda\vert  C^m N^{-1} \left( \| \mu \|^{m-\ell_1} \vert \mu_p \vert^{\ell_1} \delta_{\ell_1 >0} \| ( \mathcal{N}_+ +1)^{3/2} \psi \| + \| \mu \|^m \| b_p ( \mathcal{N}_+ +1) \xi \| \right) \notag \\
\leq& \vert \lambda \vert  C^m N^{-1} \| \mu\|^{m-1}\left(\vert \mu_p \vert \delta_{m>0}  \| ( \mathcal{N}_+ +1)^{3/2} \psi \| + \| \mu \| \| b_p ( \mathcal{N}_+ +1) \xi \| \right) . \label{eq:term15}
\end{align}
We plug \eqref{eq:term14} and \eqref{eq:term15} into \eqref{eq:term10} and conclude for sufficiently small $\|\mu \|$ at \eqref{eq:eN-dp}. The second bound \eqref{eq:eN-dpstar} follows similarly using that in the case $\ell_1 =0$ we only have $\| b_p^* ( \cN_+ + 1) \psi \| \leq \| ( \cN_+ + 1)^{3/2} \psi \|$. 
\end{proof}

Next we prove similar estimates for the conjugation of $d_p,d_p^*$ with the two exponentials $ e^{\kappa \lambda \cN_+} e^{\lambda \sqrt{N} \phi_+ (g)} $. To this end we first prove the following auxiliary estimates.

\begin{lemma}
\label{lemma:Pi2-conj}
Under the same assumptions and notations of Lemma \ref{lemma:dp}, let $\vert \kappa \lambda \vert, \vert \lambda \vert \leq 1$, $g \in \ell^2( \Lambda_+^*)$. Then for sufficiently small $\| \mu \|$ there exists $C>0$ (independent of $\kappa, \lambda$) such that 
\begin{align}
\| \Big(  & e^{\kappa \lambda \cN_+} e^{\lambda \sqrt{N} \phi_+ (g)}  \mathcal{N}_+  e^{- \lambda \sqrt{N} \phi_+ (g))} e^{- \lambda \kappa\cN_+ } -  \mathcal{N}_+  \ \Big)  \psi \|\notag \\
& \leq C \left( \lambda^2 \| g \|^2 \| \psi \|  + \vert \lambda \vert \| g \|   N^{-1/2} \| ( \mathcal{N}_+ + 1)^{1/2} \psi \| \right)\; . \; \notag \\
\| \Big(  & e^{\kappa \lambda \cN_+} e^{\lambda \sqrt{N} \phi_+ (g)} \Pi^{(2)}_{\sharp',\flat'} ( \mu^{z_1}, \mu^{z_2}, \dots, \mu^{z_n} ) e^{- \lambda \sqrt{N} \phi_+ (g))} e^{- \lambda \kappa\cN_+ } - \Pi^{(2)}_{\sharp',\flat'} ( \mu^{z_1}, \mu^{z_2}, \dots, \mu^{z_n} )   \Big)  \psi \| \notag \\
& \leq  C  N^n \left(\lambda^2  \| g \|^2 \| \psi \|  + \vert \lambda \vert (\| g \|+ \vert \kappa \vert ) N^{-1/2} \| ( \mathcal{N}_+ + 1)^{1/2} \psi \| \right) . 
\end{align}
\end{lemma}

\begin{proof}
The first bound follows with similar arguments as in the proof of Lemma \ref{lemma:bb-conj}. For the second we note that from definition \eqref{def:Pi2} it follows 
\begin{align}
 & e^{\kappa \lambda \cN_+} e^{\lambda \sqrt{N} \phi_+ (g)} \Pi^{(2)}_{\sharp',\flat'} ( \mu^{z_1}, \mu^{z_2}, \dots, \mu^{z_n} ) e^{- \lambda \sqrt{N} \phi_+ (g))} e^{- \lambda \kappa\cN_+ } - \Pi^{(2)}_{\sharp',\flat'} ( \mu^{z_1}, \mu^{z_2}, \dots, \mu^{z_n} )  \notag \\
=& \sum_{p_1, \dots, p_n \in \Lambda_+^*} \left[ e^{\kappa \lambda \cN_+} e^{\lambda \sqrt{N} \phi_+ (g)}   b_{\alpha_0, p_1}^{\flat_0} e^{- \lambda \sqrt{N} \phi_+ (g)} e^{- \lambda \kappa\cN_+ } -  b_{\alpha_0, p_1}^{\flat_0} \right] \notag \\
& \hspace{3cm} \times \prod_{t=1}^{n-1} a_{\beta_t p_t}^{\sharp_t} a_{\alpha_t p_{t+1}}^{\flat_t}  a_{\beta_n p_n}^{\sharp_n} a^{\flat_n} (g) \prod_{\ell = 1}^n f_\ell (p_\ell) \notag \\
&+ \sum_{p_1, \dots, p_n \in \Lambda_+^*}   e^{\lambda \kappa \cN_+} e^{\sqrt{N} \phi_+(g)} b_{\alpha_0, p_1}^{\flat_0}  e^{-\sqrt{N} \phi_+(g)} e^{-\lambda \kappa \cN_+} \notag \\
& \hspace{0.2cm} \times \sum_{j=1}^n \left( \prod_{t=1}^{j-1}  e^{\kappa \lambda \cN_+} e^{\lambda \sqrt{N} \phi_+ (g) } a_{\beta_t p_t}^{\sharp_t} a_{\alpha_t p_{t+1}}^{\flat_t} e^{- \lambda \sqrt{N} \phi_+ (g)} e^{- \lambda \kappa\cN_+ }\right) \notag \\
& \hspace{0.8cm} \times \left[ e^{\kappa \lambda \cN_+} e^{\lambda \sqrt{N} \phi_+ (g)} a_{\beta_{j}p_{j}}^{\sharp_{j}}a_{\alpha_{j} p_{j+1}}^{\flat_{j}} e^{- \lambda \sqrt{N} \phi_+ (g)} e^{- \lambda \kappa\cN_+ } - a_{\beta_{j}p_{j}}^{\sharp_{j}}a_{\alpha_{j} p_{j+1}}^{\flat_{j}} \right] \notag \\
& \hspace{3cm} \times \left( \prod_{u=j+1}^{n-1} a_{\beta_u p_u}^{\sharp_u} a_{\alpha_u p_{u+1}}^{\flat_u} \right) a_{\beta_n p_n}^{\sharp_n} a^{\flat_n} (g) \prod_{\ell = 1}^n f_\ell (p_\ell) \; . \label{eq:conj-Pi2-1}
\end{align}
On the one hand, it follows from Lemma \ref{lemma:b-con} that 
\begin{equation}
\label{eq:prop1-1} 
\begin{split} 
e^{\lambda \kappa \cN_+}e^{\sqrt{N} \phi_+ (g)} &  b_{\alpha_0 p_1} e^{-\sqrt{N} \phi_+(g)}e^{\lambda \kappa \cN_+}  - b_{\alpha_0 p_1} \notag \\
=&\;  \left[ \left( \gamma_{\lambda \| g \|}  -1 \right)  e^{\lambda \kappa} +  \gamma_{\lambda \| g \|} \left( e^{\lambda \kappa} - 1 \right) \right]  b_{\alpha_0 p_1} \notag \\
&+ \lambda \gamma_{\| g \|} \frac{\gamma_{\lambda \| g \|} - 1}{\lambda^2\| g \|^2} g_{-\alpha_0 p_1} \left( \gamma_{\kappa \lambda}
i \phi_- (g) + \sigma_{\kappa \lambda} \phi_+ (g) \right) \notag \\
 &- \lambda \frac{\gamma_{\lambda \| g \|} - 1}{\| g \|^2} e^{- \lambda \kappa} g_{-\alpha_0 p}  b^* (g) - \sqrt{N} \, \lambda\gamma_{\lambda \|g \|} \frac{\sigma_{\lambda \| g \|}}{\lambda \| g \|}  g_{-\alpha_0 p_1} \left( 1 - \frac{\N}{N} \right) \notag \\
 & + \lambda^3 \frac{1}{\sqrt{N}} \, \frac{\sigma_{\lambda \| g \|}}{\lambda\| g \|} \frac{\gamma_{\lambda\| g \|} - 1}{ \lambda \| g \|^2}  g_{-\alpha_0 p_1} a^* (g) a (g) + \lambda\frac{1}{\sqrt{N}} \, \frac{\sigma_{\lambda \| g \|}}{\lambda \| g \|} a^* (g) a_{\alpha_0 p_1}  \; 
\end{split} 
\end{equation}
We recall that from the estimates \eqref{eq:bounds-b}, any term is $O( \sqrt{N} \lambda)$ for small $\lambda$ and large $N$. On the other hand, from Lemma \ref{lemma:b-con} for $\sharp_j = *$ and $\flat_j = \cdot$
\begin{equation}\label{eq:prop1-2}
\begin{split} 
& e^{\lambda \kappa \cN_+}e^{\sqrt{N} \phi_+(h)} a_{\beta_j p_j}^{\sharp_j} a_{\alpha_j p_{j+1}}^{\flat_j} e^{-\sqrt{N} \phi_+(h)} e^{-\lambda \kappa \cN_+}- a_{\beta_j p_j}^{\sharp_j} a_{\alpha_j p_{j+1}}^{\flat_j} \notag \\
& \quad= \lambda \sqrt{N} \, \frac{\sigma_{\lambda \| g \|}}{\lambda \| g \|} i \left[ \gamma_{\lambda\kappa} \left( b_{\beta_j p_j}^* g_{\alpha_j p_{j+1}} - b_{\alpha_j p_{j+1}} g_{- \beta_j p_j} \right)  + \sigma_{\lambda\kappa } \left( b_{\beta_j p_j}^* g_{\alpha_j p_{j+1}} + b_{\alpha_j p_{j+1}} g_{- \beta_j p_j} \right) \right] \\
 &\quad \quad -\lambda^2  N \frac{\sigma_{\lambda \| g \|}^2}{\lambda^2 \|g \|^2} g_{\beta_j p_j} g_{-\alpha_j p_{j+1}} \left(1 - \frac{\cN}{N} \right) + \lambda^2 \frac{(\gamma_{\lambda \| g \|} -1)}{\lambda^3\| g \|^2} ( g_{-\alpha_j p_{j+1}}a^*_{\beta_j p_j} a(g) + g_{\beta_j p_j}a^*(g) a_{\alpha_j p_{j+1}})
\\ &\quad \quad + \sqrt{N} \,  \lambda^3 \frac{\sigma_{\lambda \| g \|}}{\lambda \| g \|} \frac{\gamma_{\lambda \| g \|} - 1}{\lambda \| g \|^2} g_{\beta_j p_j} g_{-\alpha_j p_{j+1}} \left( \gamma_{\kappa \lambda} i \phi_- (g) + \sigma_{\kappa \lambda} \phi_+ (g) \right) \notag \\ 
&\quad \quad +\lambda^4 \left( \frac{\gamma_{\lambda \| g \|} - 1}{\| g \|^2} \right)^2  g_{\beta_j p_j} g_{-\alpha_j p_{j+1}}  a^* (g) a (g)  \; \end{split}
\end{equation}
and, similarly to Lemma \ref{lemma:bb-conj}, any term is either bounded by multiples of $\sqrt{N} \lambda \| ( \cN_+ + 1)^{1/2} \psi \|$ or $O( \lambda^2 N)$. Note that the case $\sharp_j = \cdot $ and $\flat_j = *$ follows in the same way using the commutation relations. Moreover, \eqref{eq:prop1-1},\eqref{eq:prop1-2} show that terms appearing in  \eqref{eq:conj-Pi2-1} of the form 
\begin{align}
\| e^{\lambda \kappa \cN_+} e^{\sqrt{N} \phi_+(h)} a_{\beta_j p_j}^{\sharp_j} a_{\alpha_j p_{j+1}}^{\flat_j} e^{-\sqrt{N} \phi_+(h)} e^{-\lambda \kappa \cN_+} \psi \|, 
\| e^{\lambda \kappa \cN_+} e^{\sqrt{N} \phi_+(h)} b_{\beta_j p_j}^{\sharp_j}  e^{-\sqrt{N} \phi_+(h)}e^{-\lambda \kappa \cN_+}  \psi \| 
\end{align}
bounded through multiples of  $N^{1/2} \| (\mathcal{N}_+ + 1)^{1/2} \| +  \lambda^2 N \| \psi \| \leq C N $ resp. $ \|( \mathcal{N}_+ + 1)^{1/2} \| +  \lambda N^{1/2} \| \psi \| \leq C N^{1/2} $. Since the number of particles operator can be easily commuted through $ a_{\beta_j p_j}^{\sharp_j} a_{\alpha_j p_{j+1}}^{\flat_j}, b_{\alpha_0 p_1}  $, we get 
\begin{align}
\| \Big( e^{\kappa \lambda \cN_+} & e^{\lambda \sqrt{N} \phi_+ (g)} \Pi^{(2)}_{\sharp',\flat'} ( \mu^{z_1}, \mu^{z_2}, \dots, \mu^{z_n} ) e^{- \lambda \sqrt{N} \phi_+ (g))} e^{- \lambda \kappa\cN_+ } - \Pi^{(2)}_{\sharp',\flat'} ( \mu^{z_1}, \mu^{z_2}, \dots, \mu^{z_n} ) \Big)  \psi  \| \notag \\
\leq& n (CN)^n \vert\lambda \vert ( \| g \| +  \vert \kappa \vert ) \left( \lambda^2 \| \psi \|+  N^{-1/2} \| ( \mathcal{N}_+ +1)^{1/2} \psi \| \right)  \notag \\
 \leq& (CN)^n \vert\lambda \vert ( \| g \| +  \vert \kappa \vert ) \left( \lambda^2 \| \psi \|+  N^{-1/2} \| ( \mathcal{N}_+ +1)^{1/2} \psi \| \right)  
\end{align}
and Lemma \ref{lemma:Pi2-conj} follows. 
\end{proof}

From these estimates, we derive the following estimates for \eqref{eq:estimates-dp-start}.

\begin{lemma}
\label{lemma:conj-dp}
Under the same assumptions and notations as in Lemma \ref{lemma:dp}, $g \in \ell^2( \Lambda_+^*)$ and $\vert \lambda \vert, \vert \lambda\kappa \vert \leq 1$ and $\| \mu \|$ small enough. Then there exists $C>0$ (independent of $\kappa, \lambda$) such that 
\begin{align}
\Big\| & ( \mathcal{N}_+ + 1)^{k/2} \left( e^{\kappa \lambda \cN_+} e^{\lambda \sqrt{N} \phi_+ (g)} d_p e^{- \lambda \sqrt{N} \phi_+ (g)} e^{- \lambda \kappa\cN_+ } -  d_p \right) \psi \Big\| \notag \\
&\leq  C  \left( (\| g \|+ \vert \kappa \vert ) \vert \mu_p \vert  + \vert g_p \vert ) \right) \vert \left( \vert \lambda\vert  N^{-1/2} \| ( \mathcal{N}_+ + 1 )^{(k+2)/2} \psi \|  + \vert \lambda\vert^3 \| g \|^2 \sqrt{N} \| ( \mathcal{N}_+ +1)^{k/2}\psi \| \right) \notag \\
 &+   C  (\| g \|+ \vert \kappa \vert )  \left( \vert\lambda\vert N^{-1/2} \| b_p( \mathcal{N}_+ + 1 )^{(k+1)/2} \psi \|  + \vert\lambda\vert^3 \| g \|^2 \sqrt{N} \| b_p( \mathcal{N}_+ +1)^{(k-1)/2}\psi \| \right) \; 
\end{align}
and 
\begin{align}
\Big\|  &  ( \mathcal{N}_+ + 1)^{k/2} \left( e^{\kappa \lambda \cN_+} e^{\lambda \sqrt{N} \phi_+ (g)} d_p^* e^{- \lambda \sqrt{N} \phi_+ (g)} e^{- \lambda \kappa\cN_+ } -  d_p^* \right) \psi \Big\|\notag \\
& \leq C (\| g \|+ \vert \kappa \vert )  \left( \vert\lambda\vert N^{-1/2} \| ( \mathcal{N}_+ + 1 )^{(k+2)/2} \psi \|  + \vert\lambda\vert^3 \| g \|^2 \sqrt{N} \| ( \mathcal{N}_+ + 1)^{k/2} \psi \| \right)  \notag \\
\end{align}
\end{lemma}

\begin{proof}
 We start with the first bound and observe that from \eqref{def:d-2} we have 
 \begin{align}
 \| \Big( e^{\kappa \lambda \cN_+}  & e^{\lambda \sqrt{N} \phi_+ (g)} d_p e^{- \lambda \sqrt{N} \phi_+ (g)} e^{- \lambda \kappa\cN_+ } -  d_p \Big) \psi \|   \notag \\
 =& \sum_{\geq 0} \frac{1}{m!} \Big\|  \Big( e^{\kappa \lambda \cN_+} e^{\lambda \sqrt{N} \phi_+ (g)}  \left[ \ad_{-B (\mu)}^{(m)} (b_q) - \mu_q^m b_{\alpha_m p}^{\sharp_m} \right] e^{- \lambda \sqrt{N} \phi_+ (g)} e^{- \lambda \kappa\cN_+ }  \notag \\
 &\hspace{6cm} -  \left[ \ad_{-B (\mu)}^{(m)} (b_q)  - \mu_q^m b_{\alpha_m p}^{\sharp_m} \right]   \Big) \psi \Big\| \; . \label{eq:term20}
\end{align}
By Lemma \ref{lemma:dp} the term inside the norm can be written by a sum where one term is given by 
\begin{align}
A_p =&  e^{\kappa \lambda \cN_+} e^{\lambda \sqrt{N} \phi_+ (g)}  \left( \frac{N - \mathcal{N}_+}{N}\right)^{\frac{m+ (1-\alpha_m)/2}{2}}  \left(\frac{N + 1- \mathcal{N}_+}{N}\right)^{\frac{m- (1-\alpha_m)/2}{2}}   \notag \\
& \hspace{3cm} \times \mu_p^mb^{\sharp_m}_{\alpha_m p} e^{- \lambda \sqrt{N} \phi_+ (g)} e^{- \lambda \kappa\cN_+ }  \notag \\
&-  \left( \frac{N - \mathcal{N}_+}{N}\right)^{\frac{m+ (1-\alpha_m)/2}{2}}  \left(\frac{N + 1- \mathcal{N}_+}{N}\right)^{\frac{m- (1-\alpha_m)/2}{2}}   \mu_p^mb^{\sharp_m}_{\alpha_m p} \label{eq:term21}
\end{align}
and $2^m m! - 1$ terms are of the form 
 \begin{align}
B_p= e^{\kappa \lambda \cN_+}  & e^{\lambda \sqrt{N} \phi_+ (g)}  \Lambda_1 \dots \Lambda_{i_1} N^{-k} \Pi_{\sharp, \flat}^{(1)} ( \mu^{j_1}, \dots, \mu^{j_{k_1}}; \mu_p^{\ell_1} \varphi_{\alpha_{\ell_1} p})  e^{- \lambda \sqrt{N} \phi_+ (g)} e^{- \lambda \kappa\cN_+ } \notag \\
& \hspace{3cm} - \Lambda_1 \dots \Lambda_{i_1} N^{-k} \Pi_{\sharp, \flat}^{(1)} ( \mu^{j_1}, \dots, \mu^{j_{k_1}}; \mu_p^{\ell_1} \varphi_{\alpha_{\ell_1} p})  \label{eq:term22}
\end{align} 
where $i_1, k_1, \ell_1 \in \mathbb{N}, j_1, \dots, j_k \in \mathbb{N} \setminus \lbrace 0 \rbrace$ and where each operator $\Lambda_r$ is either a factor $( N- \cN_+)/N$, a factor $(N+1-\cN_+)/N$  or a $\Pi^{(2)}$ operator of the form 
\begin{align}
\label{eq:term23}
N^{-h}\Pi^{(2)}_{\sharp,\flat} ( \mu^{z_1}, \dots, \mu^{z_h} ) 
\end{align}
with $h, z_1, \dots, z_h \in \mathbb{N} \setminus \lbrace 0 \rbrace$. We consider terms of the form \eqref{eq:term21} and \eqref{eq:term22} separately. Each term of the form \eqref{eq:term22} has either $k_1 >0$ or at least one operator that is of the form \eqref{eq:term23}. We start with \eqref{eq:term21} that vanishes for $m=0$. We have for $\beta_m = (1-\alpha_m)/2$
\begin{align}
A_p := & \sum_{j=1}^{(m+\beta_m)/2}  \left( \frac{N - \mathcal{N}_+}{N}\right)^{(m+\beta_m)/2 -j+1} \left( e^{\kappa \lambda \cN_+}   e^{\lambda \sqrt{N} \phi_+ (g)}  \frac{\mathcal{N}_+}{N}  e^{- \lambda \sqrt{N} \phi_+ (g)} e^{- \lambda \kappa\cN_+ } -  \frac{\mathcal{N}_+}{N}\right)  \notag \\
& \hspace{2cm} \times  \left(\frac{ N-\mathcal{N}_+}{N}\right)^{j} \left(\frac{N + 1- \mathcal{N}_+}{N}\right)^{(m+-\beta_m)/2}   \mu_p^mb^{\sharp_m}_{\alpha_m p} \notag \\
&+ \left(\frac{N - \mathcal{N}_+}{N}\right)^{(m+\beta_m)/2}  \sum_{j=1}^{(m-\beta_m)/2}  \left( \frac{N +1 - \mathcal{N}_+}{N}\right)^{(m-\beta_m)/2 -j+1}  \notag \\
&\hspace{0.5cm} \times \left( e^{\kappa \lambda \cN_+}   e^{\lambda \sqrt{N} \phi_+ (g)}  \left(\frac{\mathcal{N}_+}{N}\right) e^{- \lambda \sqrt{N} \phi_+ (g)} e^{- \lambda \kappa\cN_+ } -  \frac{ \mathcal{N}_+}{N}\right)    \left(\frac{N +1- \mathcal{N}_+}{N}\right)^{j} \mu_p^mb^{\sharp_m}_{\alpha_m p} \notag \\
&+ e^{\kappa \lambda \cN_+}  e^{\lambda \sqrt{N} \phi_+ (g)}  \left( \left(\frac{N - \mathcal{N}_+}{N}\right)^{(m+\beta_m)/2} \left(\frac{N + 1- \mathcal{N}_+}{N}\right)^{(m-\beta_m)/2}  -1 \right)  \notag \\
& \hspace{2cm} \times e^{- \lambda \sqrt{N} \phi_+ (g)} e^{- \lambda \kappa\cN_+ } \notag \\
&\hspace{3cm} \times \mu_p^m \left( e^{\kappa \lambda \cN_+} e^{\lambda \sqrt{N} \phi_+ (g)}  b^{\sharp_m}_{\alpha_m p} e^{- \lambda \sqrt{N} \phi_+ (g)} e^{- \lambda \kappa\cN_+ } - b^{\sharp_m}_{\alpha_m p} \right) \label{eq:term24}
\end{align}  
Since, for the second summand of the r.h.s. of \eqref{eq:term24} we have 
\begin{align}
&\left(\frac{N - \mathcal{N}_+}{N}\right)^{(m+\beta_m)/2} \left(\frac{N + 1- \mathcal{N}_+}{N}\right)^{(m-\beta_m)/2}  -1 \notag \\
&= - \sum_{j=1}^{(m+\beta_m)/2} \left(\frac{N - \mathcal{N}_+}{N}\right)^{(m+\beta_m)/2-j+1} \frac{\mathcal{N}_+}{N} \left(\frac{N - \mathcal{N}_+}{N} \right)^j \left(\frac{N + 1- \mathcal{N}_+}{N}\right)^{(m-\beta_m)/2} \notag \\
&\quad - \left(\frac{N - \mathcal{N}_+}{N}\right)^{(m+\beta_m)/2} \sum_{j=1}^{(m-\beta_m)/2} \left(\frac{N  +1- \mathcal{N}_+}{N}\right)^{(m+\beta_m)/2-j+1}\left(\frac{\mathcal{N}_+ -1}{N}\right)\left(\frac{N - \mathcal{N}_+}{N} \right)^j  \; . 
\end{align}
we can argue similarly as in the proof of Lemma \ref{lemma:Pi2-conj}. In particular, since powers of the number of excitations $ \mathcal{N}_+ $ can be easily commuted through any operator appearing \eqref{eq:term24}, we find from Lemma \ref{lemma:Pi2-conj} for the first and Lemma \ref{lemma:b-con} for the second term of the r.h.s. of  \eqref{eq:term22} 
\begin{align}
\|& ( \mathcal{N}_+ + 1)^{k/2} A_p \psi \| \notag \\
&\leq C^m \vert \mu_p \vert \| \mu \|^{m-1} \left( \vert \lambda \vert ( \| g \| + \vert \kappa \vert ) N^{-1/2} \| ( \mathcal{N}_+ + 1 )^{(k+2)/2} \psi \|  + \lambda^3 \| g \|^3  \sqrt{N} \| ( \mathcal{N}_+ + 1)^k \psi \| \right) \; . 
\end{align}
for some constants $C >0$.  For \eqref{eq:term22}, we write 
 \begin{align}
\mathcal{N}^k B_p 
 =& \mathcal{N}^k\sum_{n=1}^{i_1} \Big( \prod_{t=1}^{n-1} e^{\kappa \lambda \cN_+} e^{\lambda \sqrt{N} \phi_+ (g)}\Lambda_{t}  e^{- \lambda \sqrt{N} \phi_+ (g)} e^{- \lambda \kappa\cN_+ } \Big) \notag \\
 & \hspace{3cm} \times  \left( e^{\kappa \lambda \cN_+} e^{\lambda \sqrt{N} \phi_+ (g)}  \Lambda_n e^{- \lambda \sqrt{N} \phi_+ (g)} e^{- \lambda \kappa\cN_+ }  - \Lambda_n \right) \notag \\
 &\hspace{3cm} \times \Lambda_{n+1} \dots \Lambda_{i_1} \Pi_{\sharp, \flat}^{(1)} ( \mu^{j_1}, \dots, \mu^{j_{k_1}}; \mu_p^{\ell_1} \varphi_{\alpha_{\ell_1} p}) \notag \\
 &+ \mathcal{N}^k \Big( \prod_{t=1}^{i_1} e^{\kappa \lambda \cN_+} e^{\lambda \sqrt{N} \phi_+ (g)}   \Lambda_{t}  e^{- \lambda \sqrt{N} \phi_+ (g)} e^{- \lambda \kappa\cN_+ } \Big) \notag \\
 & \hspace{1cm} \times \left(  e^{- \lambda \sqrt{N} \phi_+ (g)} e^{- \lambda \kappa\cN_+ } N^{-k} \Pi_{\sharp, \flat}^{(1)} ( \mu^{j_1}, \dots, \mu^{j_{k_1}}; \mu_p^{\ell_1} \varphi_{\alpha_{\ell_1} p})  e^{- \lambda \sqrt{N} \phi_+ (g)} e^{- \lambda \kappa\cN_+ }  \right. \notag \\
 & \hspace{2cm} \left. - \Pi_{\sharp, \flat}^{(1)} ( \mu^{j_1}, \dots, \mu^{j_{k_1}}; \mu_p^{\ell_1} \varphi_{\alpha_{\ell_1} p})  \right)  \label{eq:term25}
\end{align} 
First note again powers of the number of excitations can be easily commuted through operators appearing in this term. Moreover, by Lemma \ref{lemma:Pi2-conj} we have 
\begin{align}
\| e^{\kappa \lambda \cN_+} e^{\lambda \sqrt{N} \phi_+ (g)}\Lambda_{t}  e^{- \lambda \sqrt{N} \phi_+ (g)} e^{- \lambda \kappa\cN_+ }  \psi \| \leq C \; .
\end{align}
Thus the first term of the r.h.s. of \eqref{eq:term25} can be estimated by Lemma \ref{lemma:Pi2-conj} distinguishing the case $\ell_1 =0$ and $\ell_1>0$ as in the proof of Lemma \ref{lemma:eN-dp} by 
\begin{align}
 & C^m   \| \mu \|^m \vert \mu_p \vert \left( \vert\lambda\vert ( \| g \| + \vert \kappa \vert ) N^{-1/2} \| ( \mathcal{N}_+ + 1 )^{(k+2)/2} \psi \|  + \vert\lambda\vert^3 \| g \|^3 \sqrt{N} \| ( \mathcal{N}_+ +1)^{k/2}\psi \| \right) \notag \\
 &+   C^m \| \mu \|^m \left( \vert\lambda\vert ( \| g \| + \vert \kappa \vert )   N^{-1/2} \| b_p( \mathcal{N}_+ + 1 )^{(k+1)/2} \psi \|  + \vert\lambda\vert^3\| g \|^3 \sqrt{N} \| b_p( \mathcal{N}_+ +1)^{(k-1)/2}\psi \| \right)
\end{align}
For the second term of the r.h.s. of \eqref{eq:term25} we proceed similarly as in the proof of Lemma \ref{lemma:Pi2-conj} by Lemma \ref{lemma:b-con} distinguishing again the case $\ell_1 =0$ and $\ell_1>0$ and thus finally get 
\begin{align}
&  C^m   \| \mu \|^m \left( ( \| g \| + \vert \kappa \vert )\vert \mu_p  \vert + \vert g_p \vert \right) \vert \left( \vert\lambda\vert  N^{-1/2} \| ( \mathcal{N}_+ + 1 )^{(k+2)/2} \psi \|  + \vert\lambda\vert^3 \| g \|^2 \sqrt{N} \| ( \mathcal{N}_+ +1)^{k/2}\psi \| \right) \notag \\
 &+   C^m \| \mu \|^m \left( \vert\lambda\vert  ( \| g \| + \vert \kappa \vert ) N^{-1/2} \| b_p( \mathcal{N}_+ + 1 )^{(k+1)/2} \psi \|  + \vert\lambda\vert^3 \| g \|^3  \sqrt{N} \| b_p( \mathcal{N}_+ +1)^{(k-1)/2}\psi \| \right) \; . 
\end{align}
Plugging these estimates into \eqref{eq:term20} we arrive for sufficiently small $\| \mu \| $ at Lemma \ref{lemma:conj-dp}. The second estimate of Lemma \ref{lemma:conj-dp} follows similarly. 
\end{proof}

\begin{lemma}
\label{lemma:conj-dd}
Under the same assumptions as in Lemma \ref{lemma:conj-dp}, there exists $C>0$ (independent of $\lambda, \kappa$) such that 
\begin{align}
\| ( \mathcal{N}_+ + 1)^{-k/2} & \left( e^{\kappa \lambda \cN_+} e^{\lambda \sqrt{N} \phi_+ (g)} d_p^{\sharp_1} d_{\alpha p}^{\sharp_2} e^{- \lambda \sqrt{N} \phi_+ (g))} e^{- \lambda \kappa\cN_+ } -  d_p^{\sharp_1} d_p^{\sharp_2} \right) \psi \|  \notag \\
& \leq C  (\| g \|+ \vert \kappa \vert ) \vert\lambda\vert  \| ( \mathcal{N}_+ +1 )^{1-k/2} \psi \| + C \vert\lambda\vert^3 \|g \|^3 N \| \psi \| 
\end{align}
and similarly
 \begin{align}
\|   ( \mathcal{N}_+ + 1)^{-k/2}  & \left( e^{\kappa \lambda \cN_+} e^{\lambda \sqrt{N} \phi_+ (g)} d_p^{\sharp_1} b_{\alpha p}^{\sharp_2} e^{- \lambda \sqrt{N} \phi_+ (g))} e^{- \lambda \kappa\cN_+ } -  d_p^{\sharp_1} b_p^{\sharp_2} \right) \psi \| \notag \\
\leq& C  (\| g \|+ \vert \kappa \vert ) \vert\lambda\vert  \| ( \mathcal{N}_+ +1 )^{1-k/2} \psi \| + C \vert\lambda\vert^3 \|g \|^3 N \| \psi \|   \notag \\
\|   ( \mathcal{N}_+ + 1)^{-k/2} & \left( e^{\kappa \lambda \cN_+} e^{\lambda \sqrt{N} \phi_+ (g)} b_p^{\sharp_1} d_{\alpha p}^{\sharp_2} e^{- \lambda \sqrt{N} \phi_+ (g))} e^{- \lambda \kappa\cN_+ } -  b_p^{\sharp_1} d_p^{\sharp_2} \right) \psi \|  \notag \\
\leq&  C  (\| g \|+ \vert \kappa \vert ) \vert\lambda\vert  \| ( \mathcal{N}_+ +1 )^{1-k/2} \psi \| + C \vert\lambda\vert^3 \|g \|^3 N \| \psi \| 
\end{align}
with $\sharp_i \in \lbrace \cdot, * \rbrace$ for $i=1,2$ either $\sharp_1 = \sharp_2$ or $\sharp_1 =*$ and $\sharp_2 = \cdot$ and $\alpha = -1$ if $\sharp_1 = \sharp=2$ and $\alpha =1$ otherwise. 
\end{lemma}

\begin{proof}
We start with $\sharp_1 = \sharp_2 = * $. We observe that from \eqref{def:d} we have 
\begin{align}
&  \left( e^{\kappa \lambda \cN_+} e^{\lambda \sqrt{N} \phi_+ (g)} d_p^{*} d_{-p}^{*} e^{- \lambda \sqrt{N} \phi_+ (g))} e^{- \lambda \kappa\cN_+ } -  d_p^{*} d_{-p}^{*}\right)\notag \\
 & \quad = \left( e^{\kappa \lambda \cN_+} e^{\lambda \sqrt{N} \phi_+ (g)} d_p^{*} 
 e^{- \lambda \sqrt{N} \phi_+ (g))} e^{- \lambda \kappa\cN_+ } -  d_p^{*} \right) e^{\kappa \lambda \cN_+} e^{\lambda \sqrt{N} \phi_+ (g)} d_{-p}^{*} 
 e^{- \lambda \sqrt{N} \phi_+ (g))} e^{- \lambda \kappa\cN_+ }  \notag \\
 &\quad  \quad + d_p^* \left( e^{\kappa \lambda \cN_+} e^{\lambda \sqrt{N} \phi_+ (g)} d_{-p}^{*} 
 e^{- \lambda \sqrt{N} \phi_+ (g))} e^{- \lambda \kappa\cN_+ }  - d_{-p}^* \right) 
\end{align}
and thus 
\begin{align}
\|  & ( \mathcal{N}_+ + 1)^{-k/2} \left( e^{\kappa \lambda \cN_+} e^{\lambda \sqrt{N} \phi_+ (g)} d_p^{*} d_{-p}^{*} e^{- \lambda \sqrt{N} \phi_+ (g))} e^{- \lambda \kappa\cN_+ } -  d_p^{*} d_{-p}^{*}\right) \psi \| \notag \\
\leq& \| ( \mathcal{N}_+ + 1)^{-k/2} \left( e^{\kappa \lambda \cN_+} e^{\lambda \sqrt{N} \phi_+ (g)} d_p 
 e^{- \lambda \sqrt{N} \phi_+ (g))} e^{- \lambda \kappa\cN_+ } -  d_p \right) \notag \\
 &\hspace{2cm} \times  e^{\kappa \lambda \cN_+} e^{\lambda \sqrt{N} \phi_+ (g)} d_{-p}^{*} 
 e^{- \lambda \sqrt{N} \phi_+ (g))} e^{- \lambda \kappa\cN_+ }  \psi \| \notag \\
 &+ \| ( \mathcal{N}_+ + 1)^{-k/2}  d_p^* \left( e^{\kappa \lambda \cN_+} e^{\lambda \sqrt{N} \phi_+ (g)} d_{-p}^{*} 
 e^{- \lambda \sqrt{N} \phi_+ (g))} e^{- \lambda \kappa\cN_+ }  - d_{-p}^* \right)  \psi \| 
 \end{align}
With \eqref{eq:estimates-dp} and Lemma \ref{lemma:conj-dp} we find for all $\vert\lambda \vert, \vert \kappa \lambda \vert \leq 1$ 
 \begin{align}
\|  & ( \mathcal{N}_+ + 1)^{-k/2}  \left( e^{\kappa \lambda \cN_+} e^{\lambda \sqrt{N} \phi_+ (g)} d_p^{*} d_{-p}^{*} e^{- \lambda \sqrt{N} \phi_+ (g))} e^{- \lambda \kappa\cN_+ } -  d_p^{*} d_{-p}^{*}\right) \psi \| \notag \\
 \leq& C \sqrt{N} \vert\lambda\vert^3 \|g \|^3 \|  e^{\kappa \lambda \cN_+} e^{\lambda \sqrt{N} \phi_+ (g)} d_{-p}^{*} 
 e^{- \lambda \sqrt{N} \phi_+ (g))} e^{- \lambda \kappa\cN_+ }  \psi \|  \notag \\
 &+  C \vert \lambda\vert  ( \| g \| + | \kappa | ) N^{-1/2} \| ( \mathcal{N}_+ + 1)^{(2-k)/2}  e^{\kappa \lambda \cN_+} e^{\lambda \sqrt{N} \phi_+ (g)} d_{-p}^{*} 
 e^{- \lambda \sqrt{N} \phi_+ (g))} e^{- \lambda \kappa\cN_+ }  \psi \|\notag \\
  &+ CN^{-1/2} \| ( \mathcal{N}_+ + 1)^{(2-k)/2} \left( e^{\kappa \lambda \cN_+} e^{\lambda \sqrt{N} \phi_+ (g)} d_{-p}^{*} 
 e^{- \lambda \sqrt{N} \phi_+ (g))} e^{- \lambda \kappa\cN_+ }  - d_{-p}^* \right)  \psi \| \notag \\
\leq&  C N \vert\lambda\vert^3 \| g \|^3 \| \psi \| + C N^{-1} \vert\lambda\vert  ( \| g \| + | \kappa | ) \| ( \mathcal{N}_+ + 1)^{(4-k)/2} \psi \| \; . 
  \end{align}
The remaining bounds follow similarly with \eqref{eq:bounds-b} and Lemmas \ref{lemma:b-con}, \ref{lemma:conj-dp}. 
\end{proof}

Additionally, we consider the conjugation of the kinetic energy with the generalized Bogoliubov transform that we write as 
\begin{align}
e^{B( \mu)} \sum_{p \in \Lambda_+^*}&  p^2 a_p^*a_p e^{-B( \mu)} \notag \\
&= \sum_{p \in \Lambda_+^*} p^2 a_p^*a_p + \sum_{p \in \Lambda_+^*} p^2 \left( \sigma_p^2 + \sigma_p \gamma_p ( b_p^*b_{-p}^* + b_p^*b_{-p}^* ) + 2 \sigma_p^2 b_p^*b_p \right)  + \mathcal{R}_{\mathcal{K}}
\end{align}
where the remainder $\mathcal{R}_{\mathcal{K}}$ satisfies the following properties. 

\begin{lemma}
\label{lemma:RK}
Under the same assumptions as in Lemma \ref{lemma:dp}, \ref{lemma:conj-dp}, let $p^2 \mu \in \ell^1( \Lambda_+^*) $ and $\| \mu \|$ small enough. Then there exists $C>0$ such that 
\begin{align}
\|( \mathcal{N}_+ + 1)^{-1/2}   \mathcal{R}_{\mathcal{K}} \psi \| &\leq C  N^{-1/2} \| ( \mathcal{N}_+ + 1) \psi \| \; . \label{eq:bound-RK1}  \\
\| \left[ ( \mathcal{N}_+ + 1)^{1/2}, \left[ \mathcal{N}_+ +1)^{1/2},  \mathcal{R}_{\mathcal{K}} \right] \right] \psi \| &\leq C  \| ( \mathcal{N}_+ + 1) \psi \| \; . \label{eq:bound-RK3} 
\end{align}Ffurthermore for $\vert\lambda \vert, \vert \kappa \lambda \vert \leq 1$ there exists $C>0$ (independent of $\lambda, \kappa$) such that
\begin{align}
\| & \left( e^{\kappa \lambda \cN_+} e^{\lambda \sqrt{N} \phi_+ (g)} \mathcal{R}_{\mathcal{K}} e^{- \lambda \sqrt{N} \phi_+ (g))} e^{- \lambda \kappa\cN_+ } -  \mathcal{R}_{\mathcal{K}} \right) \psi \| \notag \\
&\leq  C  N \vert\lambda\vert^3 \| g \|^3 \| \psi \| + C \vert\lambda\vert (\| g \|+ \vert \kappa \vert )\| ( \mathcal{N}_+ + 1) \psi \| \; .  \label{eq:bound-RK2}
\end{align}

\end{lemma}

\begin{proof}
We compute 
\begin{align}
e^{- B( \mu)} \mathcal{K} e^{B ( \mu)}  =& \mathcal{K} + \int_0^1 ds \; \frac{d}{ds} e^{-sB( \mu)} \mathcal{K} e^{s B( \mu)} \notag \\
=&\mathcal{K} + \int_0^1 ds \;  e^{-s B(\mu)} [\mathcal{K}, B(\mu)] e^{s B(\mu)} \notag \\
=& \mathcal{K} + \int_0^1 ds \; \sum_{p \in \Lambda_+^*} p^2 \mu_p e^{-sB(\mu)} \left( b_pb_{-p} + b_p^* b_{-p}^* \right) e^{s B ( \mu)} \; \notag \\
 =& \mathcal{K} + \sum_{p \in \Lambda_+^*} \sigma_p^2 + \sum_{p \in \Lambda_+^*} p^2 \gamma_p \sigma_p ( b_pb_{-p}^2 + b_p^*b_{-p}^*) + 2 \sum_{p \in \Lambda_+^*} b_p^* b_p +  \mathcal{R}_{\mathcal{K}}
\end{align}
where 
\begin{align}
\mathcal{R}_{\mathcal{K}} =& \sum_{n,m \geq 0} \frac{(-1)^{n+m}}{n!m! (n+m+1)} \sum_{p \in \Lambda_+^*} p^2 \mu_p^{n+1} b_{\alpha_n p}^{\sharp_n} \left[ {\rm ad}^{(m)}_{B( \mu)}(b_{-p}) - \mu_p^m b^{\sharp_m}_{- \alpha_m p} \right] + {\rm h.c.}\notag \\
&+ \sum_{n,m \geq 0} \frac{(-1)^{n+m}}{n!m! (n+m+1)} \sum_{p \in \Lambda_+^*} p^2 \mu_p \left[ {\rm ad}^{(n)}_{B( \mu)}(b_{-p}) - \mu_p^n b^{\sharp_n}_{- \alpha_n p} \right] \mu_p^{m+1} b_{\alpha_m p}^{\sharp_m} + {\rm h.c.} \notag \\
&+ \sum_{n,m \geq 0} \frac{(-1)^{n+m}}{n!m! (n+m+1)} \sum_{p \in \Lambda_+^*} p^2 \mu_p \left[ {\rm ad}^{(n)}_{B( \mu)}(b_{-p}) - \mu_p^n b^{\sharp_n}_{- \alpha_n p} \right] \left[ {\rm ad}^{(m)}_{B( \mu)}(b_{-p}) - \mu_p^m b^{\sharp_m}_{- \alpha_n p} \right]  
\end{align}
We recall that it follows with the same arguments as in Lemma \ref{lemma:conj-dp} (see for example \cite[Lemma 3.4]{BBCS_optimal} that
\begin{align}
\| ( \mathcal{N}_+ +1)^{-1/2} \left( {\rm ad}^{(m)}_{B( \mu)}(b_{-p}) - \mu_p^m b^{\sharp_m}_{- \alpha_n p} \right) \psi \| \leq CN^{-1} \| ( \mathcal{N}_+ + 1)^{2} \psi \|
\end{align}
for all $p \in \Lambda_+^*$. Since $p^2 \mu \in \ell^1 ( \Lambda_+^*)$ by assumption,  the first estimate \eqref{eq:bound-RK1} follows. This estimates remains true for the double commutator, too (see \cite[Lemma 3.4]{BBCS_optimal}) and thus \eqref{eq:bound-RK3} follows. For the second estimate \eqref{eq:bound-RK2}, we recall that in the proof of Lemma \ref{lemma:conj-dp} we more precisely prove that 
\begin{align}
\|  \Big( & e^{\kappa \lambda \cN_+} e^{\lambda \sqrt{N} \phi_+ (g)}  \left[ \ad_{-B (\mu)}^{(m)} (b_q) - \mu_q^m b_{\alpha_m p}^{\sharp_m} \right] e^{- \lambda \sqrt{N} \phi_+ (g)} e^{- \lambda \kappa\cN_+ } -  \left[ \ad_{-B (\mu)}^{(m)} (b_q) \mu_q^m b_{\alpha_m p}^{\sharp_m} \right]  \Big) \psi \| \notag \\
 \leq&  C^m \vert\lambda\vert (\| g \|+ \vert \kappa \vert ) \| ( \mathcal{N}_+ + 1) \psi \| + C^m \| g \|^3 \vert\lambda\vert^3 N \| \psi \| \; \| \psi \| \; . 
\end{align}
and thus \eqref{eq:bound-RK2} follows. 
\end{proof}

 \section{Proof of Theorem \ref{thm:main}} 
 \label{sec:proof}

In this section we prove Theorem \ref{thm:main}, thus we estimate the logarithmic moment generating function. For this we define the centered (w.r.t. to the condensate's expectation value) operator 
\begin{align}
\widetilde{O} := O - \langle \varphi_0,  O \varphi_0 \rangle
\end{align}
and recall that we need to compute the moment generating function 
\begin{align}
\mathbb{E}_{\psi_{N}}  \left[ e^{\lambda O_{N} } \right]  = \langle \psi_{N}, \; e^{ \lambda O_N } \psi_{N} \rangle \; 
\end{align}
We consider the embedding of $\psi_{N} \in L_s^2 ( \mathbb{R}^{3N})$ in the full bosonic Fock space where we have the identity 
\begin{align}
\label{eq:O_N-Gamma}
O_N = \sum_{p, q \in \Lambda^*}  \widehat{\widetilde{O}}_{p,q}  \; a_p^* a_{-q} \; . 
\end{align}
where $\widehat{\widetilde{O}}_{p,q}$ denotes the Fourier coefficients of $\widetilde{O}$, i.e. $\widehat{\widetilde{O}}_{p,q} = \int_{\Lambda \times \Lambda} dxdy \; \widetilde{O}(x;y) e^{i(px+qy)}$. By definition of $\mathcal{U}_{N}$ in \eqref{def:UN} we observe that we can write $\psi_{N}$ as 
\begin{align}
\psi_{N} =\mathcal{U}_{N} \psi_{\mathcal{G}_N}
\end{align}
where $\psi_{\mathcal{G}_N}$ denotes the ground state of the excitation Hamiltonian $\mathcal{G}_N$ defined in \eqref{def:G}. The properties \eqref{eq:prop-U1}, \eqref{eq:prop-U2} of the unitary $\mathcal{U}_{N}$ show that 
\begin{align}
\mathcal{U}_{N} &  \sum_{p, q \in \Lambda^*}  \widehat{\widetilde{O}}_{p,q}  \; a_p^* a_{-q} \;  \mathcal{U}_{N}^* = \sum_{p,q \in \Lambda_+^*} \widehat{\widetilde{O}}_{p,q} a_p^*a_{-q}   + \sqrt{N} \phi_+ ( \widehat{O \varphi_0} )  
\end{align}
where we recall the notation \ref{def:phi}. Furthermore we introduce the notation 
\begin{align}
\label{def:g,B}
g = \widehat{O \varphi_0}  \quad \text{and} \quad  B = \sum_{p,q \in \Lambda_+^*}\widehat{\widetilde{O}}_{p,q}  \; a_p^*a_{-q} 
\end{align}
and thus arrive at 
\begin{align}
\mathbb{E}_{\psi_{N}}  \left[ e^{\lambda O_N } \right]  = \langle \psi_{\mathcal{G}_N}  ,\;   \; e^{ \lambda \sqrt{N}\phi_+ ( g ) +  \lambda  B } \psi_{\mathcal{G}_N}  \rangle \;  \label{eq:char-1}
\end{align}
In the following we will compute the expectation value of the r.h.s. of  \eqref{eq:char-1}. First we will show that the operator $B$ contributes to our analysis sub-leading only (see Lemma \ref{lemma:step1}). This will be based on ideas introduced in \cite{KRS,RSe}. Second we will show that the ground state $\psi_{\mathcal{G}_N}$ of the excitation Hamiltonian $\mathcal{G}_N$ (defined in \eqref{def:G}) approximately behaves as the ground state $\psi_{\mathcal{Q}}$ of the excitation Hamiltonian's quadratic approximation $\mathcal{Q}$ (defined in \eqref{def:Q}) (Lemma \ref{lemma:step2}). Then we show that $\psi_{\mathcal{Q}}$ effectively acts as a Bogoliubov transformation on the observable $\phi_+ (g)$. We remark that this would be an immediate consequence if the operator $\phi_+$ defined in \eqref{def:phi} would be formulated w.r.t. to standard creation and annihilation operators. However, $\phi_+$ is formulated w.r.t. to modified creation and annihilation operators that lead to more involved calculations (see Lemma \ref{lemma:step3}). Finally, in the last step, we compute the remaining expectation value (see Lemma \ref{lemma:step4}). 

While the first and the forth step are based on ideas presented in \cite{KRS} for the dynamical problem, the second and third step use novel ideas and techniques based on the Hellmann-Feynmann theorem and Gronwall's inequality. 

\subsection{Step 1} In this step we show that the operator $B$ defined in \eqref{def:g,B} contributes to the expectation value \eqref{eq:char-1} exponentially cubic in $\lambda$ only. This Lemma follows closely the proof of \cite[Lemma 3.3]{RSe} resp. \cite[Lemma 3.1]{KRS} considering a similar result for the dynamics in the mean-field regime ($\beta =0$). The results \cite{KRS,RSe} are formulated in position space, however the proofs and results easily carried over to momentum space. 

\begin{lemma}
\label{lemma:step1} Under the same assumptions as in Theorem \ref{thm:main} there exists $C>0$ such that for all $0 \leq \lambda \leq 1/\| O \|$ we have
\begin{align}
& e^{- C N \| O \|^3 \lambda^3}  \langle \psi_{\mathcal{G}_N}, \; e^{\lambda \sqrt{N} \phi_+ (g)/2} e^{- 2 \lambda \| O \| \mathcal{N}_+} e^{\lambda \sqrt{N} \phi_+ (g)/2}  \psi_{\mathcal{G}_N}\rangle \notag \\
&\quad  \leq  \langle \psi_{\mathcal{G}_N}, \; e^{\lambda \sqrt{N} \phi_+ (g) + \lambda B} \psi_{\mathcal{G}_N} \rangle  \notag \\
&\quad \leq e^{C  N \| O \|^3 \lambda^3 }  \langle \psi_{\mathcal{G}_N}, \;  e^{\lambda \sqrt{N} \phi_+ (g)/2} e^{ 2 \lambda \| O \| \mathcal{N}_+} e^{\lambda \sqrt{N} \phi_+ (g)/2} \psi_{\mathcal{G}_N} \rangle \; . 
\end{align}
\end{lemma}

\begin{proof} We start with the lower bound (i.e. the first inequality of Lemma \ref{lemma:step1}) and define similarly to \cite[Lemma 3.3]{RSe} for $s \in [0,1]$ and $\kappa>0$ the Fock space vector 
\begin{align}
\label{def:xi1}
\xi (s) := e^{-(1-s) \lambda \kappa\mathcal{N}_+/2} e^{(1-s) \lambda\sqrt{N} \phi_+ (g)/2} e^{s \lambda \left[ B + \sqrt{N} \phi_+ (g) \right]/2} \psi_{\mathcal{G}_N} \; . 
\end{align}
We remark that by construction $\xi_N(s)$ is an element of the Fock space of excitations $\mathcal{F}_{\perp \varphi_{N}}^{\leq \varphi_0}$, thus the number of particles of $\xi(s)$ is at most $N$. This observation will be crucial for our analysis later. Since we have for $s=0$ 
\begin{align}
\label{eq:xi-10}
\| \xi (0) \|^2 =\left\langle \psi_{\mathcal{G}_N},  \;  e^{\lambda \sqrt{N} \phi_+ (g )/2 } e^{-\kappa \lambda  \cN_+ } e^{\lambda \sqrt{N} \phi_+ (g)/2 } \psi_{\mathcal{G}_N} \right\rangle
\end{align}
and for $s=1$ 
\begin{align}
\label{eq:xi-11}
\| \xi (1) \|^2 =  \left\langle \psi_{\mathcal{G}_N}, \; e^{ \lambda \sqrt{N} \phi_+ (g) + \lambda B} \psi_{\mathcal{G}_N}\right \rangle , 
\end{align}
it suffices to control the difference of \eqref{eq:xi-10} and \eqref{eq:xi-11} to get the desired estimate. We aim to control their difference through the derivative 
\begin{align}
\label{eq:Re-1}
\partial_s \| \xi (s) \|^2 = 2 \Re \langle \xi (s) , \,  \partial_s \xi (s) \rangle = 2 \Re \langle \xi (s), \,  \cM (s) \xi (s) \rangle = \Re \langle \xi (s), \; \left( \mathcal{M} (s)  + \mathcal{M} (s) ^* \right) \xi (s)  \rangle 
\end{align}
with the operator $\cM_s$ given by
\begin{align}
\label{eq:def-M}
\cM (s) =& \frac{\lambda}{2} e^{-(1-s)\lambda \kappa \cN_+/2} e^{(1-s)\lambda \sqrt{N} \phi_+ (g )/2} \;  B  \; e^{-(1-s)\lambda \sqrt{N} \phi_+ (g )/2} e^{(1-s)\lambda \kappa \cN_+ /2} + \frac{\lambda \kappa}{2} \cN_+  .
\end{align} 
The results from \cite[Propositions~2.2--2.4]{KRS} (summarized in Lemma \ref{lemma:b-con}, \ref{prop:eN}) provide formulas to compute the operator $\cM (s) $ explicitly. We use the short-hand notation \eqref{eq:gamma} and  $h (s) = (1-s) \lambda g $  and arrive at 
 \begin{align}
 \label{eq:M1}
&\frac{\cM (s) + \cM (s)^*}{\lambda} \notag \\
 & \quad =  \sum_{p,q \in \Lambda_+^*} \widehat{\widetilde{O}}_{p,q} a_p^* a_{-q} +  \kappa \cN_+ \notag\\
 &\quad + \frac{\gamma_{\| h(s) \|} -1}{\| h(s) \|^2}\sum_{p,q, k \in \Lambda_+^*} \left[  h_p (s)  \widehat{\widetilde{O}}_{q,k} h_k (s)  a^*_p a _{-q} +  h_q (s)  \widehat{\widetilde{O}}_{p,k} h_k (s)  a^*_p a _{-q} \right] \notag \\
  &\quad \quad -  \frac{\sigma_{\| h(s) \|}^2}{\| h(s) \|^2} \langle h(s) , \widetilde{O} h(s)  \rangle \left(N - \cN_+ \right) \notag\\
  &\quad\quad +  \left( \frac{\gamma_{\| h(s) \|} - 1}{\| h(s) \|^2} \right)^2  \langle h(s) , \; \wO  h(s)  \rangle \;  \sum_{p,q \in \Lambda_+^*} h_p(s) h_{q} (s) a_p^* a_{-q} \notag \\
  &\quad\quad  + \sqrt{N} \,  \frac{\sigma_{\| h(s) \|}}{\| h(s) \|}
  \sinh ((s-1) \lambda \kappa /2)  \left[ \frac{\gamma_{\| h(s) \|} - 1}{\| h(s) \|^2} \langle h(s) , \wO h(s) \rangle \phi_+ (h (s))  + \phi_+ (\widehat{\widetilde{O}} h (s)) \right]  \; .  
 \end{align}
With the bounds \eqref{eq:bounds-b} for any Fock space vector $\xi \in \mathcal{F}_{\perp \varphi_0}^{\leq N}$ and 
\begin{align}
\label{eq:boundh}
\| h (s) \|_2 \leq \lambda \|  O \varphi_0\|_2 \leq  \lambda \| O \|  \,  \|\varphi_0\|_2\leq 1, \quad \| \widetilde{O} \| \leq  \| O \| \left( 1 + \| \varphi_0 \|_2^2 \right) = 2 \| O \| 
\end{align} 
for all $\lambda \| O \| \leq 1$, we observe that that all but the terms of the first line of the r.h.s. of \eqref{eq:M1} are bounded by $C \lambda^2 N $. For the first line, however, we find with the choice $\kappa = 2 \| O \|$
\begin{align}
 \sum_{p,q \in \Lambda_+^*} \widehat{\widetilde{O}}_{p,q} a_p^* a_{q}  +  \kappa \cN_+  \geq  (- 2 \| O \|  + \kappa ) \cN_+ \geq 0
\end{align}
as operator inequality on the Fock space of excitations $\mathcal{F}_{\perp \varphi_0}^{\leq N}$. Summarizing, we arrive at 
\begin{align}
& \frac{\cM (s) + \cM^* (s)}{\lambda} \geq C  \lambda^2 N 
\end{align}
again as operator inequality on $\mathcal{F}_{\perp \varphi_0}^{\leq N}$ that yields 
 \begin{align}
 \label{eq:lower-1}
\frac{2}{\lambda} \text{Re } \langle \xi (s)  , \cM (s) \;  \xi (s) \rangle \geq 
- C  \lambda^2 N \| O \|^3   \| \xi (s)  \|^2 \,.
 \end{align}
 In combination with \eqref{eq:Re-1} the lower bound from Lemma \ref{lemma:step1} now follows from Gronwall's inequality. 

The upper bound is proven with a similar strategy (see also \cite{KRS} for more details) replacing the constant $\kappa$ in the definition of the Fock space vector $\xi (s)$ in \eqref{def:xi1} by $-\kappa$ and estimating the terms of \eqref{eq:M1} from above instead of from below. 
\end{proof}

\subsection{Step 2.} The goal of the second step is to show that we can replace $\psi_{\mathcal{G}_N}$, the ground state of the excitation Hamiltonian $\mathcal{G}_N $ with the ground state $\psi_{\mathcal{Q}}$ of its quadratic approximation $\mathcal{Q}$. The idea is to use the family of Hamiltonians $\lbrace \mathcal{G}_N (s) \rbrace_{s \in [0,1]}$ defined in \eqref{def:Gs} interpolating between the excitation Hamiltonian $\mathcal{G}_N = \mathcal{G}_N (1)$ and its corresponding quadratic approximation $\mathcal{Q} = \mathcal{G}_N (0)$. We remark that Proposition \ref{claim:gs} summarizes useful properties of the Hamiltonians $\lbrace \mathcal{G}_N (s) \rbrace_{s \in [0,1]}$ and their corresponding ground states $\lbrace \psi_{\mathcal{G}_N (s)} \rbrace_{s \in [0,1]}$ that will be crucial for the proof of the following Lemma. 

This Lemma's proof is crucially different from the proof of \cite{KRS,RSe} where the analogous step was based on properties of the dynamical evolution.

\begin{lemma} 
\label{lemma:step2}
Under the same assumptions as in Theorem \ref{thm:ldp}, there exist constants $C_1, C_2 >0$ and $\kappa_1, \kappa_2$ such that for all $0 \leq \lambda \leq \min \lbrace 1/(\kappa_1 \vertiii{O}), 1/(\kappa_2 \vertiii{O}) \rbrace$ we have 
\begin{align}
& \langle \psi_{\mathcal{G}_N}, \; e^{\lambda \sqrt{N} \phi_+ (g)/2} e^{  2 \lambda \| O \| \mathcal{N}_+} e^{\lambda \sqrt{N} \phi_+ (g)/2}  \psi_{\mathcal{G}_N}\rangle  \notag \\
&\quad \leq e^{ C_1 ( N \lambda^3  + \lambda)}  \langle \psi_{\mathcal{Q}}, \; e^{\lambda \sqrt{N} \phi_+ (g)/2} e^{ \kappa_1 \lambda \vertiii{O} \mathcal{N}_+} e^{\lambda \sqrt{N} \phi_+ (g)/2}  \psi_{\mathcal{Q}}\rangle  
\end{align}
resp. 
\begin{align}
&  \langle \psi_{\mathcal{G}_N}, \; e^{\lambda \sqrt{N} \phi_+ (g)/2} e^{- 2 \lambda \| O \| \mathcal{N}_+} e^{\lambda \sqrt{N} \phi_+ (g)/2}  \psi_{\mathcal{G}_N}\rangle  \notag \\
&\quad \geq e^{- C_2 ( N \lambda^3 + \lambda)}  \langle \psi_{\mathcal{Q}}, \; e^{\lambda \sqrt{N} \phi_+ (g)/2} e^{- \kappa_2 \lambda \vertiii{O} \mathcal{N}_+} e^{\lambda \sqrt{N} \phi_+ (g)/2}  \psi_{\mathcal{Q}}\rangle  
\end{align}
\end{lemma}

\begin{proof} We start with the lower bound, i.e. the second inequality of Lemma \ref{lemma:step2}. The upper bound then follows with similar arguments. 

We consider the two families of Hamiltonians $\lbrace \mathcal{G}_N (s) \rbrace_{s \in [0,1]}$ defined in \eqref{def:Gs}. We shall prove first that denoting with $\psi_{\mathcal{G}_N(s)}$ the ground state of $\mathcal{G}_N (s)$ 
\begin{align}
\label{eq:lemma2-2}
&\langle \psi_{\mathcal{G}_N (1)}, \; e^{\lambda \sqrt{N} \phi_+ (g)/2} e^{  - 2 \lambda \| O \| \mathcal{N}_+} e^{\lambda \sqrt{N} \phi_+ (g)/2}  \psi_{\mathcal{G}_N (1)}\rangle  \notag \\
&\quad \geq e^{ - C_2 ( N \lambda^3 + \lambda)}  \langle \psi_{\mathcal{G}_N(0)}, \; e^{\lambda \sqrt{N} \phi_+ (g)/2} e^{- \kappa_2 \lambda \vertiii{O} \mathcal{N}_+} e^{\lambda \sqrt{N} \phi_+ (g)/2}  \psi_{\mathcal{G}_N (0)}\rangle  
\end{align}
for some constants $C_2,\kappa_2 >0$ which together with the observation $\mathcal{G}_N(0) = \mathcal{Q}$ yields the lower bound of Lemma \ref{lemma:step2}. 

For $s \in [0,1]$ let $\mathcal{G}_N (s) $ denote the Hamiltonian defined in \eqref{def:Gs} with corresponding ground state $\psi_{\mathcal{G}_N (s)}$. Then we define the Fock space vector 
\begin{align}
\xi (s) = e^{- \lambda \kappa_s  \mathcal{N}_+/2} e^{\lambda \phi_+ (g) /2} \psi_{\mathcal{G}_N (s)} \; . 
\end{align}
where $\kappa_s : [0,1] \rightarrow \mathbb{R}$ denotes a differentiable positive function with $\kappa_1 = 2 \| O \|$ chosen later. We remark that it follows from Proposition \ref{claim:gs} that $\psi_{\mathcal{G}_N(s)} \in \mathcal{F}_{\perp \varphi_0}^{\leq N}$ for all $s \in [0,1]$ and thus $\xi_N (s) \in \mathcal{F}_{\perp \varphi_0}^{\leq N}$. Moreover,  
\begin{align}
\label{eq:xi-1}
\| \xi (1)  \|_2^2 = \left\langle  \psi_{\mathcal{G}_N}, \; e^{ \lambda \sqrt{N}\phi_+ ( g) /2 } e^{ -2 \lambda \| O \| \mathcal{N}_+ } e^{\lambda \sqrt{N} \phi_+ ( g ) /2 }  \psi_{\mathcal{G}_N} \right\rangle \;. 
\end{align}
and 
\begin{align}
\label{eq:xi-0}
\| \xi (0) \|_2^2 = \left\langle \psi_{\mathcal{G}_N (0)} , \; e^{ \lambda \sqrt{N}\phi_+ (g ) /2 } e^{ - \lambda \kappa_0 \mathcal{N}_+ } e^{\lambda \sqrt{N} \phi_+ ( g) /2 } \psi_{\mathcal{G}_N (0)} \right\rangle \; . 
\end{align}
Thus we are left with controlling the difference of \eqref{eq:xi-1} and \eqref{eq:xi-0} for which we shall use estimates on the derivative
\begin{align}
\label{eq:step2-1}
\partial_s \| \xi (s) \|^2 = 2 \Re \langle \xi (s) , \; \partial_s \xi (s) \rangle \; . 
\end{align}
As a preliminary step towards computing the derivative of $\xi (s)$ we compute with the Hellmann-Feynmann theorem the ground states derivative given with the notation $q_{\psi_{\mathcal{G}_N (s)}} := 1 - \vert \psi_{\mathcal{G}_N (s)} \rangle \langle \psi_{\mathcal{G}_N (s)} \vert$ by 
\begin{align}
\label{eq:deriv-gs-step1}
  \partial_s \psi_{\mathcal{G}_N (s)} = \frac{q_{\psi_{\mathcal{G}_N (s)}}}{\cG_N (s) - E_N(s) }\mathcal{R}_N \;  \psi_{\mathcal{G}_N (s)} \; . 
\end{align}
Proposition \ref{claim:gs} ensures that the reduced resolvent $ \frac{q_{\psi_{\mathcal{G}_N (s)}}}{\cG_N (s) - E_N(s) }$ is well defined for all $s \in [0,1]$ and, in particular, bounded from above independent of $N,s$. 
We remark that by Proposition \ref{claim:gs} the r.h.s. of \eqref{eq:deriv-gs-step1} is, Eq. \eqref{def:Q} and \eqref{eq:bounds-b} in norm by bounded by a constant.  However we can not bound the derivative in norm here, but we need to compute the conjugation of the operators of the r.h.s. of \eqref{eq:deriv-gs-step1} with the exponentials of $\sqrt{N} \lambda\phi_+ (g ), \mathcal{N}_+$ to then bound the operators of the r.h.s. of \eqref{eq:deriv-gs-step1} in form. To this end, we introduce the splitting $ \mathcal{R}_N = \sum_{j=1}^2 \mathcal{R}_N^{(j)}$  given by 
\begin{align}
\label{def:Rj}
\mathcal{R}_N^{(1)} =& \frac{1}{\sqrt{N}} \sum_{p,q \in \Lambda_+^*, p \not= q} \widehat{v}(q) \left( b_{p+q}^* a_{-q}^* a_p + {\rm h.c.} \right) \notag \\
\mathcal{R}_N^{(3)} =&  \frac{1}{2N} \sum_{p,q,k \in \Lambda_+^* } \widehat{v} (k) a_{p-k}^* a_{q+k}^* a_p a_q \; 
\end{align}
We remark that $\mathcal{G}_N (s)$ leaves the truncated Fock space invariant (as it is formulated w.r.t. to modified creation and annihilation operators only). Moreover, we recall that it follows from Lemma \ref{claim:gs} that 
\begin{align}
\label{eq:bounds-tildeG}
\Big \| \frac{q_{\psi_{\mathcal{G}_N (s)}}}{\cG_N (s) - E_N(s) }  \Big\|,  \quad \| \frac{q_{\psi_{\mathcal{G}_N (s)}}}{\cG_N (s) - E_N(s) } ( \mathcal{N}_+ + 1)\Big\| \leq C  
\end{align}
and 
\begin{align}
\label{eq:bounds-tildeG2}
 \Big\|( \mathcal{N}_+ + 1)^{-1/2} \frac{q_{\psi_{\mathcal{G}_N (s)}}}{\cG_N (s) - E_N(s) } ( \mathcal{N}_+ + 1)^{3/2} \Big\| \leq C  \; . 
\end{align}
We use these properties now in the following to estimate the derivative
\begin{align}
\label{eq:deriv-norm1}
\partial_s \| \xi_N(s) \|^2 = 2 \Re \langle \xi_N (s), \; \mathcal{M} (s) \xi_N (s) \rangle 
\end{align}
where the operator $\cM (s)$ is with \eqref{eq:deriv-gs-step1} given by 
\begin{align}
\cM (s) =&   e^{-  \lambda \kappa_s  \mathcal{N}/2} e^{\lambda \sqrt{N} \phi_+ ( q O \varphi ) /2 } \frac{q_{\psi_{\mathcal{G}_N (s)}}}{\cG_N (s) - E_N(s)}  \mathcal{R}_N   e^{-\kappa_s \lambda \sqrt{N} \phi_+ ( g ) /2 } e^{\kappa_s \lambda\mathcal{N}_+ /2}  \; \notag \\
&- \lambda  \dot{\kappa}_s \mathcal{N}_+ \; .   \label{def:M1}
\end{align}
It follows that denoting with
\begin{align}
\widetilde{ \mathcal{G}}_N(s) = e^{ - \lambda \kappa_s \mathcal{N}_+ /2} e^{\lambda \sqrt{N} \phi_+ ( g ) /2 } \mathcal{G}_N(s) e^{-\lambda \sqrt{N} \phi_+ ( g) /2 } e^{  \lambda \kappa_s \mathcal{N}_+ /2} 
\end{align}
the conjugated Hamiltonian and thus 
\begin{align}
e^{ - \lambda \kappa_s \mathcal{N}_+ /2} & e^{\lambda \sqrt{N} \phi_+ ( g ) /2 } \frac{q_{\psi_{\mathcal{G}_N (s)}}}{\cG_N (s) - E_N(s)} e^{-\lambda \sqrt{N} \phi_+ ( g ) /2 } e^{ \lambda \kappa_s \mathcal{N}_+ /2} =   \frac{q_{\psi_{\widetilde{\mathcal{G}}_N (s)}}} {\widetilde{\cG}_N (s) - E_N(s)} \label{eq:red-resolv-1} 
\end{align}
where we introduced the notation 
\begin{align}
\widetilde{q}_{\psi_{\mathcal{G}_N (s)}} :=  e^{ - \lambda \kappa_s \mathcal{N}_+ /2}  e^{\lambda \sqrt{N} \phi_+ ( g ) /2 } q_{\psi_{\mathcal{G}_N (s)}}  e^{-\lambda \sqrt{N} \phi_+ ( g ) /2 } e^{ \lambda \kappa_s \mathcal{N}_+ /2} . 
\end{align}
We remark that despite that $\widetilde{q}_{\psi_{\mathcal{G}_N (s)}}$ is not a projection it still commutes with $\cG_N (s)$ and we have for any $\psi \in \mathcal{F}_{\perp \varphi}^{\leq N}$ the bound 
\begin{align}
\| \widetilde{q}_{\psi_{\mathcal{G}_N (s)}} \xi_N (s) \| &\leq C  \| \xi_N(s) \| + \|  e^{-\lambda \sqrt{N} \phi_+ ( g ) /2 } e^{ \lambda \kappa_s \mathcal{N}_+ /2}  \psi_{\mathcal{G}_N(s)} \|  \; \|  e^{\lambda \sqrt{N} \phi_+ ( g ) /2 } e^{- \lambda \kappa_s \mathcal{N}_+ /2}  \xi_N (s) \| \; \notag \\
& \leq C \| \xi_N(s) \| \; . 
\end{align}
Hence we find 
\begin{align}
\label{eq:estime-tildeG}
\Big\| \frac{\widetilde{q}_{\psi_{\mathcal{G}_N (s)}}} {\widetilde{\cG}_N (s) - E_N(s)}  \xi_N (s) \Big\| \leq C \Big\| \widetilde{q}_{\psi_{\mathcal{G}_N (s)}}\xi_N (s) \Big\| \leq C \| \xi_N (s) \| \; 
\end{align}
where the estimates are independent in $N$. Furthermore with the notation
\begin{align}
e^{ - \lambda \kappa_s \mathcal{N}_+ /2} e^{\lambda \sqrt{N} \phi_+ ( g ) /2 } \mathcal{R}_N e^{-\lambda \sqrt{N} \phi_+ ( g) /2 } e^{  \lambda  \mathcal{N}_+ /2} = \mathcal{R}_N +  A_{\mathcal{R}_N}    
\end{align}
we arrive for \eqref{def:M1} at 
\begin{align}
\cM (s) =&  \frac{\widetilde{q}_{\psi_{\mathcal{G}_N (s)}}} {\widetilde{\cG}_N (s) - E_N(s)}  \mathcal{R}_N + \frac{\widetilde{q}_{\psi_{\mathcal{G}_N (s)}}} {\widetilde{\cG}_N (s) - E_N(s)} A_{\mathcal{R}_N }  - \lambda  \dot{\kappa}_s \mathcal{N}_+ \; .   \label{eq:M2}
\end{align}
We will show in the following that the first two terms can be bounded by terms that are either $O( \lambda^3 N)$ (and thus subleading) or can be bounded in terms of operators that are compensated by the last term for properly chosen $\kappa_s$. To this end, we estimate the two first terms of the r.h.s. of \eqref{eq:M2} separately. While the second term is by definition already at least linear in $\lambda$, we need to use the reduced resolvent's properties for the first term. 

We start with the second term of the r.h.s. of \eqref{eq:M2} and consider for this  the single contributions of $\mathcal{R}_N$ separately, i.e. we define with \eqref{def:Rj} the sum $ A_{\mathcal{R}_N} =  \sum_{j=1}^2  A_{\mathcal{R}_N^{(j)}}$ where 
\begin{align}
\label{def:A-R} 
A_{\mathcal{R}_N^{(j)}}= e^{ - \lambda \kappa_s \mathcal{N}_+ /2} e^{\lambda \sqrt{N} \phi_+ ( g ) /2 } \mathcal{R}_N^{(j)} e^{-\lambda \sqrt{N} \phi_+ ( g) /2 } e^{  \lambda \kappa_s \mathcal{N}_+/2 } - \mathcal{R}_N^{(j)} \; . 
\end{align}
The first term $A_{\mathcal{R}_N^{(1)}}$ is by \cite[Proposition 2.2-2.4]{KRS} (resp. Lemma \ref{lemma:b-con}, \ref{prop:eN}) of the form 
\begin{align}
\label{eq:splitting-AR2}
A_{\mathcal{R}_N^{(1)}} = B_{\mathcal{R}_N^{(1)}} + D_{\mathcal{R}_N^{(1)}} + \mathcal{E}_{\mathcal{R}_N^{(1)}}
\end{align}
where $\| B_{\mathcal{R}_N^{(1)}} \| \leq C \| O \|^3 \lambda^3 N$ for all $\lambda \kappa_s \leq 1$ while the operator $D_{\mathcal{R}_N^{(1)}}$ quadratic in $\lambda$ is given by 
\begin{align}
D_{\mathcal{R}_N^{(1)}} =& \frac{\lambda^2}{4}  \sum_{p,q  \in \Lambda_+^*, p  \not= q} \widehat{v} (q)  \left[ \kappa_s b_{p+q}^* + \sqrt{N} g_{p+q} \Big( 1- \frac{\cN_+}{N} \Big) - \frac{1}{\sqrt{N}}  a^*_{p+q} a(g)  \right]  \left[ g_{-q} b_p - g_p b_{-q} \right]   \notag \\
&+ \frac{\lambda^2}{4\sqrt{N}}  \sum_{p,q  \in \Lambda_+^*, p  \not= q} \widehat{v} (q)  \left[ \| g \|^2_{\ell^2} b^*_{p+q} - g_{p+q} i \phi_- (g) + g_{p+q} b(g) \right] a^*_{-q}a_p \notag \\
&+ \frac{\lambda^2}{4\sqrt{N}}  \sum_{p,q  \in \Lambda_+^*, p  \not= q} \widehat{v} (q)   b_{p+q}^* \left[ \kappa_s \sqrt{N} \left( g_{-q} b_p + g_p b_{-q} \right)   -  N h_{-q} h_p \big ( 1 -\frac{\cN_+}{N} \Big) \right. \notag \\
& \hspace{5cm}  +  h_{-q} a^*(g) a_p + h_p a_{-q}^*    a(g)  \Big] + {\rm h.c.} 
\end{align}
and thus bounded for $\psi \in \mathcal{F}_{\perp \varphi_0}^{\leq N}$ and $\lambda \kappa_s \leq 1$ by 
\begin{align}
\label{eq:AR2-D}
\| D_{\mathcal{R}_N^{(1)}} \psi \| &\leq C \lambda^2 \| O \|^2 \sqrt{N} ( \kappa_s +1 ) \| ( \mathcal{N}_+ +1 )^{1/2} \psi \| \notag \\
&\leq C  \ O \|^3 \lambda^3 N( \kappa_s^2 + 1) \| \psi \|^2 +  C \lambda  \| O \| \| \cN_+^{1/2} \psi \|^2 \; . 
\end{align}
For the contribution of \eqref{eq:splitting-AR2} linear in $\lambda$ we find 
\begin{align}
\mathcal{E}_{\mathcal{R}_N^{(1)}} =& \frac{\lambda}{2\sqrt{N}}  \sum_{p,q  \in \Lambda_+^*, p  \not= q} \widehat{v} (q)  \left[ \kappa_s b_{p+q}^* + \sqrt{N} g_{p+q} \Big( 1- \frac{\cN_+}{N} \Big) - \frac{1}{\sqrt{N}}  a^*_{p+q} a(g)  \right]  a^*_{-q}a_p   \notag \\
&+ \frac{\lambda}{2}  \sum_{p,q  \in \Lambda_+^*, p  \not= q} \widehat{v} (q)  b_{p+q}^*  \left[ g_{-q} b_p - g_p b_{-q} \right] \notag \\
 &- {\rm h.c.} 
\end{align}
that is bounded for $\psi \in \mathcal{F}_{\perp \varphi_0}^{\leq N}$ and $\lambda\kappa_s \leq 1$ by 
\begin{align}
\label{eq:AR2-E}
\vert \langle \psi, \; \frac{\widetilde{q}_{\psi_{\mathcal{G}_N (s)}}} {\widetilde{\cG}_N (s) - E_N(s)} \mathcal{E}_{\mathcal{R}_N^{(1)}} \psi \rangle \vert  \leq C \lambda  \| O \| ( \kappa_s + 1 ) \| ( \mathcal{N}_+  +  1)^{1/2} \psi \|^2 \; .  
\end{align}
Thus from \eqref{eq:splitting-AR2}, \eqref{eq:AR2-D} and \eqref{eq:AR2-E} we get for all $\lambda\kappa_s \leq 1$
\begin{align}
\label{eq:AR2}
\vert \langle \xi_N (s),  & \; \frac{\widetilde{q}_{\psi_{\mathcal{G}_N (s)}}} {\widetilde{\cG}_N (s) - E_N(s)} A_{\mathcal{R}_N^{(1)}} \xi_N (s) \rangle  \vert \notag \\
&\leq  C \lambda \| O \| ( \kappa_s + 1 ) \langle \xi_N (s), ( \cN_+ + 1) \xi_N(s) \rangle + C N \lambda^3 \| \xi_N (s) \|^2 \; 
\end{align}
that is again of the desired form. For the second term of \eqref{def:Rj} we first observe that by the commutation relations we can write 
\begin{align}
\mathcal{R}_N^{(2)} = \frac{1 }{2N} \sum_{p,q,k \in \Lambda_+^*} \widehat{v} (k) a^*_{p-k}  a_q a^*_{q+k}a_p  
\end{align}
and thus, we find with \cite[Proposition 2.2-2.4]{KRS} (resp. Lemma \ref{lemma:b-con}, \ref{prop:eN}) that 
\begin{align}
\label{eq:splitting-AR3}
A_{\mathcal{R}_N^{(2)}} = B_{\mathcal{R}_N^{(2)}} + D_{\mathcal{R}_N^{(2)}} + \mathcal{E}_{\mathcal{R}_N^{(2)}}
\end{align}
where $\lambda\kappa_s \leq 1$ we have $\| B_{\mathcal{R}_N^{(3)}} \| \leq C\ |O \|^3 \lambda^3 N$ and 
\begin{align}
D_{\mathcal{R}_N^{(2)}} =& \frac{\lambda^2 N }{8N} \sum_{p,q,k \in \Lambda_+^*} \widehat{v} (k) \left[ g_{-p+k} a_{q} - g_q a^*_{p-k} \right] \left[ g_{-q-k} a_{p} - g_p a^*_{q+k}  \right] \notag \\
&+ \frac{\lambda^2}{8N}  \sum_{p,q,k \in \Lambda_+^*} \widehat{v} (k)  a^*_{p-k}a_q \Big[ \kappa_s \sqrt{N} \left[ g_{-q-k} a_{p} + g_p a^*_{q+k}  \right]- N  g_{-q-k} g_p \Big( 1 -\frac{\cN_+}{N} \Big)  \notag \\
&\hspace{6cm} +  a^*(g) g_{-q-k} a_{p} + a_{q+k}^* g_{p} a (g)  \Big] \notag \\
&+ \frac{\lambda^2}{8N}  \sum_{p,q,k \in \Lambda_+^*} \widehat{v} (k) \Big[ \kappa_s \sqrt{N} \left[ g_{-p+k} a_{q} + g_q a^*_{p-k}  \right]- N  g_{-p+k}  g_q \Big( 1 -\frac{\cN_+}{N} \Big)  \notag \\
&\hspace{6cm} +  a^*(g) g_{-p+k} a_{q} + a_{p-k}^* g_{q} a (g)  \Big] a^*_{q+k}a_p  \notag \\
&+ {\rm h.c.}
\end{align}
and thus bounded for $\psi \in \mathcal{F}_{\perp \varphi_0}^{\leq N}$ and $\lambda\kappa_s \leq 1$ by 
\begin{align}
\label{eq:AR3-D}
\| D_{\mathcal{R}_N^{(2)}} \psi \| \leq C \lambda^2 \sqrt{N} \| ( \mathcal{N}_+ +1 )^{1/2} \psi \| \leq C  \| O \|^3 \lambda^3 N \| \psi \|^2 + C \lambda \ O \| \| \cN_+^{1/2} \psi \|^2 \; 
\end{align}
For the linear contributions of \eqref{eq:splitting-AR3} we find 
\begin{align}
\mathcal{E}_{\mathcal{R}_N^{(2)}} =& \frac{\lambda \sqrt{N} }{4N} \sum_{p,q,k \in \Lambda_+^*} \widehat{v} (k)\left(  \left[ g_{-p+k} a_{q} - g_q a^*_{p-k} \right] a^*_{q+k}a_p + a^*_{p-k}a_q \left[ g_{-q-k} a_{p} - g_p a^*_{q+k}  \right]  \right) \notag \\
&- {\rm h.c.}
\end{align}
that is bounded for $\psi \in \mathcal{F}_{\perp \varphi_0}^{\leq N}$ and $\lambda\kappa_s \leq 1$ by 
\begin{align}
\label{eq:AR3-E}
\vert \langle \psi, \; \frac{\widetilde{q}_{\psi_{\mathcal{G}_N (s)}}} {\widetilde{\cG}_N (s) - E_N(s)} \mathcal{E}_{\mathcal{R}_N^{(3)}} \psi \rangle \vert  \leq C \lambda \ |O \| \| ( \mathcal{N}_+  +  1)^{1/2} \psi \|^2 \; .  
\end{align}
Thus from \eqref{eq:splitting-AR3}, \eqref{eq:AR3-D} and \eqref{eq:AR3-E} we get 
\begin{align}
\label{eq:AR3}
\vert \langle \xi_N (s),  & \; \frac{\widetilde{q}_{\psi_{\mathcal{G}_N (s)}}} {\widetilde{\cG}_N (s) - E_N(s)} A_{\mathcal{R}_N^{(2)}} \xi_N (s) \rangle  \vert \notag \\
&\leq  C \lambda \ O \| \langle \xi_N (s), ( \cN_+ + 1) \xi_N(s) \rangle + C N \| o \|^3 \lambda^3 \| \xi_N (s) \|^2 \; .  
\end{align}
With \eqref{eq:AR2} and \eqref{eq:AR3} we therefore conclude that the second term of the r.h.s. of \eqref{eq:M2} is bounded by 
\begin{align}
\vert \langle \xi_N (s),  & \; \frac{\widetilde{q}_{\psi_{\mathcal{G}_N (s)}}} {\widetilde{\cG}_N (s) - E_N(s)} A_{\mathcal{R}_N} \xi_N (s) \rangle  \vert \notag \\
&\leq  C \lambda \| O \| ( 1 + \kappa_s) \langle \xi_N (s), ( \cN_+ + 1) \xi_N(s) \rangle + C N \| O \|^3 \lambda^3 \| \xi_N (s) \|^2 \; \label{eq:M2-bound1}
\end{align}
for all $\lambda\kappa_s \leq 1$, that is of the desired form, namely the first term can be compensated (for properly chosen $\kappa_s$) by the last term of the r.h.s. of \eqref{eq:M2} while the second term is considered sub-leading here. 

It remains to show a similar bound for the first term of the r.h.s. of \eqref{eq:M2}. For this we use the resolvent's properties and first note that 
\begin{align}
\frac{\widetilde{q}_{\psi_{\mathcal{G}_N (s)}}} {\widetilde{\cG}_N (s) - E_N(s)} \mathcal{R}_N = \frac{\widetilde{q}_{\psi_{\mathcal{G}_N (s)}}} {\widetilde{\cG}_N (s) - E_N(s)} \big( p_{\psi_{\mathcal{G}_N (s)}} + q_{\psi_{\mathcal{G}_N (s)}} \big)  \mathcal{R}_N \label{eq:M22-1}
\end{align}
while the first term is with Lemma \ref{claim:gs} bounded by 
\begin{align}
\vert \langle \xi_N (s) , & \; \frac{\widetilde{q}_{\psi_{\mathcal{G}_N (s)}}} {\widetilde{\cG}_N (s) - E_N(s)} p_{\psi_{\mathcal{G}_N (s)}} \mathcal{R}_N \xi_N (s) \rangle \vert \notag \\
\leq& \Big\| \frac{\widetilde{q}_{\psi_{\mathcal{G}_N (s)}}} {\widetilde{\cG}_N (s) - E_N(s)}  \xi_N (s) \Big\| \|  p_{\psi_{\mathcal{G}_N (s)}} \mathcal{R}_N \xi_N (s) \|  \leq C N^{-1/2} \| \xi_N (s) \|^2 \label{eq:M22-}
\end{align}
we use the resolvent identity for the second term and write 
\begin{align}
\frac{\widetilde{q}_{\psi_{\mathcal{G}_N (s)}} q_{\psi_{\mathcal{G}_N (s)}}} {\widetilde{\cG}_N (s) - E_N(s)} \mathcal{R}_N =& \frac{q_{\psi_{\mathcal{G}_N (s)}}} {\cG_N (s) - E_N(s)} \mathcal{R}_N\notag \\
&+  \frac{\widetilde{q}_{\psi_{\mathcal{G}_N (s)}}} {\widetilde{\cG}_N (s) - E_N(s)}  \left(A_{\mathcal{Q}} + A_{\mathcal{R}_N} \right)  \frac{q_{\psi_{\mathcal{G}_N (s)}}} {\cG_N (s) - E_N(s)} \mathcal{R}_N\label{eq:M22-3}
\end{align}
with the notation \eqref{def:A-R} and 
\begin{align}
A_{\mathcal{Q}} = e^{ - \lambda \kappa_s \mathcal{N}_+ /2} e^{\lambda \sqrt{N} \phi_+ ( g ) /2 }  \mathcal{Q} e^{-\lambda \sqrt{N} \phi_+ ( g) /2 } e^{  \lambda \mathcal{N}_+ /2} - \mathcal{Q} \; . 
\end{align}
Thus summarizing \eqref{eq:M22-1}-\eqref{eq:M22-3}, it follows from Lemma \ref{claim:gs} that the first term of the r.h.s. of \eqref{eq:M2} is bounded by 
\begin{align}
\vert \langle \xi_N (s),& \; \frac{q_{\psi_{\widetilde{\mathcal{G}}_N (s)}}} {\widetilde{\cG}'_N (s) - E_N(s)} \mathcal{R}_N \xi_N (s) \rangle \vert \notag \\
&\leq C \| \xi_N (s) \|^2 + \vert \langle \xi_N (s), \; \frac{\widetilde{q}_{\psi_{\mathcal{G}_N (s)}}} {\widetilde{\cG}_N (s) - E_N(s)}  \left( A_{\mathcal{Q}} + A_{\mathcal{R}_N} \right)  \frac{q_{\psi_{\mathcal{G}_N (s)}}} {\cG_N (s) - E_N(s)} \mathcal{R}_N \xi_N (s) \rangle \vert  \; \label{eq:M22-weiter}
\end{align}
and we are left with bounding the last term. For this we are going to use that on the one hand the operators $A_{\mathcal{Q}'},A_{\mathcal{R}_N'}$ are at least linear in $\lambda$ and the resolvent w.r.t. to $\mathcal{G}_N'(s)$ allows to bound the number of particle operator. We proceed similarly as before and observe hat it follows from \cite[Proposition 2.2-2.4]{KRS} (resp. Lemma \ref{lemma:b-con}, \ref{prop:eN})  that for all $\lambda\kappa_s \leq 1$ we have from Lemma \ref{lemma:bb-conj} that 
\begin{align}
\label{eq:AQ1}
A_{\mathcal{Q}} = D_{\mathcal{Q}} +\mathcal{E}_{\mathcal{Q}} 
\end{align}
with $\| D_{\mathcal{Q}} \|_{\rm op} \leq N \lambda^2$ and since $\| p^2 g \|_{\ell^2} \leq C \vertiii{O} \| \varphi_0 \|_{H^2} \leq  C \vertiii{O}$ for any $\psi \in \mathcal{F}_{\perp \varphi_0}^{\leq N}$ that  
\begin{align}
\label{eq:QE-final}
\| \mathcal{E}_{\mathcal{Q}} \psi \| \leq C \sqrt{N} \lambda (  \| O \| + \kappa_s \vertiii{O}  ) \| ( \mathcal{N}_+ + 1)^{1/2} \psi \| \; . 
\end{align}
We recall that from \eqref{eq:AR2} and \eqref{eq:AR3},  we similarly have for all $\lambda\kappa_s \leq 1$
\begin{align}
A_{\mathcal{R}_N} = D_{\mathcal{R}_N} + \mathcal{E}_{\mathcal{R}_N} 
\end{align}
with $\| D_{\mathcal{R}_N} \| \leq C N \| O \|^2 \lambda^2$ and 
\begin{align}
\label{eq:RE-final}
\| \mathcal{E}_{\mathcal{R}_N} \psi \| \leq C \sqrt{N} \| O \| \lambda  ( 1+ \kappa_s) \| ( \mathcal{N}_+ + 1)^{1/2} \psi \|
\end{align}
for any $\psi \in \mathcal{F}_{\perp \varphi_0}^{\leq N}$. Thus from \eqref{eq:M22-weiter} we find 
\begin{align}
\vert \langle \xi_N (s),& \; \frac{\widetilde{q}_{\psi_{\mathcal{G}_N (s)}}} {\widetilde{\cG}_N (s) - E_N(s)} \mathcal{R}_N \xi_N (s) \rangle \vert \notag \\
&\leq C \| \xi_N (s) \|^2 + \lambda^2 N \Big\|  \frac{q_{\psi_{\mathcal{G}_N (s)}}} {\cG_N (s) - E_N(s)} \mathcal{R}_N \xi_N (s) \Big \|   \notag \\
&+  \vert \langle \xi_N (s), \; \frac{\widetilde{q}_{\psi_{\mathcal{G}_N (s)}} } {\widetilde{\cG}_N (s) - E_N(s)}  \left( \mathcal{E}_{\mathcal{Q}} + \mathcal{E}_{\mathcal{R}_N} \right)  \frac{q_{\psi_{\mathcal{G}_N (s)}}} {\cG_N (s) - E_N(s)} \mathcal{R}_N \xi_N (s) \rangle \vert  \; \label{eq:M22-weiter}
\end{align}
and we find with Lemma \ref{claim:gs} 
\begin{align}
\vert \langle \xi_N (s),& \; \frac{\widetilde{q}_{\psi_{\mathcal{G}_N (s)}}} {\widetilde{\cG}_N (s) - E_N(s)} \mathcal{R}_N \xi_N (s) \rangle \vert \notag \\
&\leq C \| \xi_N (s) \|^2 + \lambda^2 \| O \|^2 \sqrt{N} \| ( \mathcal{N}_+ + 1 ) \xi_N (s) \|^2 \notag \\
&+  \vert \langle \xi_N (s), \; \frac{\widetilde{q}_{\psi_{\mathcal{G}_N (s)}}} {\widetilde{\cG}_N (s) - E_N(s)}  \left( \mathcal{E}_{\mathcal{Q}} + \mathcal{E}_{\mathcal{R}_N} \right)  \frac{q_{\psi_{\mathcal{G}_N (s)}}} {\cG_N (s) - E_N(s)} \mathcal{R}_N \xi_N (s) \rangle \vert  \; \label{eq:M22-noch-weiter}
\end{align}
The first two summands of the r.h.s. of \eqref{eq:M22-noch-weiter} is already of the desired form as it can be bounded by terms $O( \lambda^3N)$ resp. terms of the form $\lambda \cN_+$. For the last term, we however have to estimate more carefully. We use once more the resolvent formula and write with the notation $\mathcal{E}_{\cG_N (s)} = \mathcal{E}_{\mathcal{Q}} + \mathcal{E}_{\mathcal{R}_N}$ and $D_{\cG_N (s)} = D_{\mathcal{Q}} + D_{\mathcal{R}_N}$
\begin{align}
& \frac{\widetilde{q}_{\psi_{\mathcal{G}_N (s)}} } {\widetilde{\cG}_N (s) - E_N(s)}  \mathcal{E}_{\cG_N (s)}  \frac{q_{\psi_{\mathcal{G}_N (s)}}} {\cG_N (s) - E_N(s)} \mathcal{R}_N \notag \\
&= \left( \frac{\widetilde{q}_{\psi_{\mathcal{G}_N (s)}} } {\widetilde{\cG}_N (s) - E_N(s)} ( \mathcal{E}_{\cG_N (s)} +  D_{\cG_N (s)}  )+ 1 \right)  \frac{q_{\psi_{\mathcal{G}_N (s)}}} {\cG_N (s) - E_N(s)}  \mathcal{E}_{\cG_N (s)}  \frac{q_{\psi_{\mathcal{G}_N (s)}}} {\cG_N (s) - E_N(s)} \mathcal{R}_N \; .  
\end{align}
It follows from \eqref{eq:QE-final}, \eqref{eq:RE-final} and Lemma \ref{claim:gs} that for all $\lambda\kappa_s \leq 1$ we have 
\begin{align}
\Big \|  & \mathcal{E}_{\cG_N (s)}  \frac{q_{\psi_{\mathcal{G}_N (s)}}} {\cG_N (s) - E_N(s)}  \mathcal{E}_{\cG_N (s)}  \frac{q_{\psi_{\mathcal{G}_N (s)}}} {\cG_N (s) - E_N(s)} \mathcal{R}_N \xi_N (s) \Big\| \notag \\
&\leq C \sqrt{N} \lambda  \vertiii{O} ( 1 + \kappa_s ) \Big\|   ( \mathcal{N}_+ + 1)^{1/2}\frac{q_{\psi_{\mathcal{G}_N (s)}}} {\cG_N (s) - E_N(s)}  \mathcal{E}_{\cG_N (s)}  \frac{q_{\psi_{\mathcal{G}_N (s)}}} {\cG_N' (s) - E_N(s)} \mathcal{R}_N\xi_N (s) \Big\| \notag \\
&\leq C \sqrt{N} \lambda  \vertiii{O} ( 1 + \kappa_s )  \Big\|  ( \mathcal{N}_+ + 1)^{-1/2} \mathcal{E}_{\cG_N (s)}  \frac{q_{\psi_{\mathcal{G}_N (s)}}} {\cG_N' (s) - E_N(s)} \mathcal{R}_N \xi_N (s) \Big\|  
\end{align}
With $\mathcal{N}_+ a_p = a_p (\mathcal{N}_+ - 1) $ it follows from \eqref{eq:AR2-E}, \eqref{eq:AR3-E} and  \eqref{eq:QE-final} that 
\begin{align}
\| ( \mathcal{N}_+ + 1)^{-1/2}  \mathcal{E}_{\cG_N (s)} \psi \| \leq C \sqrt{N} \lambda  \vertiii{O} ( 1 + \kappa_s ) \| \psi \| 
\end{align}
for any $\psi \in \mathcal{F}_{\perp \varphi_0}^{\leq N}$, and therefore 
\begin{align}
\Big\|  & \mathcal{E}_{\cG_N (s)}  \frac{q_{\psi_{\mathcal{G}_N (s)}}} {\cG_N (s) - E_N(s)}  \mathcal{E}_{\cG_N (s)}  \frac{q_{\psi_{\mathcal{G}_N (s)}}} {\cG_N (s) - E_N(s)} \mathcal{R}_N\xi_N (s) \Big\| \notag \\
&\leq C N \lambda^2   \vertiii{O}^2 ( 1 + \kappa_s )^2   \Big\|  \frac{q_{\psi_{\mathcal{G}_N (s)}}} {\cG_N (s) - E_N(s)} \mathcal{R}_N\xi_N (s) \Big\| \notag \\
&\leq   C \sqrt{N} \lambda^2   \vertiii{O}^2 ( 1 + \kappa_s )^2    \Big\| ( \mathcal{N}_+ + 1)^{1/2} \xi_N(s) \Big\|  \; . \label{eq:M22-E1}
\end{align}
where we used Lemma \ref{claim:gs} and \eqref{eq:R-norm}. Similarly we find with Lemma \ref{claim:gs}
\begin{align}
\Big\| &  D_{\cG_N (s)} \frac{q_{\psi_{\mathcal{G}_N (s)}}} {\cG_N (s) - E_N(s)}  \mathcal{E}_{\cG_N (s)}  \frac{q_{\psi_{\mathcal{G}_N (s)}}} {\cG_N (s) - E_N(s)} \mathcal{R}_N \xi_N(s) \Big\| \notag \\
&\leq C N^{3/2}\vertiii{O}^3 \lambda^3 (1+ \kappa_s)^3 \Big\| ( \mathcal{N}_+ +1)^{-1/2}  \frac{q_{\psi_{\mathcal{G}_N (s)}}} {\cG_N (s) - E_N(s)} \mathcal{R}_N \xi_N(s) \Big\| \notag \\
& \leq C  N \vertiii{O}^3 \lambda^3 (1+ \kappa_s)^3 \| \xi_N (s) \|^2 \; .  \label{eq:M22-E2}
\end{align}
Therefore we arrive with \eqref{eq:estime-tildeG},  \eqref{eq:M22-E1} and \eqref{eq:M22-E2} for \eqref{eq:M22-noch-weiter} for all $\lambda\kappa_s \leq 1$ at 
\begin{align}
\vert \langle & \xi_N (s), \; \frac{q_{\psi_{\widetilde{\mathcal{G}}_N (s)}}} {\widetilde{\cG}_N (s) - E_N(s)} \mathcal{R}_N \xi_N (s) \rangle \vert \notag \\&\leq C N \lambda^3  \vertiii{O}^3 (1+ \kappa_s)^3 \| \xi_N (s) \|^2 +  C \lambda \vertiii{O} \| ( \mathcal{N}_+ + 1 )^{1/2} \xi_N (s) \|^2 
\end{align}
that is of the desired form. Together with \eqref{eq:M2-bound1} we thus get for \eqref{eq:M2} that 
\begin{align}
2 \Re\langle & \xi_N (s) , \mathcal{M} (s) \xi_N(s) \rangle  \notag \\
&\geq  \left( C \lambda \vertiii{O} ( \kappa_s + 1) - \dot{\kappa}_s \right) \langle \xi_N (s),  \;  \cN_+ \xi_N(s) \rangle + C (N \lambda^3 + \lambda) \| \xi_N (s) \|^2 \; . 
\end{align}
With the choice
\begin{align}
\kappa_s = 2 \| O \| e^{C (s-1)} + \vertiii{O} ( e^{C (s-1)} - 1 )
\end{align}
yielding  $\partial_s \| \xi_N (s) \|^2 \geq -  C (\lambda + N \lambda^3) \| \xi_N (s) \|^2$ we arrive with Gronwall's inequality at the lower bound of Lemma \ref{lemma:step2}. 

The upper bound follows with similar ideas, replacing the lower with upper bounds and $\kappa_s$ with $-\kappa_s$. 
\end{proof}

\subsection{Step 3.} We recall that we are left with computing the expectation value w.r.t. to the ground state $\psi_{\mathcal{Q}}$ of the quadratic Hamiltonian $\mathcal{Q}$ given by \eqref{def:Q}. In this step we will show that the ground state is approximately given by $e^{B( \mu)} \Omega$, that is a  an generalized Bogoliubov transform applied to the vacuum vector. Furthermore, we show that $e^{B( \mu)}$ acts on the observable $\phi_+ (h)$ as a Bogoliubov transform, i.e. that $e^{-B( \mu)} \phi_+ (h) e^{B( \mu)}$ can be approximated by $\phi_+(f)$ with $f$ given by \eqref{def:f}. 
The main difficulty in this step here is that all quantities are formulated w.r.t. to modified creation and annihilation operators for which the action of the Bogoliubov transform is not explicitly given. However we use \eqref{def:d} and \eqref{eq:estimates-dp},\eqref{eq:estimates-dp*} and Lemmas \ref{lemma:dp}-\ref{lemma:conj-dd} to prove the following Lemma.

\begin{lemma} 
\label{lemma:step3}
Let $\kappa>0$. Under the same assumptions as in Theorem \ref{thm:main}, there exist constants $C_1,C_2 >0 $ and $\kappa_1, \kappa_2 >0$ such that for all $0 \leq \lambda \leq \min \lbrace 1/(\kappa_1 \vertiii{O}), 1/(\kappa_2 \vertiii{O}) \rbrace$ we have 
\begin{align}
& \langle \psi_{\mathcal{Q}}, \; e^{\lambda \sqrt{N} \phi_+ (g)/2} e^{  \kappa \vertiii{O} \mathcal{N}_+} e^{\lambda \sqrt{N} \phi_+ (g)/2}  \psi_{\mathcal{Q}}\rangle  \notag \\
&\quad \leq  e^{C_1 ( N \lambda^3 + \lambda )} \langle \Omega, \; e^{\lambda \sqrt{N} \phi_+ (f)/2} e^{  \kappa_1 \vertiii{O} \mathcal{N}_+ } e^{\lambda \sqrt{N} \phi_+ (f)/2}  \Omega \rangle  
\end{align}
resp. 
\begin{align}
&  \langle \psi_{\mathcal{Q}}, \; e^{\lambda \sqrt{N} \phi_+ (g)/2} e^{- \lambda \kappa \vertiii{O} \mathcal{N}_+} e^{\lambda \sqrt{N} \phi_+ (g)/2}  \psi_{\mathcal{Q}}\rangle  \notag \\
&\quad \geq  e^{-C_2 (N \lambda^3 + \lambda)}\langle \Omega, \; e^{\lambda \sqrt{N} \phi_+ (f)/2} e^{- \lambda \kappa_2 \vertiii{O} \mathcal{N}_+ } e^{\lambda \sqrt{N} \phi_+ (f)/2} \Omega \rangle  \label{eq:lemma3-lb}
\end{align}
where $f$ is given by \eqref{def:f}. 
\end{lemma}

\begin{proof} We start with the proof of the lower bound. The upper bound then follows with similar arguments as in the previous steps. As the generalized Bogoliubov transform \eqref{def:Bogo-b} is a unitary operator we write 
\begin{align}
\langle \psi_{\mathcal{Q}}, & \; e^{\lambda \sqrt{N} \phi_+ (g)/2} e^{ - \lambda\kappa \vertiii{O} \mathcal{N}_+} e^{\lambda \sqrt{N} \phi_+ (g)/2}  \psi_{\mathcal{Q}}\rangle  \notag \\
=&\langle e^{B( \mu)}\psi_{\mathcal{Q}}, \; e^{\lambda \sqrt{N} e^{B( \mu)}\phi_+ (g)/2e^{-B( \mu)}} e^{ - \kappa \vertiii{O} e^{B( \mu)}\mathcal{N}_+ e^{-B( \mu)}} e^{\lambda \sqrt{N} e^{B( \mu)} \phi_+ (g) e^{-B( \mu)}/2}  e^{B( \mu)}\psi_{\mathcal{Q}}\rangle  
\end{align}
The final goal is to compare the r.h.s. with the r.h.s. of the lower bound \eqref{eq:lemma3-lb}. To this end we perform three steps: we first show that we can replace the exponent $ -\kappa e^{B( \mu)}\mathcal{N}_+ e^{-B( \mu)}$ with $\kappa_2 \cN_+$ for sufficiently large $\kappa_2 >0 $. Second, we show that the exponent $e^{B( \mu)} \phi_+ (g) e^{-B( \mu)}$ can be effectively replaced by $\phi_+ (f)$ with $f$ given by \eqref{def:f}, again paying a sub-leading price. As a third and last step we then show that we can replace $e^{B( \mu)}\psi_{\mathcal{Q}}$ with an interpolation argument (similarly as in the proof of Lemma \ref{lemma:step2}) by $\Omega$. 

\subsubsection*{Step 3.1} Similarly as in the previous steps, we define for $s \in [0,1]$ the vector 
\begin{align}
\xi_1 (s) = e^{ -(1-s) \lambda \kappa_1 \vertiii{O} \cN_+ }e^{ - s \kappa \vertiii{O} e^{B( \mu)}\mathcal{N}_+ e^{-B( \mu)}} e^{\lambda \sqrt{N} e^{B( \mu)} \phi_+ (g) e^{-B( \mu)}/2}  e^{B( \mu)}\psi_{\mathcal{Q}}
\end{align}
where $\kappa_1 >0$ is chosen later. By definition we have 
\begin{align}
\|& \xi_1 (1) \|^2 \notag \\
&= \langle e^{B( \mu)}\psi_{\mathcal{Q}}, \; e^{\lambda \sqrt{N} e^{B( \mu)}\phi_+ (g)/2e^{-B( \mu)}} e^{ - \kappa \vertiii{O} e^{B( \mu)}\mathcal{N}_+ e^{-B( \mu)}} e^{\lambda \sqrt{N} e^{B( \mu)} \phi_+ (g) e^{-B( \mu)}/2}  e^{B( \mu)}\psi_{\mathcal{Q}}\rangle  
\end{align}
and 
\begin{align}
\| &\xi_1 (0) \|^2 \notag \\
& = \langle e^{B( \mu)}\psi_{\mathcal{Q}}, \; e^{\lambda \sqrt{N} e^{B( \mu)}\phi_+ (g)/2e^{-B( \mu)}} e^{ - \lambda \kappa_1  \vertiii{O} \mathcal{N}_+} e^{\lambda \sqrt{N} e^{B( \mu)} \phi_+ (g) e^{-B( \mu)}/2}  e^{B( \mu)}\psi_{\mathcal{Q}}\rangle  
\end{align}
and we control their difference by computing the derivative 
\begin{align}
\partial_s \| \xi_1 (s) \|^2 = 2 \Re \langle \xi_1(s), \, \partial_s \xi_1 (s) \rangle  = \langle \xi_1 (s), \mathcal{M}_1 (s) \xi_1 (s) \rangle
\end{align}
with
\begin{align}
\mathcal{M}_1(s) =  - \lambda \kappa \vertiii{O}  + \lambda e^{- \lambda (1-s)\kappa_1 \vertiii{O} \cN_+ }e^{- B( \mu)}\cN_+ e^{ B( \mu)}e^{ \lambda (1-s) \kappa_1 \vertiii{O} \cN_+ } + \lambda \kappa_1 \vertiii{O} \cN_+ 
\end{align}
From \eqref{def:d} we find 
\begin{align}
e^{- B( \mu)}\cN_+ e^{ B( \mu)} =& \sum_{p \in \Lambda_+^*} \left( \sigma_p^2 + ( \gamma_p^2 + \sigma_p^2)  b_p^* b_p ) + 2 \sigma_p \gamma_p (b_p^*b_{-p}^* + b_pb_{-p} ) \right)  \notag \\
+& \sum_{p \in \Lambda_+^*} \left( d_p^* ( \gamma_p b_p + \sigma_p b_{-p}^* ) + {\rm h.c.} + d_p^*d_p \right) \; .  \label{eq:tilde-N}
\end{align}
We recall that here we used the notation $\gamma_p := \cosh ( \mu_p)$ and $\sigma_p = \sinh ( \mu_p)$ with $\mu_p$ given by \eqref{def:mu}. Since $0\leq \widehat{v} \in \ell^1 ( \Lambda_+^*)$ we get $ \| \mu \|_{\ell^1 ( \Lambda_+^*)} \leq C$ and thus $\sigma \in \ell^1( \Lambda_+^*)$ and $\gamma \in \ell^\infty ( \Lambda_+^*)$. Hence for the first line of the r.h.s. of \eqref{eq:tilde-N}, we use again Propositions \ref{prop:eN},\ref{lemma:b-con} and get similarly as in the previous steps 
\begin{align}
\lambda  &\vert  \langle  \xi_1 (s), \;  e^{- \lambda (1-s)\kappa_1 \vertiii{O} \cN_+ }  \notag \\
& \hspace{2cm} \times \sum_{p \in \Lambda_+^*} \left( \sigma_p^2 + ( \gamma_p^2 + \sigma_p^2)  b_p^* b_p) + 2 \sigma_p \gamma_p (b_p^*b_{-p}^* + b_pb_{-p} ) \right)  e^{- \lambda (1-s)\kappa_1 \vertiii{O} \cN_+ }  \xi_1 (s) \rangle \vert  \notag \\
\leq&  ( C \lambda \| O \| + C_{\kappa_1} \| O \|^2\lambda^2) \langle \xi_1 (s) \;  (\cN_+ + 1) \xi_1 (s) \rangle  \notag \\
\leq&  \tilde{C}   \lambda \| O \| \langle \xi_1 (s) \;  (\cN_+ + 1) \xi_1 (s) \rangle + \tilde{C}_{\kappa_1 } N \lambda^3 \| O \|^3 \| \xi_N (s) \| \; . \label{eq:step31-1}
\end{align}
where $\tilde{C}, \tilde{C}_{\kappa_1 } >0$ and $C>0$ does not depend on $\kappa_1$.  For the second line of the r.h.s. we proceed similarly using Lemma \ref{lemma:eN-dp} and \eqref{eq:estimates-dp}, \eqref{eq:estimates-dp*} instead of Proposition \ref{prop:eN} and \eqref{eq:bounds-b}. In fact, Lemma \ref{lemma:eN-dp} and \eqref{eq:estimates-dp}, \eqref{eq:estimates-dp*} are applicable as by assumption the norm $\| \mu \|$ sufficiently small. Thus, we find 
\begin{align}
\lambda  &\vert  \langle  \xi_1 (s), \;  e^{- \lambda (1-s)\kappa_1 \vertiii{O} \cN_+ }  \notag \\
& \hspace{2cm} \times \sum_{p \in \Lambda_+^*} \left( d_p^* ( \gamma_p b_p + \sigma_p b_{-p}^* ) + {\rm h.c.} + d_p^*d_p \right) e^{- \lambda (1-s)\kappa_1 \vertiii{O} \cN_+ }  \xi_1 (s) \rangle \vert  \notag \\ 
&\leq  C  \lambda \| O \| \langle \xi_1 (s) \;  (\cN_+ + 1) \xi_1 (s) \rangle + C_{\kappa_1 }N \lambda^3 \| O \|^3 \| \xi_N (s) \| \; .  \label{eq:step31-2}
\end{align}
with $C,C_{\kappa_1} >0$ and $C>0$ independent of $\kappa_1$. 
Summarizing, we find from \eqref{eq:step31-1} and \eqref{eq:step31-2} by choosing $\kappa_1 >0$ sufficiently large 
\begin{align}
\partial_s \| \xi_1 (s) \|^2 &\geq ( \kappa_1 - C ) \lambda \vertiii{O} \langle \xi_1 (s), \; \mathcal{N}_+ \xi_1 (s) \rangle - C_{\kappa_1} (\lambda^3 N  + 1)\| \xi_1 (s) \|^2 \notag \\
&\geq  - C_{\kappa_1} (\lambda^3 N  + 1)\| \xi_1 (s) \|^2 \; . 
\end{align}
We finally arrive at 
\begin{align}
\langle \psi_{\mathcal{Q}}, & \; e^{\lambda \sqrt{N} \phi_+ (g)/2} e^{ - \lambda\kappa \vertiii{O} \mathcal{N}_+} e^{\lambda \sqrt{N} \phi_+ (g)/2}  \psi_{\mathcal{Q}}\rangle  \notag \\
\geq& e^{-C (N \lambda^3 +1)} \langle e^{B( \mu)}\psi_{\mathcal{Q}}, \; e^{\lambda \sqrt{N} e^{B( \mu)}\phi_+ (g)/2e^{-B( \mu)}} e^{ - \kappa_1 \vertiii{O} \mathcal{N}_+ } e^{\lambda \sqrt{N} e^{B( \mu)} \phi_+ (g) e^{-B( \mu)}/2}  e^{B( \mu)}\psi_{\mathcal{Q}}\rangle  
\end{align}

\subsubsection*{Step 3.2} Next we show that in the limit $N \rightarrow \infty$, we can replace the operator $  e^{B( \mu)}\phi_+ (g)e^{-B( \mu)}$ by $ \phi_+ (f)/2$ where $f$ is defined in \eqref{def:f}, i.e. that there exists $\kappa_2, C_2 >0$ such that 
\begin{align}
\langle  & e^{B( \mu)}\psi_{\mathcal{Q}}, \; e^{\lambda \sqrt{N} e^{B( \mu)}\phi_+ (g)/2e^{-B( \mu)}} e^{ - \kappa_1 \vertiii{O} \mathcal{N}_+ } e^{\lambda \sqrt{N} e^{B( \mu)} \phi_+ (g) e^{-B( \mu)}/2}  e^{B( \mu)}\psi_{\mathcal{Q}}\rangle  \notag \\
&\geq e^{-C( \lambda^3N +1)} \langle e^{B( \mu)}\psi_{\mathcal{Q}}, \; e^{\lambda \sqrt{N} \phi_+ (f)/2} e^{ - \kappa_2 \vertiii{O} \mathcal{N}_+ } e^{\lambda \sqrt{N} \phi_+ (f)/2}  e^{B( \mu)}\psi_{\mathcal{Q}}\rangle
\label{eq:estimate-step32}
\end{align}
To this end we define for $s \in [0,1]$ 
\begin{align}
\xi_2 (s) =  e^{ - \kappa(s) \vertiii{O} \mathcal{N}_+ }e^{(1-s) \lambda \sqrt{N}  \phi_+ (f)/2}  e^{s \lambda \sqrt{N} e^{B( \mu)} \phi_+ (g) e^{-B( \mu)}/2}  e^{B( \mu)}\psi_{\mathcal{Q}}
\end{align}
where $\kappa: [0,1] \rightarrow \mathbb{R}_+$ is a positive, differentiable function with $\kappa(1) = \kappa_1$ chosen later. Since 
\begin{align}
\| \xi_2 (1) \|^2 = \langle e^{B( \mu)}\psi_{\mathcal{Q}}, \; e^{\lambda \sqrt{N} e^{B( \mu)}\phi_+ (g)/2e^{-B( \mu)}} e^{ - \kappa_1 \vertiii{O} \mathcal{N}_+ } e^{\lambda \sqrt{N} e^{B( \mu)} \phi_+ (g) e^{-B( \mu)}/2}  e^{B( \mu)}\psi_{\mathcal{Q}}\rangle  
\end{align}
and 
\begin{align}
\| \xi_2 (0) \|^2 = \langle e^{B( \mu)}\psi_{\mathcal{Q}}, \; e^{\lambda \sqrt{N} \phi_+ (f)/2} e^{ - \kappa(0) \vertiii{O} \mathcal{N}_+ } e^{\lambda \sqrt{N} \phi_+ (f)/2}  e^{B( \mu)}\psi_{\mathcal{Q}}\rangle
\end{align}
it suffices to control the derivative 
\begin{align}
\partial_s \| \xi_2 (s) \|^2 = \langle \xi_2 (s), \; \mathcal{M}_2 (s) \xi_2 (s) \rangle 
\end{align}
with 
\begin{align}
\mathcal{M}_2 (s) =& -\lambda \dot{\kappa}_s \vertiii{O} \cN_+  \notag \\
&+ \lambda \sqrt{N}  e^{ - \kappa(s) \vertiii{O} \mathcal{N}_+ }e^{(1-s) \lambda \sqrt{N}  \phi_+ (f)/2} \left(  e^{B( \mu)} \phi(g) e^{B( \mu)} - \phi(f) \right) \notag \\
& \hspace{3cm} \times e^{(1-s) \lambda \sqrt{N}  \phi_+ (f)/2}e^{ - \kappa(s) \vertiii{O} \mathcal{N}_+ }
\end{align}
As before, the idea is to bound the first term w.r.t. the second term paying a price that is $O( \lambda^3 N)$. To that end we observe first that from \eqref{def:d} we have
\begin{align}
 e^{B( \mu)} \phi_+ (g) e^{B( \mu)} - \phi (f) = \sum_{p \in \Lambda_+^*} g_p \left[ d_p^* + d_{-p} \right] = d(g) + d^*(g) 
\end{align}
Since $\mu, g \in \ell^2( \Lambda_+^*)$  and thus $f \in \ell^2( \Lambda_+^*)$, it follows from \eqref{eq:estimates-dp} and Lemma \ref{lemma:dp} that 
\begin{align}
\lambda \sqrt{N} & \vert \langle \xi_2(s), \;   e^{ - \kappa(s) \vertiii{O} \mathcal{N}_+ }e^{(1-s) \lambda \sqrt{N}  \phi_+ (f)/2} \left( d(g) + d^*(g)  \right) e^{(1-s) \lambda \sqrt{N}  \phi_+ (f)/2}e^{ - \kappa(s) \vertiii{O} \mathcal{N}_+ } \xi_2(s) \rangle \vert \notag \\
\leq& C \lambda \vertiii{O} \| ( \mathcal{N}_+ + 1)^{1/2} \xi_2 (s) \|^2 + C_{\kappa_s} \vertiii{O}^2 \lambda^2 \| ( \mathcal{N}_+ + 1)^{1/2} \xi_2( s) \|^2 + C_{\kappa_s} N \lambda^3 \vertiii{O}^3 \| \xi_2 (s) \|^2 \; 
\end{align}
where $C>0$ is a positive constant that, in contrast to $C_{\kappa_s}$. does not depend on $\kappa_s$. Thus 
\begin{align}
\partial_s \| \xi_2(s) \|^2 \geq \lambda \vertiii{O} ( C - \dot{\kappa}_s ) \| \mathcal{N}_+^{1/2} \xi_2 (s) \|^2 - C_{\kappa_s} (N \lambda^3 \vertiii{O}^3+ \lambda \vertiii{O}) \| \xi_2 (s) \|^2 \; . 
\end{align}
We choose $\kappa_2 (s) = -\kappa_1 + (1-s) \kappa_2$ for sufficiently large $\kappa_2>0$ so that 
\begin{align}
\partial_s \| \xi_2(s) \|^2 \geq - C_{\kappa_s}(N \lambda^3\vertiii{O}^3 + \lambda\vertiii{O}) \| \xi_2 (s) \|^2 \; . 
\end{align}
and the desired estimate \eqref{eq:estimate-step32} follows. 

\subsubsection*{Step 3.3} Finally we prove that we can replace $e^{B( \mu)} \psi_{\mathcal{Q}}$ with the vacuum vector $\Omega$, that is the ground state of the diagonal Hamiltonian $\mathcal{D}$. More precisely we show that there exists $\kappa_3, C >0$ such that
\begin{align}
\langle e^{B( \mu)}\psi_{\mathcal{Q}}, & \; e^{\lambda \sqrt{N} \phi_+ (f)/2} e^{ - \lambda\kappa_2 \vertiii{O} \mathcal{N}_+} e^{\lambda \sqrt{N} \phi_+ (f)/2} e^{B( \mu)} \psi_{\mathcal{Q}}\rangle \notag \\
\geq & 
 e^{- C(N \lambda^3 + 1)} \langle \Omega ,  \; e^{\lambda \sqrt{N} \phi_+ (f)/2} e^{ - \lambda\kappa_3 \vertiii{O} \mathcal{N}_+} e^{\lambda \sqrt{N} \phi_+ (f)/2}  \Omega \rangle \; .
\end{align}
To this end we define for $s \in [0,1]$ the ground state $\psi_{\mathcal{Q}(s)}$ of the Hamiltonian
\begin{align}
\mathcal{Q}(s) = \mathcal{D} + s \mathcal{R}_{\mathcal{Q}} \; 
\end{align}
with corresponding eigenvalue $E(s)$ and furthermore the vector
\begin{align}
\xi_3 (s) = e^{ - \lambda\kappa(s) \vertiii{O} \mathcal{N}_+/2} e^{\lambda \sqrt{N} \phi_+ (f)/2} \psi_{\mathcal{Q} (s)} 
\end{align}
with differentiable $\kappa(s) : [0,1] \rightarrow \mathbb{R}_+$ satisfying $\kappa (1) = \kappa_2$ that we choose later. Since 
\begin{align}
\| \xi_3 (1) \|^2 = \langle e^{B( \mu)}\psi_{\mathcal{Q}}, & \; e^{\lambda \sqrt{N} \phi_+ (f)/2} e^{ - \lambda\kappa_2 \vertiii{O} \mathcal{N}_+} e^{\lambda \sqrt{N} \phi_+ (f)/2} e^{B( \mu)} \psi_{\mathcal{Q}}\rangle 
\end{align}
and 
\begin{align}
\| \xi_3 (0) \|^2 = \langle \Omega , & \; e^{\lambda \sqrt{N} \phi_+ (f)/2} e^{ - \lambda\kappa(0) \vertiii{O} \mathcal{N}_+} e^{\lambda \sqrt{N} \phi_+ (f)/2}  \Omega \rangle \; .
\end{align}
it suffices to control the derivative 
\begin{align}
\partial_s \| \xi_ 3 (s) \|^2 = 2 \Re \langle \xi_3 (s), \mathcal{M}_3 (s) \xi_3 (s) \rangle 
\end{align}
with 
\begin{align}
\label{def:M33}
\mathcal{M}_3 (s) =& e^{ - \lambda\kappa(s) \vertiii{O} \mathcal{N}_+/2} e^{\lambda \sqrt{N} \phi_+ (f)/2} \frac{q_{\psi_{\mathcal{Q}(s)}}}{ \mathcal{Q}(s) - E(s) }\mathcal{R}_{\mathcal{Q}} e^{\lambda \sqrt{N} \phi_+ (f)/2} e^{ \lambda\kappa(s) \vertiii{O} \mathcal{N}_+/2} - \dot{\kappa}_s \lambda \vertiii{O} \cN_+ \notag \\
=& \frac{\widetilde{q}_{\psi_{\mathcal{Q}(s)}}}{ \widetilde{\mathcal{Q}}(s) - E(s) }\widetilde{\mathcal{R}}_{\mathcal{Q}} - \dot{\kappa}_s \lambda \vertiii{O} \cN_+
\end{align}
where we introduced the notation
\begin{align}
 \widetilde{\mathcal{Q}}(s) = e^{ - \lambda\kappa(s) \vertiii{O} \mathcal{N}_+/2} e^{\lambda \sqrt{N} \phi_+ (f)/2}\mathcal{Q}(s)  e^{\lambda \sqrt{N} \phi_+ (f)/2} e^{ \lambda\kappa(s) \vertiii{O} \mathcal{N}_+/2}
\end{align}
and $\widetilde{q}_{\psi_{\mathcal{Q}(s)}} = e^{ - \lambda\kappa(s) \vertiii{O} \mathcal{N}_+/2} e^{\lambda \sqrt{N} \phi_+ (f)/2}q_{\psi_{\mathcal{Q}(s)}}   e^{\lambda \sqrt{N} \phi_+ (f)/2} e^{ \lambda\kappa(s) \vertiii{O} \mathcal{N}_+/2}$ resp. 
\begin{align}
 \widetilde{\mathcal{R}}_{\mathcal{Q}} = e^{ - \lambda\kappa(s) \vertiii{O} \mathcal{N}_+/2} e^{\lambda \sqrt{N} \phi_+ (f)/2} \mathcal{R}_{\mathcal{Q}}  e^{\lambda \sqrt{N} \phi_+ (f)/2} e^{ \lambda\kappa(s) \vertiii{O} \mathcal{N}_+/2} \; . 
\end{align}
First note that it follows from Lemmas \ref{lemma:bb-conj}, \ref{lemma:conj-dd} that for any formalized $\psi \in \mathcal{F}_{\perp \varphi}^{\leq N}$
\begin{align}
\langle  & \psi,   \left( \widetilde{\mathcal{Q}}(s)  - E (s) \right) \psi \rangle  \notag \\ 
&\geq \langle \psi, \; \left( \mathcal{Q} (s) - E (s) \right) \psi \rangle - C \lambda \sqrt{N} \vertiii{O} \| ( \mathcal{N}_+ + 1)^{1/2} \psi \| - CN\lambda^2\vertiii{O}^2 \; . 
\end{align}
Therefore there exists  $\varepsilon, C_\varepsilon >0$ such that by Proposition \ref{claim:Qs} 
\begin{align}
\langle \psi,   \left( \widetilde{\mathcal{Q}}(s)  - E (s) \right) \psi \rangle  \geq ( C-\varepsilon) \langle \psi, \; \mathcal{N}_+ \psi \rangle - C_\varepsilon N \vertiii{O}^2\lambda^2 \; . 
\end{align}
and consequently 
\begin{align}
\Big\| ( \mathcal{N}_+ +1)^{1/2} \frac{\widetilde{q}_{\psi_{\mathcal{Q}(s)}}}{ \widetilde{\mathcal{Q}}(s) - E(s) } \psi \| \leq C (N^{1/2} \lambda\vertiii{O} + 1) \| \widetilde{q}_{\psi_{\mathcal{Q}(s)}} \psi \| \; . 
 \label{eq:tilde-resolvent-1}
\end{align}
We can prove a similar bound not only for the square root but for the number of particle operator. For that we write with the resolvent identity 
\begin{align}
\mathcal{N}_+ \frac{\widetilde{q}_{\psi_{\mathcal{Q}(s)}}}{ \widetilde{\mathcal{Q}}(s) - E(s) } = \mathcal{N}_+ \frac{p_{\psi_{\mathcal{Q}(s)}} \widetilde{q}_{\psi_{\mathcal{Q}(s)}}}{ \widetilde{\mathcal{Q}}(s) - E(s) } + \mathcal{N}_+ \frac{q_{\psi_{\mathcal{Q}(s)}}\widetilde{q}_{\psi_{\mathcal{Q}(s)}}}{ \widetilde{\mathcal{Q}}(s) - E(s) } \; .
\end{align}
We use Proposition \ref{claim:Qs} for the first and the resolvent identity for the second term and arrive at 
\begin{align}
\Big\| \mathcal{N}_+ \frac{\widetilde{q}_{\psi_{\mathcal{Q}(s)}}}{ \widetilde{\mathcal{Q}}(s) - E(s) } \psi \Big\| \leq C + \Big\|  \mathcal{N}_+ \frac{q_{\psi_{\mathcal{Q}(s)}}}{\mathcal{Q}(s) - E(s) } \left( \widetilde{\mathcal{Q}}(s)  - \mathcal{Q} (s) \right) \frac{\widetilde{q}_{\psi_{\mathcal{Q}(s)}}}{ \widetilde{\mathcal{Q}}(s) - E(s) } \psi \Big\| \; . 
\end{align}
From Proposition \ref{claim:Qs} we find 
\begin{align}
\Big\| \mathcal{N}_+ \frac{\widetilde{q}_{\psi_{\mathcal{Q}(s)}}}{ \widetilde{\mathcal{Q}}(s) - E(s) } \psi \Big\| \leq C \| \widetilde{q}_{\psi_{\mathcal{Q}(s)}} \psi \| + C \Big\|\left( \widetilde{\mathcal{Q}}(s)  - \mathcal{Q} (s) \right) \frac{\widetilde{q}_{\psi_{\mathcal{Q}(s)}}}{ \widetilde{\mathcal{Q}}(s) - E(s) } \psi \Big\| \; 
\end{align}
and furthermore since $\widehat{v} \in \ell^1 ( \Lambda_+^*)$ from Lemma \ref{lemma:bb-conj}, \ref{lemma:conj-dd} 
\begin{align}
\Big\| &\mathcal{N}_+ \frac{\widetilde{q}_{\psi_{\mathcal{Q}(s)}}}{ \widetilde{\mathcal{Q}}(s) - E(s) } \psi \Big\| \notag \\
&\leq C ( 1 + N \lambda^2\vertiii{O}^2) \| \widetilde{q}_{\psi_{\mathcal{Q}(s)}} \psi \| + C \lambda \sqrt{N}\vertiii{O} \Big\| ( \mathcal{N}_+ + 1)^{1/2} \frac{\widetilde{q}_{\psi_{\mathcal{Q}(s)}}}{ \widetilde{\mathcal{Q}}(s) - E(s) } \psi \Big\| . 
\end{align}
With the first bound \eqref{eq:tilde-resolvent-1} we thus arrive at 
\begin{align}
\label{eq:tilde-resolvent-2}
\Big\| \mathcal{N}_+ \frac{\widetilde{q}_{\psi_{\mathcal{Q}(s)}}}{ \widetilde{\mathcal{Q}}(s) - E(s) } \psi \Big\| \leq C (N \lambda^2\vertiii{O}^2 + 1) \| \widetilde{q}_{\psi_{\mathcal{Q}(s)}} \psi \| \; . 
\end{align}
With these estimates \eqref{eq:tilde-resolvent-1}, \eqref{eq:tilde-resolvent-2} we can now bound $\mathcal{M}_3 (s)$ defined in \eqref{def:M33}. We define the operator
\begin{align}
\mathcal{A}_{\mathcal{R}_{\mathcal{Q}}}   =  \widetilde{\mathcal{R}}_{\mathcal{Q}} -  \mathcal{R}_{\mathcal{Q}}  
\end{align} 
that we can bound with Lemmas \ref{lemma:conj-dd}, \ref{lemma:RK} by 
\begin{align}
\| \mathcal{A}_{\mathcal{R}_{\mathcal{Q}}} \frac{\widetilde{q}_{\psi_{\mathcal{Q}(s)}}}{ \widetilde{\mathcal{Q}}(s) - E(s) } \xi_3 (s) \| \leq C N\lambda^3\vertiii{O}^3 \|\widetilde{q}_{\psi_{\mathcal{Q}(s)}} \xi_3 (s) \| + C \lambda\vertiii{O} \Big\| ( \mathcal{N}_+ + 1) \frac{\widetilde{q}_{\psi_{\mathcal{Q}(s)}}}{ \widetilde{\mathcal{Q}}(s) - E(s) } \xi_3 (s) \Big\|  
\end{align}
and thus with \eqref{eq:tilde-resolvent-2}
\begin{align}
\| \mathcal{A}_{\mathcal{R}_{\mathcal{Q}}} \;\frac{\widetilde{q}_{\psi_{\mathcal{Q}(s)}}}{ \widetilde{\mathcal{Q}}(s) - E(s) } \xi_3 (s) \| \leq C (N\lambda^3\vertiii{O}^3 +1) \| \widetilde{q}_{\psi_{\mathcal{Q}(s)}} \xi_3 (s) \| \leq C (N \lambda^3\vertiii{O}^3 +1) \|\xi_3 (s) \| \; . 
\end{align}
In order to control the contribution of $\mathcal{R}_{\mathcal{Q}}$ in $\mathcal{M}_3 (s)$ in \eqref{def:M33} we use Lemmas \ref{lemma:dp}, \ref{lemma:RK} that show 
\begin{align}
\vert \langle \xi_3 (s), \; & \mathcal{R}_{\mathcal{Q}} \frac{\widetilde{q}_{\psi_{\mathcal{Q}(s)}}}{ \widetilde{\mathcal{Q}}(s) - E(s) } \xi_3 (s) \rangle \vert \notag \\
&\leq CN^{-1/2} \| ( \mathcal{N}_+ +1)^{1/2} \xi_3 (s) \|  \Big\| ( \mathcal{N}_+ + 1) \frac{\widetilde{q}_{\psi_{\mathcal{Q}(s)}}}{ \widetilde{\mathcal{Q}}(s) - E(s) } \xi_3 (s) \Big\|  \notag \\
&\leq  C_{\kappa(s)}\sqrt{N} \lambda^2 \vertiii{O}^2 \| ( \mathcal{N}_+ +1)^{1/2} \xi_3 (s) \|   \| \xi_3 (s) \| \notag \\
& \leq C \lambda \vertiii{O}\| ( \mathcal{N}_+ +1)^{1/2} \xi_3 (s) \|^2 +\widetilde{C}_{\kappa(s)} (N\lambda^3 \vertiii{O}^3\ + 1) |\xi_3 (s) \|^3   \; . 
\end{align}
Hence, we find from \eqref{def:M33} choosing $\kappa(s) = \kappa_3 s + \kappa_4 (1-s)$ and $\kappa_4$ sufficiently large 
\begin{align}
\partial_s \| \xi_3 (s) \|^2 &\geq ( C\vertiii{O} - \kappa(s) \vertiii{O}) \| \mathcal{N}_+^{1/2} \xi_3(s) \|^2 - CN\lambda^3 \| \xi_3 (s) \|^2\notag \\
&\geq - C_{\kappa(s)} (N\lambda^3 \vertiii{O}^3 + 1) \| \xi_3 (s) \|^2 \; 
\end{align}
and thus the desired estimate follows. 
\end{proof}

\subsection{Step 4.} In the last step we compute the remaining expectation value in the vacuum. The following Lemma follows immediately from \cite[Lemma 3.3]{RSe}.

\begin{lemma} 
\label{lemma:step4}
Let $\kappa>0$. Under the same assumptions as in Theorem \ref{thm:ldp}, there exist constants $C_1, C_2 >0$ such that for all $0 \leq \lambda \leq 1/( \kappa \vertiii{O})$ we have 
\begin{align}
\ln  \langle \Omega, \; e^{\lambda \sqrt{N} \phi_+ (f)/2} e^{ \kappa \lambda  \vertiii{O} \mathcal{N}_+} e^{\lambda \sqrt{N} \phi_+ (f)/2}  \Omega \rangle  \leq N \left( \frac{ \lambda^2 \| f \|^2}{2} + C_1 \lambda^3 \vertiii{O}^3 \right) + C_1 \lambda \vertiii{O}
\end{align}
resp. 
\begin{align}
\ln  \langle \Omega, \; e^{\lambda \sqrt{N} \phi_+ (f)/2} e^{ - \kappa \lambda \vertiii{O} \mathcal{N}_+} e^{- \lambda \sqrt{N} \phi_+ (f)/2}  \Omega \rangle  \geq N \left( \frac{ \lambda^2 \| f \|^2}{2} - C_2 \lambda^3 \vertiii{O}^3 \right) - C_2\vertiii{O} \lambda 
\end{align}
\end{lemma}

\subsection*{Acknowledgments} S.R. would like to thank Phan Th\`anh Nam and Robert Seiringer for fruitful discussions, helpful suggestions and continuous support during this project. 


\end{document}